\documentclass{article} 

\PassOptionsToPackage{numbers, compress}{natbib}

\usepackage[letterpaper,top=2cm,bottom=2cm,left=3cm,right=3cm,marginparwidth=1.75cm]{geometry}

\title{Discretely Beyond $1/e$: Guided Combinatorial Algorithms for Submodular Maximization}
\usepackage{times}
\date{}
\usepackage[utf8]{inputenc} 
\usepackage[T1]{fontenc}    
\usepackage{hyperref}       
\usepackage{url}            
\usepackage{booktabs}       
\usepackage{amsfonts}       
\usepackage{nicefrac}       
\usepackage{microtype}      
\pdfoutput=1

\usepackage{natbib}
\usepackage{amsthm}
\usepackage{amsmath}
\usepackage{thmtools}
\usepackage[ruled,linesnumbered]{algorithm2e}
\usepackage{verbatim}
\usepackage{graphicx}
\usepackage[space]{grffile}
\usepackage{tikz}
\usepackage{subfigure}
\usepackage{bbold}
\usetikzlibrary{arrows}

\graphicspath{{figures/}}
\usepackage{mathrsfs}
\usepackage{amssymb}
\usepackage{booktabs}
\usepackage{changepage}
\usepackage{xspace}
\usepackage{hyperref}
\allowdisplaybreaks[4]
\usepackage{multirow}
\usepackage{thm-restate}
\usepackage{wrapfig}

\usepackage{tikz}
\usetikzlibrary{arrows}

\usepackage{xcolor}
\hypersetup{
    colorlinks,
    linkcolor={red!80!black},
    citecolor={blue!60!black},
    urlcolor={blue!80!black}
}

\SetKwProg{Proc}{Procedure}{\string:}{}
\DontPrintSemicolon

\declaretheorem[style=definition,numberwithin=section]{theorem}
\declaretheorem[style=definition,sibling=theorem]{lemma}

\declaretheorem[style=definition,numberwithin=section]{definition}

\declaretheorem[style=definition,numberwithin=section]{corollary}

\newcommand{\ie}{\textit{i.e.}\xspace}
\newcommand{\eg}{\textit{e.g.}\xspace}
\newcommand{\st}{\textit{s.t.}\xspace}

\newcommand{\ex}[1]{ \mathbb{E} \left[ #1 \right] }

\newcommand{\exc}[2]{ \mathbb{E} \left[ #1 \, | \; #2 \right] }

\newcommand{\epsi}[0]{ \varepsilon }
\newcommand{\reals}[0]{ \mathbb{R}^+ }
\newcommand{\func}[2]{ #1 \left( #2 \right) }
\newcommand{\ff}[1]{ f \left( #1 \right) }
\newcommand{\brk}[1]{\left( #1 \right)}
\newcommand{\sbrk}[1]{\left[ #1 \right]}

\newcommand{\marge}[2]{\Delta{\left( #1 |  #2 \right) }}
\DeclareMathOperator*{\argmax}{arg\,max}

\newcommand{\uni}{\mathcal U}
\newcommand{\mtr}{\mathcal M}
\newcommand{\iset}{\mathcal I}
\newcommand{\note}[1]{\textcolor{black}{#1}}

\usepackage{pifont}
\newcommand{\oh}[1]{\mathcal{O}\left( #1 \right)}
\newcommand{\ohs}[2]{\mathcal{O}_{ #1 }\left( #2 \right)}

\newcommand{\sg}{\textsc{StandardGreedy}\xspace}

\newcommand{\sm}{\textsc{SM}\xspace}

\newcommand{\opt}{\text{OPT}\xspace}
\newcommand{\add}{\textsc{ADD}\xspace}
\newcommand{\nmon}{\textsc{SMCC}\xspace}

\newcommand{\fls}{\textsc{FastLS}\xspace}

\newcommand{\grg}{\textsc{GuidedRG}\xspace}

\newcommand{\gig}{\textsc{GuidedIG-S}\xspace}
\newcommand{\gigm}{\textsc{GuidedIG-M}\xspace}

\newcommand{\ig}{\textsc{InterlaceGreedy}\xspace}
\newcommand{\itpg}{\textsc{InterpolatedGreedy}\xspace}
\newcommand{\tgig}{\textsc{ThreshGuidedIG}\xspace}

\newcommand{\rg}{\textsc{RandomGreedy}\xspace}

\newcommand{\prune}{\textsc{Prune}\xspace}
\newcommand{\lc}{\textsc{Linear\-Card}\xspace}
\newcommand{\tg}{\textsc{TwinGreedy}\xspace}

\author{
  Yixin Chen, Ankur Nath, Chunli Peng, Alan Kuhnle \\
  Department of Computer Science \& Engineering \\
  Texas A\&M University \\
  Colloge Station, TX\\
  \texttt{\{chen777, anath, chunli.peng, kuhnle\}@tamu.edu} \\
}

\begin{document}

\maketitle

\begin{abstract}%
  For constrained, not necessarily monotone submodular maximization, all known approximation algorithms with ratio greater than $1/e$ require continuous ideas, such as queries to the multilinear extension of a submodular function and its gradient, which are typically expensive to simulate with the original set function. For combinatorial algorithms, the best known approximation ratios for both size and matroid constraint are obtained by a simple randomized greedy algorithm of \citet{buchbinder2014submodular}: $1/e \approx 0.367$ for size constraint and $0.281$ for the matroid constraint in $\mathcal O (kn)$ queries, where $k$ is the rank of the matroid. In this work, we develop the first combinatorial algorithms to break the $1/e$ barrier: we obtain approximation ratio of $0.385$ in $\mathcal O (kn)$ queries to the submodular set function for size constraint, and $0.305$ for a general matroid constraint. These are achieved by guiding the randomized greedy algorithm with a fast local search algorithm. Further, we develop deterministic versions of these algorithms, maintaining the same ratio and asymptotic time complexity. Finally, we develop a deterministic, nearly linear time algorithm with ratio $0.377$.
\end{abstract}


\section{Introduction}
A nonnegative set function $f: 2^\uni \to \reals$ is \textit{submodular}
iff for all $S\subseteq T\subseteq \uni$, $x \in \uni \setminus T$, 
$\ff{S \cup \{x\}}-\ff{S} \ge \ff{T \cup \{x\}}-\ff{T}$;
and $f$ is monotone iff $\ff{S} \le \ff{T}$ for all $S\subseteq T\subseteq \uni$.
Submodular optimization plays an important role in data science and machine learning~\citep{DBLP:journals/corr/abs-2202-00132},
particularly in tasks that involve selecting a representative subset of data or features.
Its diminishing returns property makes it ideal for scenarios where the
incremental benefit of adding an element to a set decreases as the set grows.
Applications include sensor placement for environmental monitoring~\citep{DBLP:journals/jmlr/KrauseSG08,DBLP:conf/fusion/PowersBKA16}, where the goal is to maximize coverage with limited sensors, feature selection~\citep{DBLP:conf/icassp/LiuWKSB13,DBLP:conf/aistats/KhannaEDNG17,DBLP:conf/kdd/BaoHZ22} in machine learning to improve model performance and reduce overfitting, and data summarization~\citep{DBLP:conf/aaai/MirzasoleimanJ018,DBLP:conf/nips/TschiatschekIWB14} for creating concise and informative summaries of large datasets. Further, many of these applications employ
submodular objective functions that are non-monotone, \eg{} \citet{DBLP:conf/aaai/MirzasoleimanJ018,DBLP:conf/nips/TschiatschekIWB14}. 
Formally, we study the  optimization problem (\sm):
$\max\ff{S}, \st \, S \in \iset,$
where $f$ is nonnegative, submodular and not necessarily monotone;
and $\iset \subseteq 2^{\uni}$ is a family of feasible subsets.
Specifically, we consider two cases: when $\iset$ is a size constraint (all
sets of size at most $k$); and more generally, when $\iset$ is an arbitrary
matroid of rank $k$. 

In this field, algorithms typically assume access to a \textit{value oracle} for the
submodular function $f$, and the efficiency of an algorithm is measured by the
number of queries to the oracle,
because evaluation of the submodular function is typically expensive
and dominates other parts of the computation.
In the general, not necessarily monotone case, the approximability of constrained
submodular optimization in the value oracle model is not well understood.
For several years, $1/e \approx 0.367$ was conjectured to be the best ratio,
as this ratio is obtained by the measured continuous greedy
\citep{DBLP:conf/focs/FeldmanNS11} 
algorithm
that also gets the $1 - 1/e$ ratio in the monotone setting, which is known to
be optimal \citep{DBLP:journals/mor/NemhauserW78}. 
However, in several landmark works, the $1/e$ barrier was broken:
first to $0.371$ by \citet{buchbinder2014submodular}
(for size constraint only)
and subsequently to $0.372$
by \citet{DBLP:conf/focs/EneN16}, then $0.385$ by \citet{buchbinder2019constrained}.
Very recently, the best known approximation factor has been improved to
0.401 \citep{buchbinder2023constrained}.
On the other hand, the best hardness result
is $0.478$ \citep{DBLP:conf/soda/GharanV11,DBLP:conf/isaac/Qi22}. 

All of the algorithms improving on the $1/e$ ratio
use oracle queries to the \textit{multilinear extension} of a submodular function
and its gradient. The multilinear extension relaxes the submodular set function
to allow choosing an element with probability in $[0,1]$. Although this is a powerful
technique,
the multilinear extension must be approximated by
polynomially many random samples of the original
set function oracle.
Unfortunately, this leads to a high query complexity for these algorithms,
which we term \textit{continuous algorithms}; typically, the query complexity
to the original submodular function is left uncomputed.
As an illustration, we compute
in Appendix~\ref{apx:0.385-time}
that the continuous algorithm of \citet{buchbinder2019constrained}
achieves ratio of $0.385$ with query complexity of $\oh{n^{11}\log(n)}$ to the set
function oracle. Consequently, these algorithms are of mostly
theoretical interest --
the time cost of running on tiny instances (say, $n < 100$)
is already prohibitive, as demonstrated by \citet{DBLP:journals/jair/ChenK24} where
a continuous algorithm required more than $10^9$ queries to the set
function on an instance with $n=87, k=10$.


\begin{table}[t]
  \centering \small
      \caption{The prior state-of-the-art and the ratios achieved in this paper,
        in each category: deterministic (det), randomized combinatorial (cmb), and continuous (cts).}
    \begin{tabular}{ccccc}
        \toprule
        Constraint & Reference & Query & Ratio & Type \\
      \midrule
               \multirow{6}{*}{Size}    & \citet{buchbinder2018deterministic}  & $\oh{k^3n}$ & $ 1/e \approx 0.367$ & Det  \\ 
                   & \citet{buchbinder2014submodular} & $\oh{kn}$ & $ 1/e$ & Cmb \\
                   & \citet{buchbinder2023constrained}  & poly$(n)$ & $0.401$ & Cts \\ 
        \cmidrule(rl){2-5}
        &Algorithm~\ref{alg:irg} & $\oh{kn/\epsi}$ & $0.385-\epsi$ & Cmb \\ 
        & Algorithm~\ref{alg:determ2}  & $\oh{kn\brk{\frac{10}{9\epsi}}^{\frac{20}{9\epsi}-1}}$ & $0.385-\epsi$ & Det \\ 
        & Algorithm~\ref{alg:dgig-2}  & 
        $\oh{\log(k)n\brk{\frac{10}{3\epsi}}^{\frac{20}{3\epsi}}\brk{\frac{5}{\epsi}}^{\frac{10}{\epsi}-1}}$ 
        & $0.377-\epsi$ & Det \\ 
        \midrule
      \multirow{5}{*}{Matroid} & \citet{sun2022improved}  & $\oh{k^2n^2}$ & $0.283-\oh{\frac{1}{k^2}}$ & Det \\ 
                   & {\citet{buchbinder2014submodular}}  & $\oh{kn}$ & $0.283-\epsi$ & Cmb \\ 
                   & \citet{buchbinder2023constrained}  & poly$(n)$ & $0.401$ & Cts \\ 
                   
        \cmidrule(rl){2-5}
        &Algorithm~\ref{alg:irg}  & $\oh{kn/\epsi}$ & $0.305-\epsi$ & Cmb \\ 
        & Algorithm~\ref{alg:determ2}  & $\oh{kn\brk{\frac{10}{9\epsi}}^{\frac{20}{9\epsi}-1}}$ & $0.305-\epsi$ & Det \\ 
        \bottomrule
    \end{tabular}
    \label{tab:random}
\end{table}

For size and matroid constraints, the
current state-of-the-art approximation ratio for a combinatorial algorithm (\ie not continuous)
is obtained by the \rg algorithm (Algorithm \ref{alg:rg}) of \citet{buchbinder2014submodular}.
\rg achieves ratio $1/e \approx 0.367$ for size constraint,
\begin{wrapfigure}{l}{0.45\textwidth} \vspace{-0.1cm}
\begin{minipage}{0.45\textwidth}
\begin{algorithm}[H]\caption{\citet{buchbinder2014submodular}}\label{alg:rg}
    \Proc{\rg$(f, k)$}{
     \textbf{Input:} oracle $f$, size constraint $k$\;
     \textbf{Initialize:} $A_0\gets \emptyset$\;
     \For{$i \gets 1$ to $ k$}{
        $M_i\gets \argmax_{S\subseteq \uni, |S| = k}\sum_{x\in S}\marge{x}{A_{i-1}}$\;
        $x_i \gets$ a uniformly random element from $M_i$\;
        $A_i\gets A_{i-1}+x_i$\;
     }
     \textbf{return} $A_k$ \;}
 \end{algorithm}
\end{minipage} \vspace{-1cm}
\end{wrapfigure}
and $0.283-\epsi$ ratio for matroid constraint; its
query complexity is $\oh{kn}$.
Thus, there is no known combinatorial algorithm that
breaks the $1/e$ barrier; and therefore, no such algorithm
is available to be used in practice on any of the applications
of \sm{} described above. 

Moreover, closing the gap between
ratios achieved by deterministic and randomized algorithms for \sm has been the
focus of a number of recent works \citep{buchbinder2018deterministic,han2020deterministic,chen2023approximation,DBLP:journals/siamcomp/BuchbinderFG23}.
In addition to theoretical interest, deterministic algorithms
are desirable in practice, as a ratio that holds in expectation
may fail on any given run with constant probability. 
\citet{buchbinder2018deterministic} introduced a
linear programming method to derandomize the
\rg algorithm (at the expense of additional time complexity),
meaning that the best known ratios for deterministic algorithms are again given by \rg. 
There is no known method to derandomize continuous algorithms, 
as the only known way to approximate the multilinear extension \note{of a general submodular set function} relies on random sampling methods.
\note{\citet{DBLP:conf/sdm/OzcanMI21}, however, introduced a deterministic estimation via Taylor series approximation,
but this approach is limited to a specific class of submodular functions that can be 
expressed as weighted compositions of analytic and multilinear functions.}
Therefore, there is no known deterministic algorithm that
breaks the $1/e$ barrier.
The best known ratio in each category of continuous, combinatorial, and deterministic algorithms
is summarized in Table \ref{tab:random}.
In this work, we consider the following questions:

\begin{center} 
  \textit{Can combinatorial algorithms, and separately, deterministic algorithms,  obtain approximation ratios for \sm beyond $1/e$? If so, are the resulting algorithms practical and do they yield empirical improvements in objective value over existing algorithms?}
\end{center}
\subsection{Contributions} \label{sec:contribution}
In this work, we improve the best known ratio
for a combinatorial algorithm for size-constrained
\sm to $0.385 - \epsi \approx 1/e + 0.018$. This is achieved by using the result of
a novel local search algorithm to guide the \rg{} algorithm.
Overall, we obtain query complexity of $\oh{kn / \epsi}$, which is at worst
quadratic in the size of the ground set, since $k \le n$. 
Thus, this algorithm is practical and can run on moderate
instance sizes; the first algorithm with ratio beyond $1/e$
for which this is possible.
Further, we extend this algorithm to the matroid constraint,
where it improves the best known ratio of a combinatorial algorithm
for a general matroid constraint from $0.283$ of \rg{} to $0.305 - \epsi$.

Secondly, we obtain these same approximation ratios with deterministic algorithms. 
The ideas are similar to the randomized case, except we leverage a recently
formulated algorithm \itpg{}
\citep{chen2023approximation}
as a replacement for guided \rg{}. The analysis of
\itpg{} has similar recurrences (up to low order terms)
and the algorithm can be guided in a similar fashion to \rg{},
but is amenable to derandomization.
The derandomization only adds a constant factor, albeit one that is
exponential in $(1/\epsi)$. 

Next, we seek to lower the query complexity further, while still
improving the $1/e$ ratio.
As \itpg{} can be sped up to $\mathcal O_\epsi ( n \log k )$,
the bottleneck becomes the local search procedure. Thus, we develop a faster
way to produce the guiding set $Z$ by exploiting a run of (unguided) \itpg{}
and demonstrating that a decent guiding set is produced if the algorithm exhibits
nearly worst-case behavior. With this method, we achieve a deterministic algorithm
with ratio $0.377 \approx 1/e + 0.01$ in $\mathcal O_\epsi ( n \log k )$ query complexity,
which is nearly linear in the size of the ground set (since $k = O(n)$).

Finally, we demonstrate the practical utility of our combinatorial $0.385$-approximation algorithm 
by implementing it and evaluating in the context of two applications of size-constrained
\sm{} on moderate instance sizes (up to $n = 10^4$). We evaluate it with parameters
set to enforce a ratio $> 1/e$.
It outperforms both the standard greedy algorithm
and \rg{} by a significant margin in terms of objective
value; moreover, it uses about twice the queries of \rg and is
orders of magnitude faster than existing local search algorithms. 

\subsection{Additional Related Work}
\textbf{Derandomization.}
\citet{buchbinder2018deterministic} introduced a
linear programming (LP) method to derandomize the
\rg algorithm, thereby obtaining ratio $1/e$ with
a deterministic algorithm. 
Further,
\citet{sun2022improved} were able to apply this technique
to \rg for matroids.
A disadvantage of this approach is an increase in
the query complexity over the original randomized algorithm.
Moreover, we attempted to use this method
to derandomize our guided \rg algorithm, but were unsuccessful.
Instead, we obtained our deterministic algorithms by
guiding the \itpg algorithm instead of \rg; this algorithm
is easier to derandomize, notably without increasing the
asymptotic query complexity. 

\textbf{Relationship to \citet{buchbinder2019constrained}.}
The continuous, $0.385$-approximation algorithm of
\citet{buchbinder2019constrained}
guides the measured continuous greedy algorithm
using the output of a continuous local search algorithm,
in analogous fashion to how we guide \rg{} with the output
of a combinatorial local search. However, the analysis of \rg{} is much different from
that of measured continuous greedy, although the resulting
approximation factor is the same. Specifically,
\citet{buchbinder2019constrained} obtain their ratio
by optimizing a linear program mixing the continous local search
and guided measured continous greedy; in contrast, we
use submodularity and the output of our fast local search
to formulate new recurrences for guided \rg{}, which we then solve. 

\textbf{Local search algorithms.} 
Local search is a technique widely used in combinatorial optimization.
\citet{DBLP:journals/mp/NemhauserWF78}
introduced a local search algorithm
for monotone functions under size constraint; they
showed a ratio of $1/2$, but noted that their algorithm
may run in exponential time. 
Subsequently, local search has been found to be useful, especially for non-monotone functions.
\citet{DBLP:journals/siamcomp/FeigeMV11} proposed
a $1/3$ approximation algorithm with $\oh{n^4/\epsi}$ queries
for the unconstrained submodular
maximization problem utilizing local search.
Meanwhile, \citet{DBLP:conf/stoc/LeeMNS09}
proposed a local search algorithm for general \sm with matroid constraint,
attaining $1/4-\epsi$ approximation ratio with a query complexity of $\oh{k^5\log(k)n/\epsi}$.
We propose our own \fls in Section~\ref{sec:fls},
yielding a ratio of $1/2$ for monotone cases
and $1/4$ for non-monotone cases through repeated applications of \fls,
while running in $\oh{kn/\epsi}$ queries.

\textbf{Fast approximation algorithms.}
\citet{buchbinder2017comparing} developed a faster version
of \rg for size constraint that reduces the query complexity
to $\ohs{\epsi}{n}$ with ratio of $1/e - \epsi$. 
\citet{chen2023approximation} proposed \lc, the first deterministic, linear-time
algorithm with an $1/11.657$-approximation ratio for size constraints.
Also, \citet{han2020deterministic} introduced \tg, a $0.25$-approximation
algorithm with a query complexity of $\oh{kn}$ for matroid constraints.
These algorithms are fast enough to be used as building blocks for our \fls,
which requires as an input a constant-factor approximation in $\oh{kn}$ queries.

\note{
\textbf{Relationship to \citet{tukan2024practical}.}
During the submission of this paper, we noticed an independent and parallel work by 
\citet{tukan2024practical}, which proposed a different $0.385$-approximation algorithm.
Both papers start from a similar idea-guiding the random greedy algorithm 
with a fast algorithm to find a local optimum.
However, \citet{tukan2024practical} only considered size constraint and focused on algorithm speedup.
They introduced a randomized local search algorithm and used its output to guide the stochastic greedy of \citet{buchbinder2014submodular}, 
achieving a query complexity of $\ohs{\epsi}{n+k^2}$.
On the other hand, we 1) address a more general constraint-matroid constraint;
2) for size constraint, present an asymptotically faster algorithm 
that uses a novel way of guiding with partial solutions from random greedy itself, 
which are not local optima, thereby achieving ratio $0.377-\epsi$ with $\ohs{\epsi}{n\log (k)}$ queries;
and 3) derandomize these algorithms.}

\subsection{Preliminaries}\label{sec:prelim}
\textbf{Notation.}
In this section, we establish the notations employed throughout the paper.
We denote the marginal gain of adding $A$ to $B$
by $\marge{A}{B} = \ff{A\cup B}- \ff{B}$.
For every set $S \subseteq \uni$ and an element $x \in \uni$,
we denote $S\cup \{x\}$ by $S+x$, and $S \backslash \{x\}$ by $S-x$.
Given a constraint and its related feasible sets $\iset$,
let $O \in \argmax_{S\in \iset} \ff{S}$; that is, $O$ is
an optimal solution. 
To simplify the pseudocode and the analysis,
we add $k$ \textit{dummy elements} into the ground set, 
where the dummy element serves 
as a null element with zero marginal gain when added to any set.
The symbol $e_0$ is utilized to represent a dummy element.

\textbf{Submodularity.}
A set function $f:2^{\uni} \to \reals$ is submodular,
if $\marge{x}{S} \ge \marge{x}{T}$ for all $S\subseteq T \subseteq \uni$
and $x\in \uni \setminus T$,
or equivalently, for all $ A, B \subseteq \uni$,
it holds that $\ff{A}+\ff{B} \ge \ff{A\cup B} + \ff{A\cap B}$.

\textbf{Constraints.}
In this paper, our focus lies on two constraints: 
size constraint and matroid constraint.
For size constraint, we define the feasible subsets as 
$\iset(k) = \{S\subseteq \uni: |S| \le k\}$, where $k$
is an input parameter. 
The matroid constraint is defined in Appendix \ref{apx:prelim}.

\textbf{Organization.}
Our randomized algorithms are described in Section~\ref{sec:0.385},
with two subroutines, \fls and \grg, in Section~\ref{sec:fls} and~\ref{sec:grg}, respectively.
Due to space constraints, we provide only a sketch of the analysis for
size constraint in the main text. The full pseudocodes and formal proofs for
both size and matroid constraint are provided in Appendix \ref{apx:irg}.
Then, we briefly sketch the deterministic approximation algorithms 
in Section~\ref{sec:determ}, with full details provided in Appendix \ref{apx:determ-full}.
Next, we introduce the nearly linear-time deterministic algorithm
in Section~\ref{sec:0.377}, with omitted analysis provided in Appendix \ref{apx:0.377}.
Our empirical evaluation is summarized in Section~\ref{sec:exp}.
In Section \ref{sec:limitations}, we discuss limitations and future directions.


\section{A Randomized $(0.385 - \epsi)$-approximation in $\oh{kn/\epsi}$ Queries} \label{sec:0.385} 
\begin{algorithm}[t]\label{alg:irg}
   \caption{Randomized combinatorial approximation algorithm.}
     \textbf{Input:} Instance $(f, \iset)$, a constant-factor approximation $Z_0$, switch time $t \in [0,1]$, accuracy $\epsi > 0$\;
     $Z \gets \fls(f, \iset, Z_0, \epsi)$
     \tcc*[r]{find local optimum $Z$}
     $A \gets \grg(f, \iset, Z, t)$
     \tcc*[r]{guided by local optimum $Z$}
     \textbf{return} $\argmax\{\ff{Z}, \ff{A}\}$\;
\end{algorithm}
In this section, we present our randomized approximation algorithm (Alg.~\ref{alg:irg}) for both
size and matroid constraints.
This algorithm improves the state-of-the-art, combinatorial approximation ratio
to $0.385-\epsi \approx 1/e + 0.018$ for size constraint,
and to $0.305-\epsi \approx 0.283 + 0.022$
for matroid constraint.

\textbf{Algorithm Overview.}
\note{In overview, Alg. \ref{alg:irg} consists of two components, 
which are detailed below. 
The first component is a local search algorithm, \fls (Alg.~\ref{alg:fls} in Appendix~\ref{apx:fls}), 
described in detail in Section \ref{sec:fls}. 
In brief,}
the local search algorithm
takes an accuracy parameter $\epsi > 0$ and
a constant-factor, approximate solution $Z_0$ as input, which may be
produced by any approximation algorithm with better than $\oh{kn}$ query complexity.
\note{The second component is a random greedy algorithm,
\grg (Alg.~\ref{alg:grg} in Appendix~\ref{apx:grg}), that is guided by the output $Z$ of the local search, described in detail in Section~\ref{sec:grg}.}
Also, \grg takes a parameter $t \in [0,1]$, which is the switching time (as fraction
of the budget or rank $k$)
from guided to unguided behavior. The candidate with best $f$ value from the two subroutines is returned.

If $f(Z) < \alpha \opt$ (otherwise, there is nothing to show),
then the local search
set satisfies our definition of
$((1+\epsi)\alpha,\alpha)$-guidance set (Def.~\ref{def:guide} below).
Under this guidance, we show that \grg
produces a superior solution compared to its unguided counterpart.
The two components, \fls and \grg are described in Sections~\ref{sec:fls}
and~\ref{sec:grg}, respectively. The following theorem is proven in Section \ref{sec:grg} (size constraint)
and Appendix \ref{apx:mc} (matroid constraint). 
\begin{theorem}\label{thm:irg}
Let $(f,\iset)$ be an instance of \sm. 
Let $\epsi > 0$, and $k \ge 1/\epsi$.
Algorithm~\ref{alg:irg} achieves an expected $(0.385-\epsi)$-approximation ratio for size constraint with $t=0.372$, and an expected $(0.305-\epsi)$-approximation ratio for matroid constraint with $t=0.559$.
The query complexity of the algorithm is $\oh{kn/\epsi}$.
\end{theorem}

\subsection{The Fast Local Search Algorithm}\label{sec:fls}
In this section, we introduce \fls (Alg. \ref{alg:fls}),
which is the same for size or matroid constraints. 
There are several innovations in \fls that
result in $\oh{ kn / \epsi }$ time complexity, where $k$ is the maximum size
of a feasible set, and $\epsi > 0$ is an input accuracy parameter.

In overview, the algorithm maintains
a feasible set $Z$; initially, $Z = Z_0$,
where $Z_0$ is an input set which is a constant
approximation to \opt.
The value of $Z$ is iteratively improved via
swapping, which is done in the following way.
For each element $a \in \uni$,
we compute $\marge{a}{Z \setminus a}$; if $a \not \in Z$,
this is just the gain of $a$ to $Z$; this requires
$\oh{n}$ queries.
Then, if $a \in Z$ and $e \not \in Z$ such that
$Z \setminus a + e$ is feasible, and
$\marge{e}{Z} - \marge{a}{Z \setminus a} \ge \frac{\epsi}{k}f(Z)$,
then $a$ is swapped in favor of $e$. If no such swap exists,
the algorithm terminates. 

One can show that, for each swap, the value of $Z$ increases
by at least a multiplicative $(1 + \epsi / k)$ factor.
Since $f(Z)$ is initialized to a constant fraction of \opt,
it follows that we make at most $\mathcal O( k / \epsi )$
swaps. Since each swap requires $\oh{n}$ queries,
this yields the query complexity of the algorithm:
$\oh{ kn / \epsi}$.
In addition, if $f$ is monotone, \fls{}
gets ratio of nearly $1/2$
for \fls. A second repetition of \fls yields a ratio
of $1/4$ in the case of general (non-monotone) $f$,
as shown in Appendix~\ref{apx:fls-ratio}.
Thus, \fls may be of independent interest, as local search
obtains good objective values empirically and is commonly used in applications.

For our purposes, we want to use the output of \fls to guide
\rg. Since we will also use another algorithm for a similar purpose
in Section \ref{sec:0.377},
we abstract the properties needed for such a guidance set.
Intuitively, a set $Z$ is a good guidance set if it has
a low $f$-value and also ensures that the value of its intersection
and union with an optimal solution are poor. 
\begin{definition}\label{def:guide}
A set $Z$ is a ($\alpha,\beta$)-guidance set, if given constants 
$\alpha,\beta\in (0,0.5)$
and optimum solution $O$, it holds that:
1) $\ff{Z} < \alpha \ff{O}$; 2) $\ff{O\cap Z} \le \alpha \ff{O}$; 3) $\ff{O\cup Z}\le \beta\ff{O}$,
or alternatively, $3'$) $\ff{O\cap Z}+\ff{O\cup Z}\le(\alpha + \beta)\ff{O}$.
\end{definition}
Lemma~\ref{lemma:fls}
(proved in Appendix~\ref{apx:fls-proof})
implies that for
the \fls output $Z$, if $f(Z) < \alpha \opt$,
then $Z$ is $((1+\epsi)\alpha, \alpha)$-guidance set. 
\begin{restatable}{lemma}{lemmafls}
\label{lemma:fls}
Let $\epsi > 0$, and let $(f,\iset(\mtr))$ be an instance of \sm.
The input set $Z_0$ is an $\alpha_0$-approximate solution to $(f, \iset(\mtr))$.
\fls (Alg.~\ref{alg:fls}) returns a solution $Z$ with $\oh{kn\log(1/\alpha_0)/\epsi}$ queries such that 
$\ff{S \cup Z} + \ff{S \cap Z} < (2+\epsi)\ff{Z}$,
where $S \in \iset(\mtr)$.
\end{restatable}

\subsection{Guiding the \rg Algorithm}\label{sec:grg}
\begin{figure}
    \centering
    \includegraphics[width=0.93\linewidth]{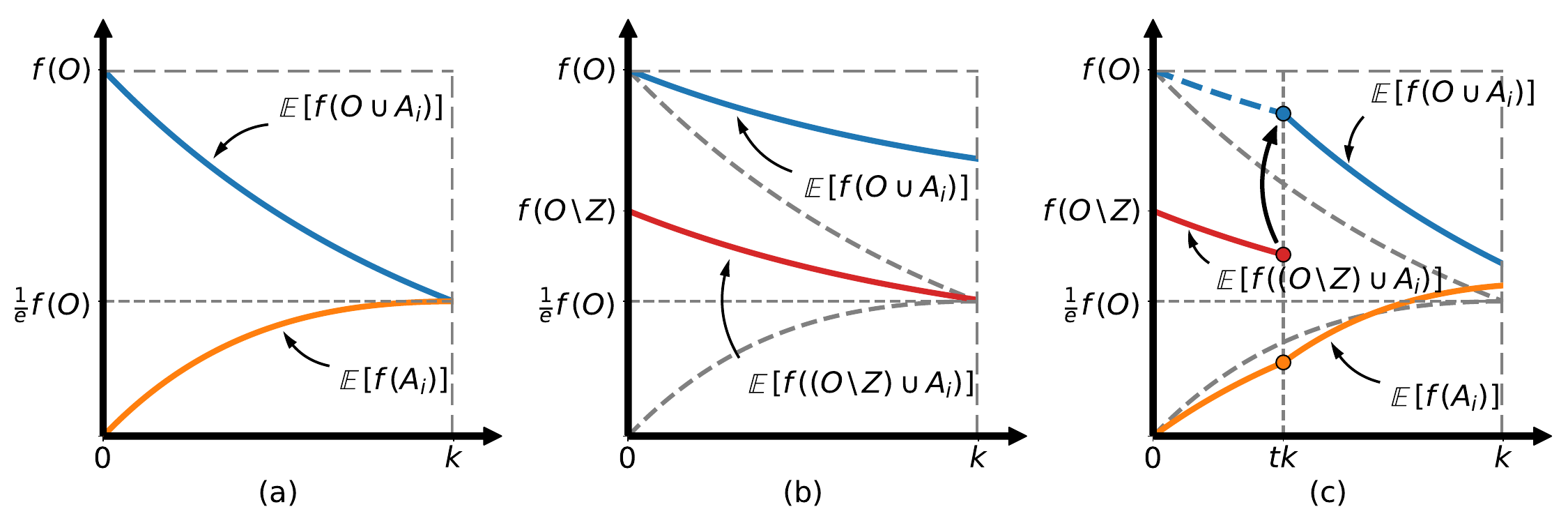}
    \caption{\textbf{(a)}: The evolution of $\ex{ \ff{O\cup A_i}}$ and $\ex{\ff{A_i}}$ in the
    worst case of the analysis of \rg, as the partial solution size increases to $k$.
      \textbf{(b)}: Illustration of how the degradation of $\ex{ \ff{O\cup A_i} }$ changes as we introduce
    an $(0.385+\epsi, 0.385)$-guidance set. 
    \textbf{(c)}: The updated degradation with a switch point $tk$, where the algorithm
    starts with guidance and then switches to running without guidance.
    The dashed curved lines depict the unguided values from \textbf{(a)}.
    }
    \label{fig:grg-deg}
\end{figure}

In this section, we discuss the guided \rg algorithm
(Alg.~\ref{alg:grg})
using an $((1+\epsi)\alpha, \alpha)$-guidance set $Z$ returned by \fls.
Due to space constraints, we only consider the size constraint in the main text.
The ideas for the matroid constraint are similar, although the final recurrences
obtained differ. The version for matroid constraints is
given in Appendix~\ref{apx:mc}.

The algorithm \grg{} is simple to describe: it maintains
a partial solution $A$, initially empty.
It takes as parameters the switching time $t$ and guidance set $Z$.
While the partial solution satisfies $|A| < tk$, the algorithm operates as
\rg{} with ground set $\uni \setminus Z$;
after $|A| \ge tk$, it operates as \rg{} with ground set $\uni$. 
Pseudocode is provided in Appendix~\ref{apx:grg}.

\textbf{Overview of analysis.}
For clarity, we first describe the (unguided) \rg{} analysis
from \citet{buchbinder2014submodular}.
There are two recurrences:
the first is the greedy gain:
$$\ex{\ff{A_i} -\ff{A_{i-1}}} \ge \frac{1}{k}\ex{\ff{O\cup A_{i-1}}-\ff{A_{i-1}}}.$$
Intuitively, the gain at iteration $i$ is at least a $1/k$ fraction
of the difference between $f(O \cup A)$ and $A$, in expectation, where $A$ is the partial solution.
If $f$ were monotone, the right hand side would be at least $(\opt - f(A))/k$.
However, in the case that $f$ is not monotone, the set $O \cup A$ may have value smaller than $\opt$.

To handle this case, it can be shown
that the expected value of $f(O \cup A)$ satisfies a second
recurrence:
{\small
  $$\ex{\ff{O\cup A_i}} \overset{(a)}{\ge} \left(1-\frac{1}{k}\right)\ex{\ff{O\cup A_{i-1}}}+\frac{1}{k}\ex{\ff{O\cup A_{i-1}\cup M_i}} \overset{(b)}{\ge} \left(1-\frac{1}{k}\right)\ex{\ff{O\cup A_{i-1}}},$$}%

where $M_i$ is the set of elements with the top $k$ marginal gains at iteration $i$,
\textit{(a)} is from submodularity, and \textit{(b)} is from nonnegativity. 
Thus, this expected value, while initially \opt (since $A_0 = \emptyset$),
may degrade but is bounded.

Both of these recurrences are solved
together to prove the expected ratio of $1/e$ for \rg{}:
the worst-case evolution of the expected values of $f(A_i)$, $f(O \cup A_i)$, according
to this analysis, is illustrated in Fig. \ref{fig:grg-deg}(a). Observe that $f(A_i)$
converges to $\opt / e$ (as required for the ratio), and \textit{observe that $f(O \cup A_i)$
  also converges to $\opt / e$}. Thus, very little gain is obtained in the later
stages of the algorithm, as illustrated in the plot. 
The overarching idea of the guided version of the algorithm
is to obtain a better degradation
of $\ex{\ff{O\cup A_i}}$, leading to better gains later in the
algorithm that improve the worst-case ratio.
In the following, we elaborate on this goal,
the  achievement of which is illustrated in Fig. \ref{fig:grg-deg}(c). 

\textbf{Stage 1: Recurrences when avoiding  $Z$.}
Suppose $Z$ is an $(\alpha, \beta)$-guidance set,
and that \rg{} selects elements as before, but excluding $Z$
from the ground set. 
Then, the recurrences change as follows. The recurrence
for the gain becomes:
\begin{equation}
  \ex{\ff{A_i} -\ff{A_{i-1}}} \ge \frac{1}{k}\ex{\ff{(O \setminus Z) \cup A_{i-1}}-\ff{A_{i-1}}},
\end{equation}
where $O\setminus Z$ replaces $O$
since we select elements outside of the set $Z$. For the second recurrence,
we can lower bound the term $\ex{\ff{O\cup A_{i-1} \cup M_i}}$
using submodularity and 
the fact that $Z \cap A_{i - 1} = \emptyset$:
\begin{equation}
  \ex{\ff{O\cup A_i}} \overset{}{\ge} \left(1-\frac{1}{k}\right)\ex{\ff{O\cup A_{i-1}}}+\frac{1}{k}\ex{\ff{O} - \ff{O \cup Z}}.
\end{equation}
Finally, a similar
recurrence to (2) also holds for $\ff{(O \setminus Z) \cup A_i}$; both
are needed for the analysis. Since $Z$ is a guidance set,
by submodularity, $f(O \setminus Z)\ge f(O) - f(O \cap Z) \ge (1 - \alpha) \opt,$
which ensures that some gain is available by selection outside of $Z$.
And $f(O) - f(O \cup Z) \ge (1 - \beta )\opt$, which means that the degradation
recurrences are improved. 

The blue line in Figure~\ref{fig:grg-deg}(b) depicts this improved degradation
with the size of the partial solution.
However, this improvement comes at a cost:
a smaller increase in $\ex{\ff{A_i}}$ is obtained
over the unguided version. Therefore,
to obtain an improved ratio we switch
back to the regular behavior of \rg{} -- intuitively,
this shifts the relatively good, earlier behavior of \rg{} to later
in the algorithm. 

\textbf{Stage 2: Switching back to selection from whole ground set.}
After the switch, the recurrences revert back to the original ones, but with
different starting values. 
Since $\ex{\ff{O\cup A_i}}$ was significantly enhanced in the first stage,
in the final analysis we get an overall improvement over the unguided version.
The blue line in Figure~\ref{fig:grg-deg}(c) demonstrates
the degradation of $\ex{\ff{O\cup A_i}}$ over two stages,
while the orange line depicts how the approximation ratio converges
to a value $0.385 > 1/e$. 

The above analysis sketch can be formalized and the resulting recurrences
solved: the results are stated in the following lemma,
which is formally proven in Appendix \ref{apx:grg-size}.
\begin{restatable}{lemma}{lemmagrgsize}
\label{lemma:grg3-size}
With an input size constraint $\iset$ and 
a $((1+\epsi)\alpha, \alpha)$-guidance set $Z$,
\grg returns set $A_k$ with $\oh{kn}$ queries, \st
$\ex{\ff{A_k}}\ge \left[\left(2-t-\frac{1}{k}\right)\left(1-\frac{1}{k}\right)e^{t-1}-e^{-1}- (1+\epsi)\alpha\left(\left(1-\frac{1}{k}\right)^2e^{t-1}-e^{-1}\right) 
\right.$ $\left. -\alpha\left(\left(1+\frac{1-t}{1-\frac{1}{k}}\right)e^{t-1}-\left(2-\frac{1}{k}\right)e^{-1}\right)\right]\ff{O}.$
\end{restatable}
From Lemma \ref{lemma:grg3-size}, 
we can directly prove the main result for size constraint. 
\begin{proof}[Proof of Theorem \ref{thm:irg} under size constraint]
  Let $(f,\iset)$ be an instance of \sm, with optimal solution set $O$.
If $\ff{Z}\ge (0.385-\epsi)\ff{O}$ under size constraint,
the approximation ratio holds immediately.
Otherwise, by Lemma~\ref{lemma:fls}, 
\fls returns a set $Z$ which is an $((1+\epsi)\alpha, \alpha)$-guidance set,
where $\alpha = 0.385-\epsi$.
By Lemma~\ref{lemma:grg3-size}, 
\begin{align*}
&\ex{\ff{A_k}}
\ge \left[\brk{2-t-\epsi}(1-\epsi)e^{t-1}-e^{-1}-(0.385-0.615\epsi)\brk{(1-\epsi)^2e^{t-1}-e^{-1}}\right.\\
&\left.-(0.385-\epsi)\brk{\left(1+\frac{1-t}{1-\epsi}\right) e^{t-1}-(2-\epsi)e^{-1}}\right]\ff{O} \tag{$\forall k \ge \frac{1}{\epsi}$}\\
&\ge (0.385-\epsi)\ff{O}. \tag{$t=0.372$}
\end{align*}
\end{proof}
\subsection{Deterministic approximation algorithms}\label{sec:determ}
In this section, we outline the deterministic algorithms,
for size and matroid constraints.
The main idea is similar, but we
replace \grg with a deterministic subroutine.
For simplicity, we present a randomized version in Appendix~\ref{apx:determ} as Alg.~\ref{alg:determ},
which we then derandomize (Alg.~\ref{alg:determ2} in Appendix~\ref{apx:derand}).
Further discussion is provided in Appendix~\ref{apx:determ-full}.

\textbf{Algorithm overview.}
\citet{chen2023approximation} proposed a randomized algorithm,
\itpg,
which may be thought of as an interpolation between 
standard greedy~\citep{DBLP:journals/mp/NemhauserWF78}
and \rg~\citep{buchbinder2014submodular}.
Instead of picking $k$ elements, each randomly
chosen from the top $k$ marginal gains,
it picks $\ell = \oh{1 / \epsi}$ sets randomly
from $\oh{\ell}$ candidates. Although it uses
only a constant number of rounds, the recurrences
for \itpg are similar to the \rg{} ones discussed
above, so we can guide it similarly.

To select the candidate sets in each iteration,
we replace \ig \note{(the subroutine of \itpg proposed in \citet{chen2023approximation})} with a guided version:
\gig (Alg.~\ref{alg:gig} in Appendix~\ref{apx:ig-size}) for size constraint,
and \gigm (Alg.~\ref{alg:gigm} in Appendix~\ref{apx:ig}) for matroid constraint.
Since only $\ell$ random choices are made, each
from $\oh{\ell}$ sets, there are at most
$\oh{\ell^{\oh{ \ell }}}$ possible solutions,
where $\ell$ is a constant number depending on $\epsi$.
Notably, we are still able to obtain the same approximation
factors as in Section \ref{sec:0.385}. 
The proof of Theorem~\ref{thm:dgig}
is provided in Appendices~\ref{apx:determ} and~\ref{apx:derand}.
\begin{restatable}{theorem}{thmdgig}
\label{thm:dgig}
Let $(f,k)$ be an instance of \sm, with the optimal solution set $O$.
Alg.~\ref{alg:determ2} achieves a deterministic $(0.385-\epsi)$ approximation ratio
with $t = 0.372$, 
and a deterministic $(0.305-\epsi)$ approximation ratio with $t=0.559$.
The query complexity of the algorithm is $\oh{kn\ell^{2\ell-1}}$
where $\ell = \frac{10}{9\epsi}$.
\end{restatable}

\section{Deterministic Algorithm with Nearly Linear Query Complexity} \label{sec:0.377}
In this section, we sketch a deterministic algorithm with
$(0.377-\epsi)$ approximation ratio and $\ohs{\epsi}{n\log(k)}$ query complexity
for the size constraint. A full pseudocode (Alg.~\ref{alg:dgig-2})
and analysis is provided in Appendix~\ref{apx:0.377}.

\textbf{Description of algorithm.} Our goal
is to improve the asymptotic $\mathcal O_\epsi ( kn )$ query
complexity. 
Recall that in Section \ref{sec:0.385},
we described a deterministic algorithm
that employed the output of local
search to guide the \itpg algorithm,
which obeys similar recurrences to \rg.
To produce the $\ell$ candidate
sets for each iteration of \itpg,
a greedy algorithm (guided \ig) is used.
These algorithms can be sped
up using a descending thresholds technique.
This results in \tgig (Alg.~\ref{alg:tgig} in Appendix~\ref{apx:tgig}),
which achieves $\ohs{\epsi}{n \log k}$ query complexity
for the guided part of our algorithm. 
\begin{wrapfigure}{r}{0.4\textwidth}
    \includegraphics[width=0.38\textwidth]{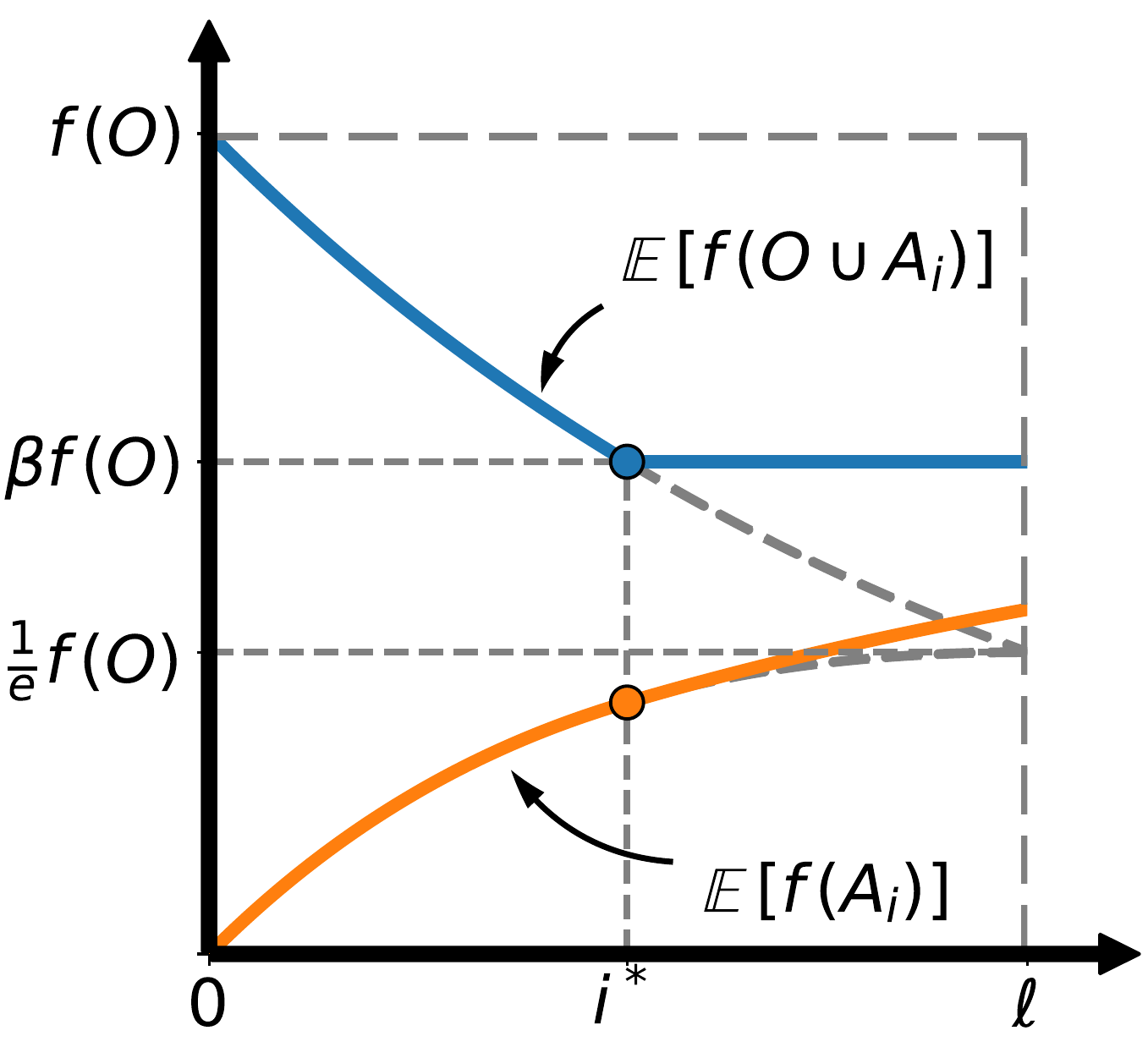}
    \caption{Depiction of how analysis of \itpg changes
      if there is no 
      $(0.377, 0.46)$-guidance set.}
    \label{fig:0.377-deg}
\end{wrapfigure}

However, the local search \fls{} still requires $\oh{kn / \epsi}$
queries, so we seek to find a guidance set in a faster way.
Recall that, in the definition of guidance set $Z$, the
value $f(Z)$
needs to dominate both $f( O \cap Z )$ and $f( O \cup Z)$.
To achieve this with faster query complexity, we employ
a run of unguided \itpg. 
Consider the recurrences plotted in
Fig. \ref{fig:grg-deg}(a) -- if the worst-case degradation occurs,
then at some point the value of $f(A_i)$ becomes close to $f( O \cup A_i )$.
On the other hand, if the worst-case degradation does not occur, then the approximation
factor of $\itpg$ is improved (see Fig. \ref{fig:0.377-deg}).
Moreover, if we ensure that at every stage,
$A_i$ contains no elements that contribute a negative gain, then we will
also have $f(A_i) \ge f(O \cap A_i)$.

To execute this idea, we run
(derandomized, unguided) \itpg, and consider
all $\oh{\ell^\ell}$ intermediate solutions.
Each one of these is pruned (by which we mean,
any element with negative contribution is discarded
until none such remain). Then, the guided part of our
algorithm executes with every possible candidate
intermediate solution as the guiding set; finally,
the feasible set encountered with maximum $f$ value is returned.

The tradeoff
between the first and second parts of the algorithm
is optimized 
with $\alpha=0.377$ and $\beta=0.46$. That is,
if \itpg produces an $(\alpha, \beta)$-guidance
set, the guided part of our algorithm achieves
ratio $0.377$; otherwise, \itpg has ratio
at least $0.377$. We have the following theorem.
The algorithms and analysis sketched here are
formalized in Appendix \ref{apx:0.377}. 
\begin{restatable}{theorem}{thmdgigtwo}
\label{thm:dgig-2}
Let $(f,k)$ be an instance of \sm, with the optimal solution set $O$.
Algorithm~\ref{alg:dgig-2} achieves a deterministic $(0.377-\epsi)$ approximation ratio 
with 
$\mathcal O (n\log(k){\ell_1}^{2\ell_1}{\ell_2}^{2\ell_2-1})$
queries, where 
$\ell_1 = \frac{10}{3\epsi}$ and $\ell_2 = \frac{5}{\epsi}$.
\end{restatable}
%


\section{Empirical Evaluation}\label{sec:exp}

In this section, \note{along with Appendix~\ref{exp-appendix},} we implement and empirically evaluate our
randomized $(0.385 - \epsi)$-approximation algorithm (Alg. \ref{alg:irg}, $\fls$+$\grg$) on two applications of
size-constrained \sm,
and compare to several baselines in terms of objective value of solution and number of queries
to $f$. In summary, our algorithm uses roughly twice the queries as the standard greedy algorithm,
but obtains competitive objective values with an expensive local search that uses one to two
orders of magnitude more queries.
\footnote{Our code is available at https://gitlab.com/luciacyx/guided-rg.git.}

\textbf{Baselines.} 1) \sg: the classical greedy algorithm
\citep{DBLP:journals/mp/NemhauserWF78},
which often performs well empirically
on non-monotone objectives despite having no theoretical guarantee.
2) \rg, the current state-of-the-art combinatorial algorithm as discussed extensively above. 3) The local search
algorithm of \citet{DBLP:conf/stoc/LeeMNS09}, which is the
only prior polynomial-time local search algorithm with a theoretical guarantee:
ratio $1/4 - \epsi$ in $\oh{k^5\log(k)n/\epsi}$ queries. As our emphasis is on
theoretical guarantees above $1/e$, we set $\epsi = 0.01$ for our algorithm, which
yields ratio at least $0.375$ in this evaluation. For \citet{DBLP:conf/stoc/LeeMNS09},
we set $\epsi = 0.1$, which is the standard value of the accuracy parameter in the literature
-- running their algorithm with $\epsi = 0.01$ produced identical results. 

\textbf{Applications and datasets.}
For instances of \sm upon which to evaluate,
we chose video summarization  and maximum cut (MC).
For video summarization, our objective is to select a subset of frames from a video to create a summary.
As in \citet{banihashem2023dynamic} ,
we use a Determinantal Point Process objective function to select a diverse set of elements \citep{kulesza2012determinantal}.
Maximum cut is a classical example of a non-monotone, submodular objective function. 
We run experiments on unweighted Erd{\H{o}}s-R{\'e}nyi (ER), Barab{\'a}si-Albert (BA) and Watts-Strogatz (WS) graphs which have been used to model many real-world networks. 
The formal definition of problems, details of datasets, and hyperparameters of graph generation can be found in the Appendix \ref{exp-appendix}. In video summarization, there are $n = 100$ frames.
On all the instances of maximum cut, the number of vertices $n = 10000$. The mean of 20 independent runs is plotted,
and the shaded region represents one standard deviation about the mean. 

\begin{figure}[t] 
  \subfigure[Video Summarization, solution value] { 
    \includegraphics[width=0.23\textwidth]{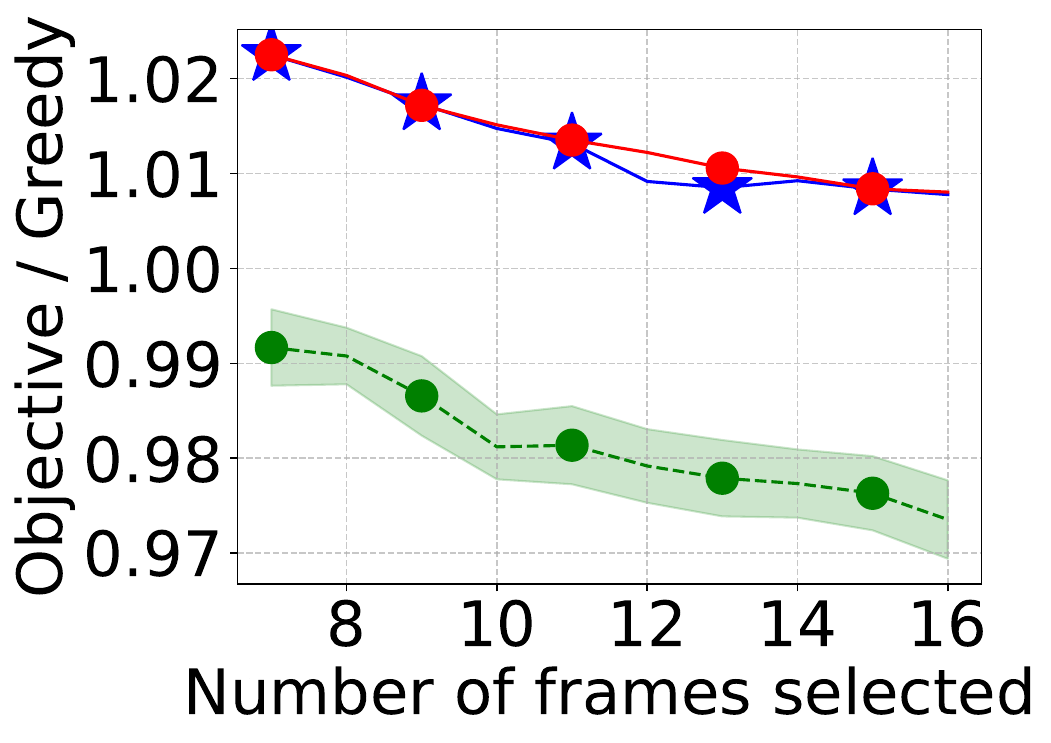}
  }
  \subfigure[Video Summarization, queries] { 
    \includegraphics[width=0.23\textwidth]{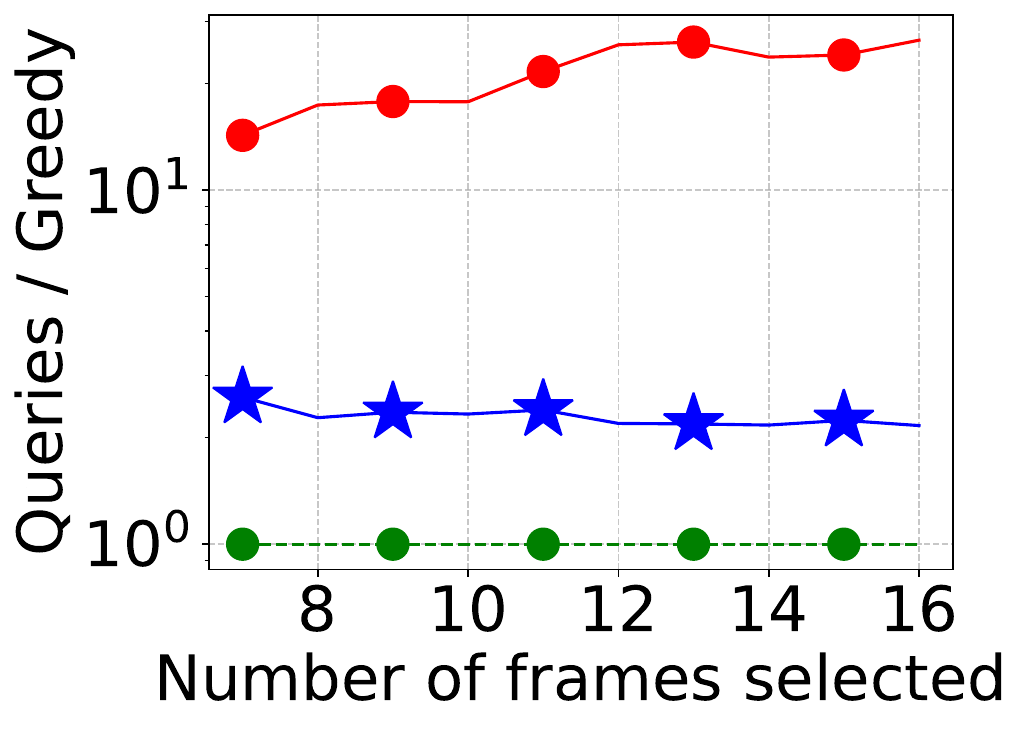}
  }
  \subfigure[Maximum Cut (ER), solution value] { 
    \includegraphics[width=0.225\textwidth]{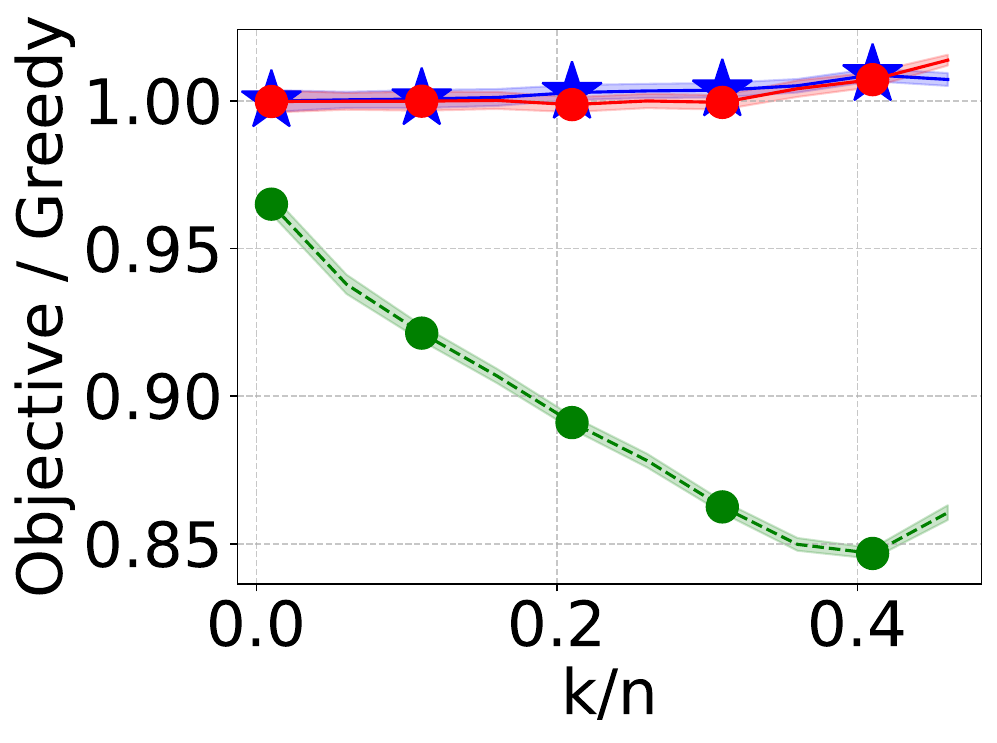}
  }
    \subfigure[Maximum Cut (ER), queries] { 
    \includegraphics[width=0.22\textwidth]{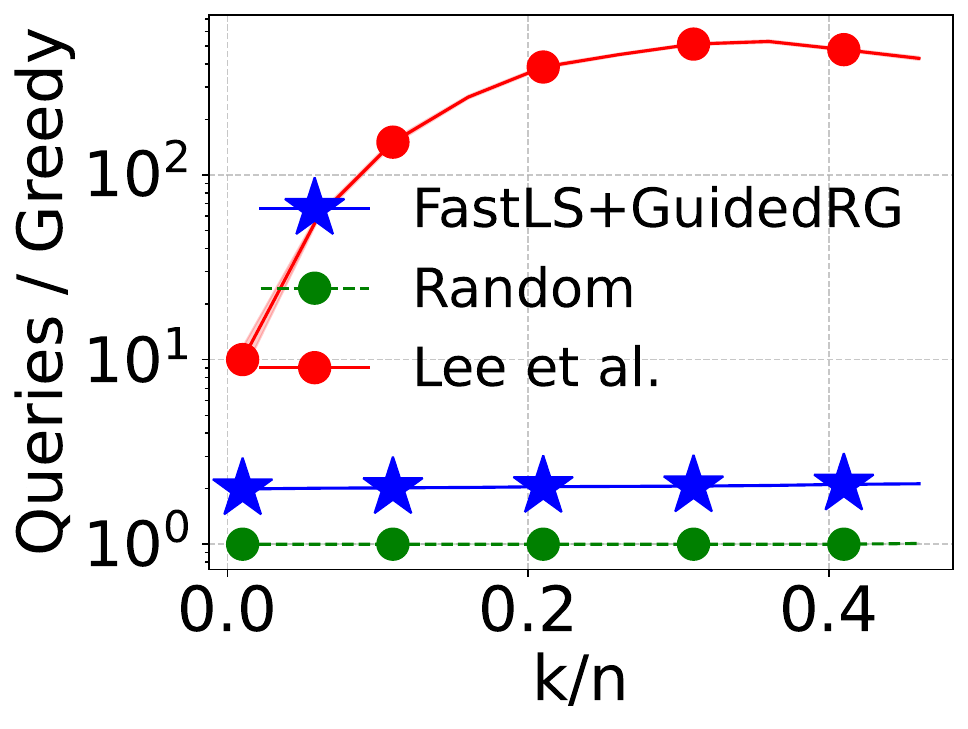}
  }
  
  
  \caption{The objective value (higher is better) and the number of queries (log scale, lower is better) are normalized by those of \textsc{StandardGreedy}. Our algorithm (blue star) outperforms every baseline on at least one of these two metrics. }
  \label{fig:MAXCUTVideoSummarization}
  \vspace*{-1em}
\end{figure}
\textbf{Results.}
As shown in Figure \note{\ref{fig:MAXCUTVideoSummarization}
in this section, and Figure \ref{fig:frames}
and \ref{max_cut_additional_results} in Appendix~\ref{exp-appendix}},
on both applications, \fls+\grg produces solutions of higher objective
value than \sg, and also higher than \rg.  The objective values of \fls+\grg often matches with
\citet{DBLP:conf/stoc/LeeMNS09} which performs the best; this agrees with the intuition that, empirically, local search is nearly optimal. In terms of queries, our algorithm uses roughly twice the number
of queries as \sg, but we improve on \citet{DBLP:conf/stoc/LeeMNS09}
typically by at least a factor of $10$ and often by more than a factor of $100$. 



\section{Discussion and Limitations} \label{sec:limitations}
Prior to this work, the state-of-the-art combinatorial ratios
were $1/e \approx 0.367$ and $0.283$ for size constrained
and matroid constrained \sm, respectively, both achieved
by the \rg{} algorithm. In this work, we show 
how to guide \rg{} with a fast local search algorithm
to achieve ratios $0.385$ and $0.305$,
respectively, in $\oh{kn / \epsi }$ queries. The resulting algorithm
is practical and empirically outperforms both
\rg{} and standard greedy in objective value on several applications
of \sm{}. However, if $k$ is on the order of $n$,
the query complexity is quadratic in $n$, which is too
slow for modern data sizes. Therefore, an interesting
question for future work is whether further improvements
in the query complexity to achieve these ratios (or better ones) could be made. 

In addition, we achieve the same approximation ratios and asymptotic
query complexity with deterministic algorithms, achieved by guiding a different
algorithm; moreover, we speed up the deterministic algorithm to $\mathcal O_\epsi (n \log k)$
by obtaining the guidance set in another way. This result is a partial answer to the limitation in the previous paragraph, as we achieve a ratio beyond $1/e$ in nearly linear query complexity. However, for all of our deterministic algorithms, there is an exponential dependence on $1/\epsi$, which makes these algorithms impractical and mostly of theoretical interest. 
\bibliographystyle{apalike} 
\bibliography{main.bib}

\begin{thebibliography}{}

\bibitem[Banihashem et~al., 2023]{banihashem2023dynamic}
Banihashem, K., Biabani, L., Goudarzi, S., Hajiaghayi, M., Jabbarzade, P., and
  Monemizadeh, M. (2023).
\newblock Dynamic non-monotone submodular maximization.
\newblock {\em Advances in Neural Information Processing Systems 36: Annual
  Conference on Neural Information Processing Systems 2023, NeurIPS 2023}.

\bibitem[Bao et~al., 2022]{DBLP:conf/kdd/BaoHZ22}
Bao, W., Hang, J., and Zhang, M. (2022).
\newblock Submodular feature selection for partial label learning.
\newblock In {\em Proceedings of the 28th {ACM} {SIGKDD} Conference on
  Knowledge Discovery and Data Mining, {KDD} 2022}.

\bibitem[Bilmes, 2022]{DBLP:journals/corr/abs-2202-00132}
Bilmes, J.~A. (2022).
\newblock Submodularity in machine learning and artificial intelligence.
\newblock {\em arXiv preprint arXiv:2202.00132}.

\bibitem[Brualdi, 1969]{brualdi1969comments}
Brualdi, R.~A. (1969).
\newblock Comments on bases in dependence structures.
\newblock {\em Bulletin of the Australian Mathematical Society}, 1(2):161--167.

\bibitem[Buchbinder and Feldman, 2018]{buchbinder2018deterministic}
Buchbinder, N. and Feldman, M. (2018).
\newblock Deterministic algorithms for submodular maximization problems.
\newblock {\em {ACM} Trans. Algorithms}, 14(3):32:1--32:20.

\bibitem[Buchbinder and Feldman, 2019]{buchbinder2019constrained}
Buchbinder, N. and Feldman, M. (2019).
\newblock Constrained submodular maximization via a nonsymmetric technique.
\newblock {\em Math. Oper. Res.}, 44(3):988--1005.

\bibitem[Buchbinder and Feldman, 2023]{buchbinder2023constrained}
Buchbinder, N. and Feldman, M. (2023).
\newblock Constrained submodular maximization via new bounds for dr-submodular
  functions.
\newblock {\em arXiv preprint arXiv:2311.01129}.

\bibitem[Buchbinder et~al., 2023]{DBLP:journals/siamcomp/BuchbinderFG23}
Buchbinder, N., Feldman, M., and Garg, M. (2023).
\newblock Deterministic {(1/2} + {\(\epsilon\)})-approximation for submodular
  maximization over a matroid.
\newblock {\em {SIAM} J. Comput.}, 52(4):945--967.

\bibitem[Buchbinder et~al., 2014]{buchbinder2014submodular}
Buchbinder, N., Feldman, M., Naor, J., and Schwartz, R. (2014).
\newblock Submodular maximization with cardinality constraints.
\newblock In {\em Proceedings of the Twenty-Fifth Annual {ACM-SIAM} Symposium
  on Discrete Algorithms, {SODA} 2014}.

\bibitem[Buchbinder et~al., 2017]{buchbinder2017comparing}
Buchbinder, N., Feldman, M., and Schwartz, R. (2017).
\newblock Comparing apples and oranges: Query trade-off in submodular
  maximization.
\newblock {\em Math. Oper. Res.}, 42(2):308--329.

\bibitem[Chen and Kuhnle, 2023]{chen2023approximation}
Chen, Y. and Kuhnle, A. (2023).
\newblock Approximation algorithms for size-constrained non-monotone submodular
  maximization in deterministic linear time.
\newblock In {\em Proceedings of the 29th {ACM} {SIGKDD} Conference on
  Knowledge Discovery and Data Mining, {KDD} 2023}.

\bibitem[Chen and Kuhnle, 2024]{DBLP:journals/jair/ChenK24}
Chen, Y. and Kuhnle, A. (2024).
\newblock Practical and parallelizable algorithms for non-monotone submodular
  maximization with size constraint.
\newblock {\em J. Artif. Intell. Res.}, 79:599--637.

\bibitem[Ene and Nguyen, 2016]{DBLP:conf/focs/EneN16}
Ene, A. and Nguyen, H.~L. (2016).
\newblock Constrained submodular maximization: Beyond 1/e.
\newblock In {\em {IEEE} 57th Annual Symposium on Foundations of Computer
  Science, {FOCS} 2016}.

\bibitem[Feige et~al., 2011]{DBLP:journals/siamcomp/FeigeMV11}
Feige, U., Mirrokni, V.~S., and Vondr{\'{a}}k, J. (2011).
\newblock Maximizing non-monotone submodular functions.
\newblock {\em {SIAM} J. Comput.}, 40(4):1133--1153.

\bibitem[Feldman et~al., 2011]{DBLP:conf/focs/FeldmanNS11}
Feldman, M., Naor, J., and Schwartz, R. (2011).
\newblock A unified continuous greedy algorithm for submodular maximization.
\newblock In {\em {IEEE} 52nd Annual Symposium on Foundations of Computer
  Science, {FOCS} 2011}.

\bibitem[Gharan and Vondr{\'{a}}k, 2011]{DBLP:conf/soda/GharanV11}
Gharan, S.~O. and Vondr{\'{a}}k, J. (2011).
\newblock Submodular maximization by simulated annealing.
\newblock In {\em Proceedings of the Twenty-Second Annual {ACM-SIAM} Symposium
  on Discrete Algorithms, {SODA} 2011}.

\bibitem[Han et~al., 2020]{han2020deterministic}
Han, K., Cao, Z., Cui, S., and Wu, B. (2020).
\newblock Deterministic approximation for submodular maximization over a
  matroid in nearly linear time.
\newblock In {\em Advances in Neural Information Processing Systems 33: Annual
  Conference on Neural Information Processing Systems 2020, NeurIPS 2020}.

\bibitem[Khanna et~al., 2017]{DBLP:conf/aistats/KhannaEDNG17}
Khanna, R., Elenberg, E.~R., Dimakis, A.~G., Negahban, S.~N., and Ghosh, J.
  (2017).
\newblock Scalable greedy feature selection via weak submodularity.
\newblock In {\em Proceedings of the 20th International Conference on
  Artificial Intelligence and Statistics, {AISTATS} 2017}.

\bibitem[Krause et~al., 2008]{DBLP:journals/jmlr/KrauseSG08}
Krause, A., Singh, A.~P., and Guestrin, C. (2008).
\newblock Near-optimal sensor placements in gaussian processes: Theory,
  efficient algorithms and empirical studies.
\newblock {\em J. Mach. Learn. Res.}, 9:235--284.

\bibitem[Kulesza et~al., 2012]{kulesza2012determinantal}
Kulesza, A., Taskar, B., et~al. (2012).
\newblock Determinantal point processes for machine learning.
\newblock {\em Foundations and Trends{\textregistered} in Machine Learning},
  5(2--3):123--286.

\bibitem[Lee et~al., 2009]{DBLP:conf/stoc/LeeMNS09}
Lee, J., Mirrokni, V.~S., Nagarajan, V., and Sviridenko, M. (2009).
\newblock Non-monotone submodular maximization under matroid and knapsack
  constraints.
\newblock In {\em Proceedings of the 41st Annual {ACM} Symposium on Theory of
  Computing, {STOC} 2009}.

\bibitem[Liu et~al., 2013]{DBLP:conf/icassp/LiuWKSB13}
Liu, Y., Wei, K., Kirchhoff, K., Song, Y., and Bilmes, J.~A. (2013).
\newblock Submodular feature selection for high-dimensional acoustic score
  spaces.
\newblock In {\em {IEEE} International Conference on Acoustics, Speech and
  Signal Processing, {ICASSP} 2013}.

\bibitem[Mirzasoleiman et~al., 2018]{DBLP:conf/aaai/MirzasoleimanJ018}
Mirzasoleiman, B., Jegelka, S., and Krause, A. (2018).
\newblock Streaming non-monotone submodular maximization: Personalized video
  summarization on the fly.
\newblock In {\em Proceedings of the Thirty-Second {AAAI} Conference on
  Artificial Intelligence, AAAI 2018}.

\bibitem[Nemhauser and Wolsey, 1978]{DBLP:journals/mor/NemhauserW78}
Nemhauser, G.~L. and Wolsey, L.~A. (1978).
\newblock Best algorithms for approximating the maximum of a submodular set
  function.
\newblock {\em Math. Oper. Res.}, 3(3):177--188.

\bibitem[Nemhauser et~al., 1978]{DBLP:journals/mp/NemhauserWF78}
Nemhauser, G.~L., Wolsey, L.~A., and Fisher, M.~L. (1978).
\newblock An analysis of approximations for maximizing submodular set functions
  - {I}.
\newblock {\em Math. Program.}, 14(1):265--294.

\bibitem[{\"{O}}zcan et~al., 2021]{DBLP:conf/sdm/OzcanMI21}
{\"{O}}zcan, G., Moharrer, A., and Ioannidis, S. (2021).
\newblock Submodular maximization via taylor series approximation.
\newblock In Demeniconi, C. and Davidson, I., editors, {\em Proceedings of the
  2021 {SIAM} International Conference on Data Mining, {SDM} 2021, Virtual
  Event, April 29 - May 1, 2021}, pages 423--431. {SIAM}.

\bibitem[Powers et~al., 2016]{DBLP:conf/fusion/PowersBKA16}
Powers, T., Bilmes, J.~A., Krout, D.~W., and Atlas, L.~E. (2016).
\newblock Constrained robust submodular sensor selection with applications to
  multistatic sonar arrays.
\newblock In {\em 19th International Conference on Information Fusion, {FUSION}
  2016}.

\bibitem[Qi, 2022]{DBLP:conf/isaac/Qi22}
Qi, B. (2022).
\newblock On maximizing sums of non-monotone submodular and linear functions.
\newblock In {\em 33rd International Symposium on Algorithms and Computation,
  {ISAAC} 2022}.

\bibitem[Roy, 2022]{CookingShowWithRoy}
Roy (2022).
\newblock Cooking video.
\newblock https://youtu.be/voUDP4rUKvQ?si=ZUezR4jMVzOz5Kcw.

\bibitem[Sun et~al., 2022]{sun2022improved}
Sun, X., Zhang, J., Zhang, S., and Zhang, Z. (2022).
\newblock Improved deterministic algorithms for non-monotone submodular
  maximization.
\newblock In {\em Computing and Combinatorics - 28th International Conference,
  {COCOON} 2022}.

\bibitem[Tschiatschek et~al., 2014]{DBLP:conf/nips/TschiatschekIWB14}
Tschiatschek, S., Iyer, R.~K., Wei, H., and Bilmes, J.~A. (2014).
\newblock Learning mixtures of submodular functions for image collection
  summarization.
\newblock In {\em Advances in Neural Information Processing Systems 27: Annual
  Conference on Neural Information Processing Systems 2014, NeurIPS 2014}.

\bibitem[Tukan et~al., 2024]{tukan2024practical}
Tukan, M., Mualem, L., and Feldman, M. (2024).
\newblock Practical $0.385 $-approximation for submodular maximization subject
  to a cardinality constraint.
\newblock {\em arXiv preprint arXiv:2405.13994}.

\end{thebibliography}
\clearpage
\appendix
\section{Additional Preliminaries} \label{apx:prelim}
\subsection{Constraints}\label{apx:prb-def}
In this paper, our focus lies on two constraints: 
size constraint and matroid constraint.
For size constraint, we define the feasible subsets as 
$\iset(k) = \{S\subseteq \uni: |S| \le k\}$.
For matroid constraint, the definition is as follows:
\begin{definition}\label{def:mtr}
A matroid $\mtr$ is a pair $(\uni, \iset)$,
where $\uni$ is the ground set and $\iset$ is the independent sets
with the following properties:
(1) $\emptyset \in \iset$; 
(2) hereditary property: $A \in \iset \Rightarrow B \in \iset, \forall B \subseteq A$;
(3) exchange property: $A, B \in \iset, |A| > |B| \Rightarrow \exists x \in A\setminus B, \st B+x \in \iset$.
\end{definition}
Specifically, we use $\iset(\mtr)$ to represent the independent sets of matroid $\mtr$.
A maximal independent set in $\iset(\mtr)$ is called a basis.
Let $k$ be the size of the maximal independent set.

In the following, we consider an extended matroid with $k$ dummy elements added to
the ground set $\mtr'=(\uni', \iset')$.
We show that \sm on $\mtr'$ return the same solution as on $\mtr$.
\begin{lemma}\label{lemma:mtr-dummy}
Let $\mtr = (\uni, \iset)$ be a matroid,
and $\uni' = \uni \cup E$,
where $E = \{e_0^1, \ldots, e_0^k\}$ and $e_0^i$ is a dummy element for each $i \in [k]$.
Let $\iset' = \bigcup_{S\in \iset} \{S,S\cup \{e_0^1\}, \ldots, S\cup \{e_0^1, \ldots, e_0^{k-|S|}\}\}$.
Then, $\mtr'=(\uni', \iset')$ is also a matroid and $\max_{S\in \iset}\ff{S} = \max_{S\in \iset'}\ff{S}$.
\end{lemma}
\begin{proof}
Firstly, we prove that $\mtr'=(\uni', \iset')$ is also a matroid by Definiton~\ref{def:mtr}.

Since $\emptyset \in \iset$ and $\iset \subseteq \iset'$,
it holds that $\emptyset \in \iset'$.

Let $A' \in \iset'$, and $A = A'\setminus E$.
Then, $A \in \iset$.
For every $B'\subseteq A'$,
since $(\uni, \iset)$ is a matroid,
$B = B'\setminus E \subseteq A \in \iset$.
Since $| B'\cap E|\le |A'\cap E| \le k-|A| \le k-|B|$,
it holds that $B \in \iset'$ by the construction of $\iset'$. 

Let $A', B' \in \iset'$ and $|A'|> |B'|$.
Let $A = A'\setminus E$ and $B = B'\setminus E$.
Then, $A, B\in \iset$.
If $|A| > |B|$, by Def.~\ref{def:mtr},
there exists $x \in A\setminus B$ \st $B+x \in \iset$.
Since $(B'+x)\setminus (B+x) \subseteq E$ and $|B'| < |A'|\le k$,
$B'+x \in \iset'$.
Otherwise, $|A| < |B|$, which indicates that $|A'\cap E| > |B'\cap E|$.
Since $|B'| < |A'|\le k$, by adding a dummy element $e_0 \in A'\cap E\setminus B'$ to $B'$,
it holds that $B'+e_0 \in \iset'$.

Thus, by Def.~\ref{def:mtr}, $\mtr'=(\uni', \iset')$ is also a matroid.

As for \sm on $\iset'$,
since dummy element does not contribute to the objective value,
it holds that, for every $S \in \iset$,
$\ff{S} = \ff{S\cup E'}$, where $E'\subseteq E$.
Then, $\{\ff{S}: S\in \iset\} = \{\ff{S}: S\in \iset'\}$. 
Further, $\max_{S\in \iset}\ff{S} = \max_{S\in \iset'}\ff{S}$.
\end{proof}

\subsection{Technical Lemma} \label{apx:lemmata}
\begin{lemma}\label{lemma:mtr-bij}
    (\citet{brualdi1969comments})
    If $B_1$ and $B_2$ are finite bases,
    then there exists a bijection 
    $\sigma:B_1\setminus B_2\to B_2\setminus B_1$
    such that $B_2 +e -\sigma(e)$ is a basis for all $e \in B_1\setminus B_2$
\end{lemma}
\begin{lemma}\label{lemma:recurrence}
Let $a \in \reals$, $b, X_0\in \mathbb{R}$.
If $X_i \ge aX_{i-1} +b$ for every $i \in[k]$,
then
\[X_k \ge a^kX_0 + \frac{b\left(1-a^k\right)}{1-a}.\]
\end{lemma}
\begin{proof}
By repeatedly implementing that $X_i \ge aX_{i-1} +b$,
we can bound $X_k$ as follows,
\begin{align*}
X_k &\ge aX_{k-1} +b\\
&\ge a^2X_{k-2} +ab + b\\
&\ldots\\
&\ge a^kX_0 + b\left(\sum_{i=0}^{k-1}a^i\right)\\
&=a^kX_0 + \frac{b\left(1-a^k\right)}{1-a}.\qedhere
\end{align*}
\end{proof}
\begin{lemma}\label{lemma:val-inq}
    \begin{align*}
        & 1-\frac{1}{x}\le \log(x) \le x-1, &\forall x>0\\
        & 1-\frac{1}{x+1}\ge e^{-\frac{1}{x}} , &\forall x\in \mathbb{R}\\
        & (1-x)^{y-1}\ge e^{-xy}, &\forall xy \le 1
    \end{align*}
\end{lemma}

\section{Query Complexity Analysis of the Continuous Algorithm in \citet{buchbinder2019constrained}}\label{apx:0.385-time}
\begin{algorithm}[h]\caption{Aided Measured Continuous Greedy $(f, P, Z, t_s)$~\citep{buchbinder2019constrained}}\label{alg:amcg} 
    \SetKwInOut{Input}{input}
    \Input{oracle $f$, a solvable down-closed polytope $P$, a set $Z\in \uni$, $t_s \in (0,1)$}
    \tcc{Initialization}
    Let $\bar{\delta}_1 \leftarrow t_s \cdot n^{-4}$ and $\bar{\delta}_2 \leftarrow\left(1-t_s\right) \cdot n^{-4}$\;
    Let $t \gets 0$ and $y(t) \gets \mathbf{1}_{\emptyset}$\;
    \tcc{Growing $y(t)$}
    \While{$t<1$}{\label{line:amcg-while-begin}
        \ForEach{$u \in \uni$}{
        Let $w_u(t)$ be an estimate of $\ex{\marge{u}{R(y(t))}}$ obtained by averaging the values of $\marge{u}{R(y(t))}$ for $r=\left\lceil 48n^6\log(2n) \right\rceil$ independent samples of $R(y(t))$\;\label{line:amcg-query}
        }
        Let $x(t) \leftarrow \begin{cases}\argmax _{x \in P}\left\{\sum_{u \in \uni \backslash Z} w_u(t) \cdot x_u(t)-\sum_{u \in Z} x_u(t)\right\} & \text { if } t \in\left[0, t_s\right), \\ \argmax _{x \in P}\left\{\sum_{u \in \uni} w_u(t) \cdot x_u(t)\right\} & \text { if } t \in\left[t_s, 1\right) .\end{cases}$\;
        Let $\delta_t$ be $\bar{\delta}_1$ when $t < t_s$ and $\bar{\delta}_2$ when $t\ge t_s$\;
        Let $y\left(t+\delta_t\right) \leftarrow y(t)+\delta_t\left(\mathbf{1}_{\uni}-y(t)\right) \circ x(t)$\;
        Update $t\gets t+\delta_t$\;
    }\label{line:amcg-while-end}
    \Return{$y(1)$}
\end{algorithm}
In this section, we analyze the query complexity of 
Aided Measured Continuous Greedy (Alg.~\ref{alg:amcg})
proposed by \citet{buchbinder2019constrained},
which is $\oh{n^{11}\log(n)}$.

In Alg.~\ref{alg:amcg}, queries to the oracle $f$ only occur on Line~\ref{line:amcg-query}.
For each element $u$ in the ground set $\uni$,
$r=\left\lceil 48n^6\log(2n) \right\rceil$ queries are made.
These queries correspond to
$r$ independent samples of $R(y(t))$ to estimate $\ex{\marge{u}{R(y(t))}}$.
Therefore, there are $nr=\oh{n^7\log(n)}$ queries for each iteration of the while loop
(Line~\ref{line:amcg-while-begin}-\ref{line:amcg-while-end}).

Time variable $t$ is increased by $\bar{\delta}_1=t_s \cdot n^{-4}$, when $t < t_s$,
and is increased by $\bar{\delta}_2=\left(1-t_s\right) \cdot n^{-4}$, when $t \ge t_s$.
Thus, there are a total of $2n^{4}$ iterations within the while loop.
In conclusion, the total number of queries made by the algorithm is $\oh{n^{11}\log(n)}$.

\section{Analysis of Randomized Approximation Algorithm, Alg.~\ref{alg:irg}}  \label{apx:irg}
In this section, we provide a detailed analysis of our randomized approximation algorithm
and its two components, \fls and \grg.
This section is organized as follows:
Appendix~\ref{apx:fls} analyzes the theoretical guarantee of single run of \fls (Appendix~\ref{apx:fls-proof}),
which is needed to show that it finds a good guidance set.
Then, although it is not needed for our results, we show the \fls independently
achieves an approximation ratio achieved for \nmon under monotone and non-monotone objectives (Appendix~\ref{apx:fls-ratio}).

In Appendix~\ref{apx:grg}, we provide pseudocode and formally prove the results for \grg. Specifically, we solve the
recurrences for both size and matroid constraint. 

\subsection{Analysis of \fls (Alg.~\ref{alg:fls})}\label{apx:fls}
\begin{algorithm}[t]
    \caption{A fast local search algorithm 
      with query complexity $\oh{kn/\epsi}$.
    }\label{alg:fls}
    \Proc{\fls($f, \iset(\mtr), Z_0, \epsi$)}
    {\textbf{Input:} oracle $f$, matroid constraint $\iset(\mtr)$,
        an approximation result $Z_0$, accuracy parameter $\epsi$\;
        \textbf{Initialize:} $Z \gets Z_0$
        \tcc*[r]{add dummy elements to $Z$ until $|Z|=k$}
        \While{$\exists a \in Z, e \in \uni \setminus Z, \st Z-a+e \in \iset(\mtr)$ and $ \marge{e}{Z} - \marge{a}{Z \setminus a} \ge \frac{\epsi}{k}\ff{Z}$ \;}{
            $Z\gets Z-a+e$ \label{line:fls-rpl}\;
        }
        \textbf{return} $Z$
      }
    \end{algorithm}
    Pseudocode for \fls is provided in Alg. \ref{alg:fls}.
\subsubsection{Finding A Good Guidance Set -- Proofs for Lemma~\ref{lemma:fls} of Alg.~\ref{alg:fls} in Section~\ref{sec:fls}}\label{apx:fls-proof}
Recall that \fls (Alg.~\ref{alg:fls}) takes a matroid constraint $\iset(\mtr)$ and
an approximation result $Z_0$ as inputs, and outputs a local optimum $Z$.
Here, we restate the theoretical guarantees of \fls (Lemma~\ref{lemma:fls}).
Using the conclusions drawn from Lemma~\ref{lemma:fls}, 
we demonstrate in Corollary~\ref{cor:fls} that 
$Z$ is a $\left((1+\epsi)\alpha, \alpha\right)$-guidance set.
At the end of this section,
we provide the proof for Lemma~\ref{lemma:fls}.
\lemmafls*
\begin{restatable}{corollary}{corfls}
\label{cor:fls}
Let $Z$ be the solution of \fls$(f, \iset(\mtr), Z_0, \epsi)$.
If $\ff{Z} < \alpha \opt$, $Z$ is a $\left((1+\epsi)\alpha, \alpha\right)$-guidance set.
\end{restatable}
\begin{proof}[Proof of Corollary~\ref{cor:fls}]
By Lemma~\ref{lemma:fls}, let $S = O$ and $S = O \cap Z$ respectively,
it holds that
\begin{align*}
&\ff{O\cup Z} + \ff{O\cap Z} < (2+\epsi)\ff{Z},\\
&\ff{O\cap Z} < (1+\epsi)\ff{Z}.
\end{align*}
If $\ff{Z} < \alpha \opt$, then
\begin{align*}
&\ff{O\cup Z} + \ff{O\cap Z} < (2+\epsi)\alpha \opt,\\
&\ff{O\cap Z} < (1+\epsi)\alpha \opt.
\end{align*}
By Definition~\ref{def:guide},
$Z$ is a $\left((1+\epsi)\alpha, \alpha\right)$-guidance set.
\end{proof}
In the following, we prove Lemma~\ref{lemma:fls} for Alg.~\ref{alg:fls}.
\begin{proof}[Proof of Lemma~\ref{lemma:fls}]
\textbf{Query Complexity.}
For each successful replacement of elements on Line~\ref{line:fls-rpl},
it holds that $\marge{e}{Z} - \marge{a}{Z-a} \ge
     \frac{\epsi}{k}\ff{Z}$.
By submodularity, 
\[\ff{Z-a+e}-\ff{Z} = \marge{e}{Z-a}-\marge{a}{Z-a}
\ge \marge{e}{Z} - \marge{a}{Z-a} \ge
     \frac{\epsi}{k}\ff{Z}.\]
Hence, the oracle value $f(Z)$ is increased 
by a factor of at least $(1+\epsi/k)$ after the swap.
Therefore, there are at most 
$\left\lceil\log_{1+\epsi/k}\brk{\frac{\ff{O}}{\ff{Z_0}}}\right\rceil$ iterations,
since otherwise, it would entail $\ff{Z} > \ff{O}$,
which contradicts the fact that $O$ is the optimal solution.
Then, since the algorithm makes at most $\oh{n}$ queries at each iteration,
the query complexity can be bounded as follows,
\begin{align*}
\text{\# queries}& \le \oh{n \left\lceil\log_{1+\epsi/k}\brk{\frac{\ff{O}}{\ff{Z_0}}}\right\rceil} \le \oh{n\frac{k}{\epsi}\log\left(\frac{1}{\alpha}\right)} \tag{$\ff{Z_0}\ge \alpha \ff{O}$; Lemma~\ref{lemma:val-inq}}
\end{align*}

\textbf{Objective Value.}
For any $S\in \iset(\mtr)$,
we consider adding dummy elements into $S$ until $|S| = k$.
By Lemma~\ref{lemma:mtr-dummy},
we consider $Z$ and $S$ are bases of matroid $\mtr$ with or without dummy elements.
Then, by Lemma~\ref{lemma:mtr-bij},
there exits a bijection $\sigma: S\setminus Z \to Z \setminus S$
such that $Z+e-\sigma(e)$ is a basis for all $e \in S\setminus Z$.
After the algorithm terminates, 
for every $e \in S\setminus Z$,
it holds that,
$\marge{e}{Z} - \marge{\sigma(e)}{Z-\sigma(e)} < \frac{\epsi}{k} \ff{Z}$.
Then,
\begin{align*}
   \epsi \ff{Z}
    > \sum_{e\in S\setminus Z} \marge{e}{Z}-\sum_{a \in Z\setminus S}\marge{a}{Z-a}
\end{align*}
Let $\ell = |S\setminus Z| = |Z\setminus S|$, $S\setminus Z = \{e_1, \ldots, e_\ell\}$,
$Z\setminus S = \{a_1, \ldots, a_\ell\}$.
By submodularity,
\begin{align*}
    \sum_{e\in S\setminus Z} \marge{e}{Z} 
    &=  \brk{\ff{Z+e_1}-\ff{Z}} + \brk{\ff{Z+e_2}-\ff{Z}} + \ldots + \brk{\ff{Z+e_\ell}-\ff{Z}}\\
    &\ge \brk{\ff{Z+e_1+e_2}-\ff{Z}} + \brk{\ff{Z+e_3}-\ff{Z}} + \ldots + \brk{\ff{Z+e_\ell}-\ff{Z}}\\
    &\ge \ldots \\
    &\ge \ff{Z+e_1+\ldots+e_\ell} - \ff{Z} = \ff{S \cup Z} - \ff{Z}. \text{ Also by submodularity,}\\
    \sum_{a\in Z\setminus S} \marge{a}{Z-a}
    &= \sum_{i=1}^\ell \marge{a_i}{Z-a_i}
    \le \sum_{i=1}^\ell \marge{a_i}{Z-a_1-\ldots -a_i}
    = \ff{Z}-\ff{S\cap Z}
\end{align*}
Thus,
\begin{align*}
&\epsi \ff{Z} > \ff{S\cup Z}-\ff{Z}-\brk{\ff{Z}-\ff{S\cap Z}}\\
\Rightarrow & (2+\epsi)\ff{Z} > \ff{S\cup Z} + \ff{S\cap Z}.
\end{align*}
\end{proof}

\subsubsection{Approximation Ratio achieved by \fls}\label{apx:fls-ratio}
In this section, we show that \fls can be employed independently to achieves approximation
ratios of nearly $1/2$ and $1/4$ for both the monotone and non-monotone versions
of the problem, respectively.

\textbf{Monotone submodular functions.}
By employing \fls once, it returns an $\frac{1}{2+\epsi}$-approximation result
in monotone cases.
\begin{theorem}
    Let $\epsi > 0$, and let $(f,\iset(\mtr))$ be an instance of \sm,
    where $f$ is monotone.
    The input set $Z_0$ is an $\alpha_0$-approximate solution to $(f, \iset(\mtr))$.
    \fls (Alg.~\ref{alg:fls}) returns a solution $Z$ 
    such that $\ff{Z} \ge \ff{O}/(2+\epsi)$
    with $\oh{kn\log(1/\alpha_0)/\epsi}$ queries.
\end{theorem}
\begin{proof}
    By Lemma~\ref{lemma:fls}, set $S = O$,
    it holds that
    \[(2+\epsi)\ff{Z} > \ff{O\cup Z}+\ff{O\cap Z}\ge \ff{O},\]
    where the last inequality follows by monotonicity and non-negativity.
\end{proof}

\textbf{Non-monotone submodular functions.}
For the non-monotone problem, $2$ repetitions of \fls (Alg.~\ref{alg:fls-non}) yields a ratio of $\frac{1}{4+2\epsi}$.
The theoretical guarantees and the corresponding analysis are provided as follows. We remark that this is
a primitive implementation of the guiding idea: the second run of \fls avoids the output of the first one. 
\begin{algorithm}[h]\label{alg:fls-non}
   \caption{An $1/(4+2\epsi)$-approximation algorithm with $\oh{kn/\epsi}$}
     \textbf{Input:} oracle $f$, constraint $\iset$, an approximation result $Z_0$, switch point $t$, error rate $\epsi$\;
     $Z_1 \gets \fls(f, \uni, \iset, Z_0, \epsi)$
     \tcc*[r]{run \fls with ground set $\uni$}
     $Z_2 \gets \fls(f, \uni\setminus Z_1, \iset, Z_0, \epsi)$
     \tcc*[r]{run \fls with ground set $\uni\setminus Z_1$}
     \textbf{return} $Z\gets \argmax\{\ff{Z_1}, \ff{Z_2}\}$\;
\end{algorithm}
\begin{theorem}
    Let $\epsi > 0$, and let $(f,\iset(\mtr))$ be an instance of \sm,
    where $f$ is not necessarily monotone.
    The input set $Z_0$ is an $\alpha_0$-approximate solution to $(f, \iset(\mtr))$.
    Alg.~\ref{alg:fls-non} returns a solution $Z$ 
    such that $\ff{Z} \ge \ff{O}/(4+2\epsi)$
    with $\oh{kn\log(1/\alpha_0)/\epsi}$ queries.
\end{theorem}
\begin{proof}
    By repeated application of Lemma~\ref{lemma:fls} for
    the two calls of \fls in Alg.~\ref{alg:fls-non},
    it holds that 
    \begin{align*}
        &\ff{O\cup Z_1}+\ff{O\cap Z_1} < (2+\epsi)\ff{Z_1}\\
        &\ff{(O\setminus Z_1)\cup Z_2}+\ff{(O\setminus Z_1)\cap Z_2} < (2+\epsi)\ff{Z_2}
    \end{align*}
    By summing up the above two inequalities, it holds that
    \begin{align*}
        (4+\epsi)\ff{Z}&\ge (2+\epsi)\ff{Z_1}+(2+\epsi)\ff{Z_2}\\
        &\ge \ff{O\cup Z_1}+\ff{O\cap Z_1}+\ff{(O\setminus Z_1)\cup Z_2}+\ff{(O\setminus Z_1)\cap Z_2}\\
        &\ge \ff{O\cup Z_1}+\ff{(O\setminus Z_1)\cup Z_2}+\ff{O\cap Z_1}\tag{nonnegativity}\\
        &\ge \ff{O\setminus Z_1} + \ff{O\cap Z_1}\tag{submodularity}\\
        &\ge \ff{O}. \tag{submodularity}
    \end{align*}
\end{proof}

\subsection{Pseudocode of \grg (Alg.~\ref{alg:grg}) and its Analysis}\label{apx:grg}
\begin{algorithm}[h]
   \caption{An algorithm guided by an $(\alpha,\beta)$-guidance set $Z$ with $\oh{kn}$ queries}\label{alg:grg}
     \Proc{\grg($f, \iset, Z, t$)}{
     \textbf{Input:} oracle $f$, constraint $\iset$, guidance set $Z$, switch point $t$\;
     $A_{0} \gets k$ dummy elements
     \tcc*[r]{Equivalent to an empty set}
     \For{$i \gets 1$ to $ k$}
        { \textbf{if} $i \le t\cdot k$ \textbf{then}
          $M_i \gets \argmax_{M\subseteq \uni \setminus (A_{i-1}\cup Z), M \text{ is a basis}} \sum_{a\in M}\marge{a}{A_{i-1}}$\label{line:grg-candidate1}\;
          \textbf{else} $M_i \gets \argmax_{M\subseteq \uni \setminus A_{i-1}, M  \text{ is a basis}} \sum_{a\in M}\marge{a}{A_{i-1}}$\label{line:grg-candidate2}\;
        \uIf{$\iset$ represents the size constraint}{
            $a_i \gets$ randomly pick an element from $M_i$\;
            $A_{i} \gets A_{i-1} +a_i-e_0$
            \tcc*[r]{$e_0$ is the dummy element}
         }
         \uElseIf{$\iset$ represents the matroid constraint}{
            $\sigma_i\gets $ a bijection from $M_i$ to $A_{i-1}$,
            where $A_{i-1}+x-\func{\sigma_i}{x} \in \iset(\mtr), \forall x \in M_i$ \;
            $e_i \gets$ randomly pick an element from $M_i$\;
            $A_{i} \gets A_{i-1} + e_i-\sigma_i(e_i)$\;
         }
         }
     \textbf{return} $A_k$}
\end{algorithm}
In this section, we present the pseudocode for \grg as Alg.~\ref{alg:grg}.
Then, we provide the detailed proof of Lemma~\ref{lemma:grg3-size}
in Appendix~\ref{apx:grg-size}, which addresses size constraints.
Finally, we analyze the algorithm under matroid constraints and
provide the guarantees and its analysis in Appendix~\ref{apx:mc}.
\subsubsection{\grg under Size Constraints}\label{apx:grg-size}
In Sec.~\ref{sec:grg}, we introduce the intuition behind \grg under size constraint.
Below, we reiterate theoretical guarantees achieved by \grg
under size constraints and provide the detailed analysis. 
\lemmagrgsize*
We provide the recurrence of $\ff{(O\setminus Z)\cup A_i}$, $\ff{O\cup A_i}$
and $\ff{A_i}$ in Lemmata~\ref{lemma:grg1} and~\ref{lemma:grg2}
and their analysis below to help prove Lemma~\ref{lemma:grg3-size} under size constraint.
\begin{lemma}\label{lemma:grg1}
When $0 <i \le t\cdot k$, it holds that
\begin{align*}
&\ex{\ff{(O\setminus Z) \cup A_i}} \ge \left(1-\frac{1}{k}\right)\ex{\ff{(O\setminus Z)\cup A_{i-1}}}
+\frac{1}{k}\left[\ff{O\setminus Z} - \ff{O\cup Z}\right],\\
&\ex{f\left(O \cup A_i\right)} \ge \left(1-\frac{1}{k}\right)\ex{\ff{O\cup A_{i-1}}}
+\frac{1}{k}\left[\ff{O}-\ff{O\cup Z}\right].
\end{align*}
When $t\cdot k < i \le k$, it holds that
\begin{align*}
\ex{f\left(O \cup A_i\right)} \ge \left(1-\frac{1}{k}\right)\ex{\ff{O\cup A_{i-1}}}
\end{align*}
\end{lemma}
\begin{proof}
At iteration $i$, condition on a given $A_{i-1}$.
When $i \le tk$, $A_{i-1}\cap Z = \emptyset$ and $M_i$ is selected out of $A_{i-1}\cup Z$,
so
\begin{equation}
  (O \cup Z)\cap ((O\setminus Z) \cup A_{i-1}\cup M_i) = O\setminus Z
\end{equation}
\begin{equation}
  ((O \cup Z)\cap (O \cup A_{i-1}\cup M_i) = O.
\end{equation}
Then,
\begin{align*}
&\exc{\ff{(O\setminus Z)\cup A_i}}{A_{i-1}} = \frac{1}{k}\sum_{x\in M_i}\ff{(O\setminus Z)\cup A_{i-1}\cup \{x\}} \tag{selection of next element}\\
&\ge \frac{1}{k} \sbrk{(k-1)\ff{(O\setminus Z)\cup A_{i-1}} + \ff{(O\setminus Z) \cup A_{i-1}\cup M_i}} \tag{submodularity}\\
&\ge \frac{1}{k} \sbrk{(k-1)\ff{(O\setminus Z)\cup A_{i-1}} + \ff{O\setminus Z}+\ff{O\cup Z \cup A_{i-1}\cup M_i}-\ff{O \cup Z}} \tag{submodularity}\\
&\ge \frac{1}{k} \sbrk{(k-1)\ff{(O\setminus Z)\cup A_{i-1}} + \ff{O\setminus Z}-\ff{O \cup Z}} \tag{nonnegativity}\\
&\exc{\ff{O\cup A_i}}{A_{i-1}} = \frac{1}{k}\sum_{x\in M_i}\ff{O\cup A_{i-1}\cup \{x\}}\\
&\ge \frac{1}{k} \sbrk{(k-1)\ff{O\cup A_{i-1}} + \ff{O \cup A_{i-1}\cup M_i}} \tag{submodularity}\\
&\ge \frac{1}{k} \sbrk{(k-1)\ff{O\cup A_{i-1}} + \ff{O} +\ff{O\cup Z \cup A_{i-1}\cup M_i}- \ff{O \cup Z}} \tag{submodularity}\\
&\ge \frac{1}{k} \sbrk{(k-1)\ff{O\cup A_{i-1}} + \ff{O} - \ff{O \cup Z}} \tag{nonnegativity}
\end{align*}
When $i > tk$, it holds that
\begin{align*}
&\exc{\ff{O\cup A_i}}{A_{i-1}} = \frac{1}{k}\sum_{x\in M_i}\ff{O\cup A_{i-1}\cup \{x\}}\\
&\ge \frac{1}{k} \sbrk{(k-1)\ff{O\cup A_{i-1}} + \ff{O \cup A_{i-1}\cup M_i}} \tag{submodularity}\\
&\ge \brk{1-\frac{1}{k}} \ff{O\cup A_{i-1}} \tag{nonnegativity}
\end{align*}
By unconditioning $A_{i-1}$, the lemma is proved.
\end{proof}
\begin{lemma}\label{lemma:grg2}
When $0 <i \le t\cdot k$, it holds that
\[\ex{\ff{A_i}} -\ex{\ff{A_{i-1}}} \ge \frac{1}{k}\brk{\ex{\ff{(O\setminus Z)\cup A_{i-1}}}-\ex{\ff{A_{i-1}}}}.\]
When $t\cdot k < i \le k$, it holds that
\[\ex{\ff{A_i}} -\ex{\ff{A_{i-1}}} \ge \frac{1}{k}\brk{\ex{\ff{O\cup A_{i-1}}}-\ex{\ff{A_{i-1}}}}.\]
\end{lemma}
\begin{proof}
Given $A_{i-1}$ at iteration $i$.
When $i \le t\cdot k$, it holds that
\begin{align*}
&\exc{\ff{A_i}- \ff{A_{i-1}}}{A_{i-1}} = \frac{1}{k}\sum_{x \in M_i} \marge{x}{A_{i-1}}\\
&\ge \frac{1}{k}\sum_{x \in O\setminus \brk{A_{i-1} \cup Z}} \marge{x}{A_{i-1}}\tag{Line~\ref{line:grg-candidate1} in Alg.~\ref{alg:grg}}\\
&\ge \frac{1}{k}\sbrk{\ff{(O \setminus Z)\cup A_{i-1}}-\ff{A_{i-1}}}. \tag{submodularity}
\end{align*}
When $i > t\cdot k$, it holds that
\begin{align*}
&\exc{\ff{A_i}- \ff{A_{i-1}}}{A_{i-1}} = \frac{1}{k}\sum_{x \in M_i} \marge{x}{A_{i-1}}\\
&\ge \frac{1}{k}\sum_{x \in O\setminus A_{i-1} } \marge{x}{A_{i-1}}\tag{Line~\ref{line:grg-candidate2} in Alg.~\ref{alg:grg}}\\
&\ge \frac{1}{k}\sbrk{\ff{O \cup A_{i-1}}-\ff{A_{i-1}}}.\tag{submodularity}
\end{align*}
By unconditioning $A_{i-1}$, the lemma is proved.
\end{proof}

\begin{proof}[Proof of Lemma~\ref{lemma:grg3-size}]
It follows from Lemma~\ref{lemma:grg1} and the closed form
for a recurrence provided in Lemma~\ref{lemma:recurrence} that
\begin{equation}\label{inq:grg3-size-O}
\left\{
\begin{aligned}
&\ex{\ff{(O\setminus Z) \cup A_i}} \ge \ff{O\setminus Z} - \left(1-\left(1-\frac{1}{k}\right)^i\right)\ff{O\cup Z}, &\forall 0<i\le tk\\
&\ex{\ff{O \cup A_i}} \ge \left(1-\frac{1}{k}\right)^{i-\lfloor tk \rfloor}\left[\ff{O} - \left(1-\left(1-\frac{1}{k}\right)^{\lfloor tk \rfloor}\right)\ff{O\cup Z}\right], &\forall tk<i\le k\\
\end{aligned}
\right.
\end{equation}
With the above inequalities,
we can solve the recursion in Lemma~\ref{lemma:grg2} as follows,
\begin{align*}
&\ex{\ff{A_{\lfloor tk \rfloor}}}
\stackrel{(a)}{\ge} \left(1-\left(1-\frac{1}{k}\right)^{\lfloor tk \rfloor}\right)\ff{O\setminus Z}-
\left(1-\left(1-\frac{1}{k}\right)^{\lfloor tk \rfloor} - t\left(1-\frac{1}{k}\right)^{\lfloor tk \rfloor-1}\right)\ff{O\cup Z}\\
&\ex{\ff{A_{k}}}
\stackrel{(b)}{\ge} \left(1-\frac{1}{k}\right)^{k-\lfloor tk \rfloor}\ex{\ff{A_{\lfloor tk \rfloor}}}
+(1-t)\left(1-\frac{1}{k}\right)^{k-\lfloor tk \rfloor-1}\left[\ff{O}-\left(1-\left(1-\frac{1}{k}\right)^{\lfloor tk \rfloor}\right)\ff{O\cup Z}\right]\\
&\ge (1-t)\left(1-\frac{1}{k}\right)^{k-\lfloor tk \rfloor-1} \ff{O}
+\left(\left(1-\frac{1}{k}\right)^{k-\lfloor tk \rfloor}-\left(1-\frac{1}{k}\right)^{k}\right)\ff{O\setminus Z}\\
&\hspace{2em} - \left(\left(1+\frac{1-t}{1-\frac{1}{k}}\right)\left(1-\frac{1}{k}\right)^{k-\lfloor tk \rfloor}-\left(2-\frac{1}{k}\right)\left(1-\frac{1}{k}\right)^{k-1}\right)\ff{O\cup Z}\tag{Inequality $(a)$}\\
&\ge (1-t)\left(1-\frac{1}{k}\right)^{(1-t)k} \ff{O}
+\left(\left(1-\frac{1}{k}\right)^{(1-t)k+1}-\left(1-\frac{1}{k}\right)^{k}\right)\ff{O\setminus Z}\\
&\hspace{2em} - \left(\left(1+\frac{1-t}{1-\frac{1}{k}}\right)\left(1-\frac{1}{k}\right)^{(1-t)k}-\left(2-\frac{1}{k}\right)\left(1-\frac{1}{k}\right)^{k-1}\right)\ff{O\cup Z}\tag{$tk-1 < \lfloor tk \rfloor\le tk$}\\
&\ge (1-t)\left(1-\frac{1}{k}\right)e^{t-1}\ff{O}+\left(\left(1-\frac{1}{k}\right)^2e^{t-1}-e^{-1}\right)\ff{O\setminus Z}\\
&\hspace{2em} - \left(\left(1+\frac{1-t}{1-\frac{1}{k}}\right)e^{t-1}-\left(2-\frac{1}{k}\right)e^{-1}\right)\ff{O\cup Z}
\tag{nonnegativity;Lemma~\ref{lemma:val-inq}}\\
&\ge (1-t)\left(1-\frac{1}{k}\right)e^{t-1}\ff{O}+\left(\left(1-\frac{1}{k}\right)^2e^{t-1}-e^{-1}\right)\left(\ff{O}-\ff{O\cap Z}\right)\\
&\hspace{2em} - \left(\left(1+\frac{1-t}{1-\frac{1}{k}}\right)e^{t-1}-\left(2-\frac{1}{k}\right)e^{-1}\right)\ff{O\cup Z}
\tag{submodularity}\\
& = \left(\left(2-t-\frac{1}{k}\right)\left(1-\frac{1}{k}\right)e^{t-1}-e^{-1}\right)\ff{O}-\left(\left(1-\frac{1}{k}\right)^2e^{t-1}-e^{-1}\right)\ff{O\cap Z}\\
&\hspace{2em}- \left(\left(1+\frac{1-t}{1-\frac{1}{k}}\right)e^{t-1}-\left(2-\frac{1}{k}\right)e^{-1}\right)\ff{O\cup Z},
\end{align*}
where Inequality $(a)$ and $(b)$ follow from Inequality~\ref{inq:grg3-size-O},
Lemma ~\ref{lemma:grg2} and~\ref{lemma:recurrence}.
\end{proof}
\subsubsection{\grg under Matroid Constraints}\label{apx:mc}
\begin{algorithm}[h]\caption{\rg for Matroid}\label{alg:rgm} 
    \Proc{\rg$(f, \mtr)$}{
     \textbf{Input:} oracle $f$, matroid constraint $\mtr$\;
     \textbf{Initialize:} $A_0\gets$ arbitrary basis in $\iset(\mtr)$\;
     \For{$i \gets 1$ to $ k$}{
        $M_i\gets \argmax_{S\subseteq \uni, S \text{ is a basis}}\sum_{x\in S}\marge{x}{A_{i-1}}$\;
        $\sigma \gets$ a bijection from $M_i$ to $A_{i-1}$\;
        $x_i \gets$ a uniformly random element from $M_i$\;
        $A_i\gets A_{i-1}+x_i-\sigma(x_i)$\;
     }
     \textbf{return} $A_k$ \;}
 \end{algorithm}
\textbf{Discussion about Intuition behind \grg under Matroid Constraints.}
The pseudocode for \rg under matroid constraints
is provided in Alg.~\ref{alg:rgm}.
To deal with the feasibility for matroid constraints,
Alg.~\ref{alg:rgm} starts with an arbitrary basis
and builds the solution by randomly swapping the elements
in ground set with a candidate basis.
The analysis of it proceeds according to two main recurrences.
\begin{align*}
&\text{1)} \; \ex{\ff{A_i} -\ff{A_{i-1}}} \ge \frac{1}{k}\ex{\ff{O\cup A_{i-1}}-2\ff{A_{i-1}}}, \\
&\text{2)} \; \ex{\ff{O\cup A_i}}\ge \left(1-\frac{2}{k}\right)\ex{\ff{O\cup A_{i-1}}}+\frac{1}{k}\ex{\ff{O}+\ff{O\cup A_{i-1}\cup M_i}}.
\end{align*}
Fig.~\ref{fig:grgm-deg}(a) depicts the worse-case behavior
of $\ex{\ff{A_i}}$ and $\ex{\ff{O\cup A_i}}$.
As discussed in Section~\ref{sec:grg},
we consider improving the degradation of $\ex{\ff{O\cup A_i}}$
by selecting elements from outside of an $(\alpha+\epsi, \alpha)$-guidance set $Z$
to enhance the lower bound of $\ex{\ff{O\cup A_{i-1}\cup M_i}}$.
The blue line in Fig.~\ref{fig:grgm-deg}(b) illustrated 
the improvement of $\ex{\ff{O\cup A_i}}$ with an $(\alpha+\epsi, \alpha)$-guidance set.
However, restricting the selection only to elements outside $Z$
restricts the increase in $\ex{\ff{A_i}}$ to the difference between
$\ex{\ff{(O\setminus Z)\cup A_{i-1}}}$ and $\ex{\ff{A_{i-1}}}$.
This restriction is illustrated by
the red line in Figure~\ref{fig:grgm-deg}(b), indicating
degradation in $\ex{\ff{(O\setminus Z)\cup A_{i}}}$.

To benefit from the improved degradation of $\ex{\ff{O\cup A_i}}$,
we consider transitioning to selecting elements from the whole ground set at a suitable point.
The blue line in Fig.~\ref{fig:grgm-deg}(b) illustrates
how $\ex{\ff{O\cup A_i}}$ degrades before and after we switch,
and the orange line illustrates the evolution of $\ex{\ff{A_i}}$.
Even starting with an inferior selection at the first stage,
we still get an overall improvement on the objective value.

We provide the updated recursions for $\ff{(O\setminus Z)\cup A_i}$,
$\ff{O\cup A_i}$ and $\ff{A_i}$ in Lemma~\ref{lemma:grg1-m} and \ref{lemma:grg2-m} below.
Then, the closed form of the solution value, derived from these lemmata,
is presented in Lemma~\ref{lemma:grg3-matroid}.
After that, we prove the approximation ratio of the randomized algorithm
under matroid constraint.
\begin{figure}
    \centering
    \includegraphics[width=\linewidth]{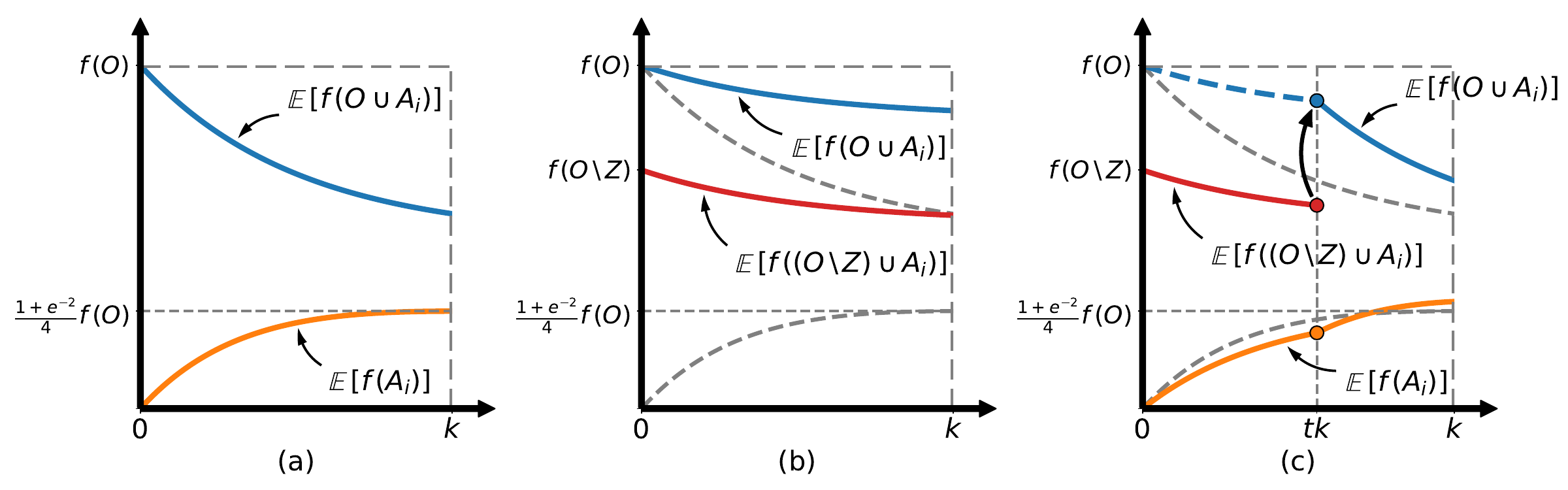}
    \vspace*{-1em}
    \caption{
    This set of figures indicates how guiding benefits \rg under matroid constraints.
    The figure (a) depicts the evolution of $\ff{O\cup A_i}$ and $\ff{A_i}$ with \rg.
    The figure (b) illustrates how the degradation of $\ff{O\cup A_i}$ changes as we introduce
    an $(0.305+\epsi, 0.305)$-guidance set. 
    Additionally, we also need to consider the degradation of $\ff{(O\setminus Z)\cup A_i}$,
    which is the value that the solution approaches with the guidance.
    The figure (c) shows the updated degradation with a switch point $tk$, where the algorithm
    starts with guidance and then switches to running without guidance.
    It demonstrates that even though the value of $A_i$ decreases initially
    when the selection starts outside of $Z$,
    it benefits from the improved degradation of $\ff{O\cup A_i}$
    upon switching back to the original algorithm.}
    \label{fig:grgm-deg}
    \vspace*{-1em}
\end{figure}
\begin{lemma}\label{lemma:grg1-m}
When $0 <i \le t\cdot k$, it holds that
\begin{align*}
&\ex{f\left((O\setminus Z) \cup A_i\right)} \ge \left(1-\frac{2}{k}\right)\ex{\ff{(O\setminus Z)\cup A_{i-1}}}
+\frac{1}{k}\left[2\ff{O\setminus Z} - \ff{O\cup Z}\right],\\
&\ex{f\left(O \cup A_i\right)} \ge \left(1-\frac{2}{k}\right)\ex{\ff{O\cup A_{i-1}}}
+\frac{1}{k}\left[2\ff{O}-\ff{O\cup Z}\right].
\end{align*}
When $t\cdot k < i \le k$, it holds that
\begin{align*}
\ex{f\left(O \cup A_i\right)} \ge \left(1-\frac{2}{k}\right)\ex{\ff{O\cup A_{i-1}}}+\frac{1}{k}\ff{O}
\end{align*}
\end{lemma}
\begin{proof}
When $i \le tk$, it holds that
\begin{align*}
&\exc{\ff{(O\setminus Z)\cup A_i}}{A_{i-1}} = \frac{1}{k}\sum_{x\in M_i}\ff{(O\setminus Z)\cup \brk{A_{i-1}+x-\sigma_i(x)}}\\
&\ge \frac{1}{k}\sum_{x\in M_i}\sbrk{\marge{x}{(O\setminus Z)\cup A_{i-1}}+\ff{(O\setminus Z)\cup \brk{A_{i-1}-\sigma_i(x)}}}\tag{submodularity}\\
&\ge \frac{1}{k} \sbrk{\ff{(O\setminus Z) \cup A_{i-1}\cup M_i}-\ff{(O\setminus Z)\cup A_{i-1}}+ (k-1)\ff{(O\setminus Z)\cup A_{i-1}} + \ff{O\setminus Z}} \tag{submodularity}\\
&\ge \frac{1}{k} \sbrk{(k-2)\ff{(O\setminus Z)\cup A_{i-1}} + 2\ff{O\setminus Z}-\ff{O \cup Z}} \tag{submodularity; nonnegativity}\\
&\exc{\ff{O\cup A_i}}{A_{i-1}} = \frac{1}{k}\sum_{x\in M_i}\ff{O\cup \brk{A_{i-1}+x-\sigma_i(x)}}\\
&\ge \frac{1}{k}\sum_{x\in M_i}\sbrk{\marge{x}{O\cup A_{i-1}}+\ff{O\cup \brk{A_{i-1}-\sigma_i(x)}}} \tag{submodularity}\\
&\ge \frac{1}{k} \sbrk{\ff{O \cup A_{i-1}\cup M_i} - \ff{O\cup A_{i-1}} + (k-1)\ff{O\cup A_{i-1}} +\ff{O}} \tag{submodularity}\\
&\ge \frac{1}{k} \sbrk{(k-2)\ff{O\cup A_{i-1}} + 2\ff{O} - \ff{O \cup Z}} \tag{submodularity; nonnegativity}
\end{align*}

When $i > tk$, it holds that
\begin{align*}
&\exc{\ff{O\cup A_i}}{A_{i-1}} = \frac{1}{k}\sum_{x\in M_i}\ff{O\cup \brk{A_{i-1}+x-\sigma_i(x)}}\\
&\ge \frac{1}{k}\sum_{x\in M_i}\sbrk{\marge{x}{O\cup A_{i-1}}+\ff{O\cup \brk{A_{i-1}-\sigma_i(x)}}} \tag{submodularity}\\
&\ge \frac{1}{k} \sbrk{\ff{O \cup A_{i-1}\cup M_i} - \ff{O\cup A_{i-1}} + (k-1)\ff{O\cup A_{i-1}} +\ff{O}} \tag{submodularity}\\
&\ge \frac{1}{k} \sbrk{(k-2)\ff{O\cup A_{i-1}} + \ff{O} } \tag{nonnegativity}
\end{align*}
By unconditioning $A_{i-1}$, the lemma is proved.
\end{proof}

\begin{lemma}\label{lemma:grg2-m}
When $0 <i \le t\cdot k$, it holds that
\[\ex{\ff{A_i}} -\ex{\ff{A_{i-1}}} \ge \frac{1}{k}\brk{\ex{\ff{(O\setminus Z)\cup A_{i-1}}}-2\ex{\ff{A_{i-1}}}}.\]
When $t\cdot k < i \le k$, it holds that
\[\ex{\ff{A_i}} -\ex{\ff{A_{i-1}}} \ge \frac{1}{k}\brk{\ex{\ff{O\cup A_{i-1}}}-2\ex{\ff{A_{i-1}}}}.\]
\end{lemma}
\begin{proof}
Given $A_{i-1}$ at iteration $i$.
When $i \le t\cdot k$, 
since $O$ is a base,
$O\setminus (A_{i-1}\cup Z)$ with dummy elements is also a base.
It holds that
\begin{align*}
&\exc{\ff{A_i}- \ff{A_{i-1}}}{A_{i-1}} = 
\frac{1}{k}\sum_{x \in M_i}\sbrk{\ff{A_{i-1}+x-\sigma_i(x)}-\ff{A_{i-1}}}\\
&\ge \frac{1}{k}\sum_{x \in M_i} \sbrk{\marge{x}{A_{i-1}} + \ff{A_{i-1}-\sigma_i(x)} - \ff{A_{i-1}}} \tag{submodularity}\\
&\ge \frac{1}{k}\sum_{x \in O\setminus (A_{i-1}\cup Z)} \marge{x}{A_{i-1}} + \frac{1}{k}\sum_{x \in M_i}\sbrk{\ff{A_{i-1}-\sigma_i(x)} - \ff{A_{i-1}}}\tag{Line~\ref{line:grg-candidate1} in Alg.~\ref{alg:grg}}\\
&\ge \frac{1}{k}\sbrk{\ff{(O \setminus Z)\cup A_{i-1}}-\ff{A_{i-1}}} - \frac{1}{k} \ff{A_{i-1}}. \tag{submodularity}
\end{align*}
When $i > t\cdot k$, it holds that
\begin{align*}
&\exc{\ff{A_i}- \ff{A_{i-1}}}{A_{i-1}} = 
\frac{1}{k}\sum_{x \in M_i}\sbrk{\ff{A_{i-1}+x-\sigma_i(x)}-\ff{A_{i-1}}}\\
&\ge \frac{1}{k}\sum_{x \in M_i} \sbrk{\marge{x}{A_{i-1}} + \ff{A_{i-1}-\sigma_i(x)} - \ff{A_{i-1}}} \tag{submodularity}\\
&\ge \frac{1}{k}\sum_{x \in O} \marge{x}{A_{i-1}} + \frac{1}{k}\sum_{x \in M_i}\sbrk{\ff{A_{i-1}-\sigma_i(x)} - \ff{A_{i-1}}}\tag{Line~\ref{line:grg-candidate2} in Alg.~\ref{alg:grg}}\\
&\ge \frac{1}{k}\sbrk{\ff{O\cup A_{i-1}}-\ff{A_{i-1}}} - \frac{1}{k} \ff{A_{i-1}}. \tag{submodularity}
\end{align*}
\end{proof}
\begin{restatable}{lemma}{lemmagrgmatroid}
\label{lemma:grg3-matroid}
With an input matroid constraint $\iset$ and 
a $((1+\epsi)\alpha, \alpha)$-guidance set $Z$,
\grg returns set $A_k$ with $\oh{kn}$ queries, \st
$\ex{\ff{A_k}}\ge\frac{1}{2} \left(\frac{1}{2}+\left(\frac{3}{2}-t-\frac{1}{k}\right)\brk{1-\frac{2}{k}}e^{2(t-1)}-e^{-2} - (1+\epsi)\alpha\left(\brk{1-\frac{2}{k}}^2e^{2(t-1)} - e^{-2}\right)\right.$
$\left. -\alpha\left(\brk{\frac{1}{2}+\frac{1-t}{1-\frac{2}{k}}}e^{2(t-1)}-\brk{\frac{3}{2}-\frac{1}{k}}e^{-2}\right)\right)\ff{O}.$
\end{restatable}
\begin{proof}
It follows from Lemma~\ref{lemma:grg1-m}
and the closed form
for a recurrence provided in Lemma~\ref{lemma:recurrence} that
\begin{equation}\label{inq:grg3-m-O}
\left\{
\begin{aligned}
&\ex{\ff{(O\setminus Z) \cup A_i}} \ge \ff{O\setminus Z} - \frac{1}{2}\left(1-\left(1-\frac{2}{k}\right)^i\right)\ff{O\cup Z}, &\forall 0<i\le tk\\
&\ex{\ff{O \cup A_i}} \ge \frac{1}{2}\left(1+\left(1-\frac{2}{k}\right)^{i-\lfloor tk \rfloor}\right)\ff{O} - \frac{1}{2}\left(\left(1-\frac{2}{k}\right)^{i-\lfloor tk \rfloor}-\left(1-\frac{2}{k}\right)^{i}\right)\ff{O\cup Z}, &\forall tk<i\le k\\
\end{aligned}
\right.
\end{equation}
Then, by sovling the recursion in Lemma~\ref{lemma:grg2-m} with the above inequalities, it holds that
\begin{align*}
&\ex{\ff{A_{\lfloor tk \rfloor}}}\stackrel{(a)}{\ge} \frac{1}{2}\left[1-\left(1-\frac{2}{k}\right)^{\lfloor tk \rfloor}\right]\ff{O\setminus Z} - \left[\frac{1}{4}-\frac{1}{4}\left(1-\frac{2}{k}\right)^{\lfloor tk \rfloor}-\frac{t}{2}\left(1-\frac{2}{k}\right)^{\lfloor tk \rfloor-1}\right]\ff{O\cup Z}\\
&\ex{\ff{A_k}} \stackrel{(b)}{\ge} \left(1-\frac{2}{k}\right)^{k-\lfloor tk \rfloor}\ex{\ff{A_{\lfloor tk \rfloor}}}
+\frac{1}{2}\left[\frac{1}{2}+\brk{\frac{1}{2}-t+\frac{1}{k}}\brk{1-\frac{2}{k}}^{k-\lfloor tk \rfloor-1}\right]f(O)\\
&\hspace{5em}-\frac{1-t}{2}\left[\brk{1-\frac{2}{k}}^{k-\lfloor tk \rfloor-1} 
- \brk{1-\frac{2}{k}}^{k-1}\right]\ff{O\cup Z}\\
&\ge \frac{1}{2} \left[\left(\brk{1-\frac{2}{k}}^{k-\lfloor tk \rfloor} - \brk{1-\frac{2}{k}}^k\right)\ff{O\setminus Z}
+\left(\frac{1}{2}+\brk{\frac{1}{2}-t+\frac{1}{k}}\brk{1-\frac{2}{k}}^{k-\lfloor tk \rfloor-1}\right)f(O)\right.\\
&\hspace{2em}\left.-\left(\brk{\frac{1}{2}+\frac{1-t}{1-\frac{2}{k}}}\brk{1-\frac{2}{k}}^{k-\lfloor tk \rfloor}-\brk{\frac{3}{2}-\frac{1}{k}}\brk{1-\frac{2}{k}}^{k-1}\right)\right]\ff{O\cup Z}\tag{Inequality $(a)$}\\
&\ge \frac{1}{2} \left[\left(\brk{1-\frac{2}{k}}^{(1-t)k+1} - \brk{1-\frac{2}{k}}^k\right)\ff{O\setminus Z}
+\left(\frac{1}{2}+\brk{\frac{1}{2}-t+\frac{1}{k}}\brk{1-\frac{2}{k}}^{(1-t)k}\right)f(O)\right.\\
&\hspace{2em}\left.-\left(\brk{\frac{1}{2}+\frac{1-t}{1-\frac{2}{k}}}\brk{1-\frac{2}{k}}^{(1-t)k}-\brk{\frac{3}{2}-\frac{1}{k}}\brk{1-\frac{2}{k}}^{k-1}\right)\right]\ff{O\cup Z}\tag{$tk-1 < \lfloor tk \rfloor\le tk$}\\
&\ge \frac{1}{2} \left[\left(\brk{1-\frac{2}{k}}^2e^{2(t-1)} - e^{-2}\right)\ff{O\setminus Z}
+\left(\frac{1}{2}+\brk{\frac{1}{2}-t+\frac{1}{k}}\brk{1-\frac{2}{k}}e^{2(t-1)}\right)f(O)\right.\\
&\hspace{2em}\left.-\left(\brk{\frac{1}{2}+\frac{1-t}{1-\frac{2}{k}}}e^{2(t-1)}-\brk{\frac{3}{2}-\frac{1}{k}}e^{-2}\right)\ff{O\cup Z}\right]
\tag{nonnegativity;Lemma~\ref{lemma:val-inq}}\\
&\ge\frac{1}{2}\left[\left(\frac{1}{2}+\left(\frac{3}{2}-t-\frac{1}{k}\right)\brk{1-\frac{2}{k}}e^{2(t-1)}-e^{-2}\right)\ff{O}-\left(\brk{1-\frac{2}{k}}^2e^{2(t-1)} - e^{-2}\right)\ff{O\cap Z}\right.\\
&\left.-\left(\brk{\frac{1}{2}+\frac{1-t}{1-\frac{2}{k}}}e^{2(t-1)}-\brk{\frac{3}{2}-\frac{1}{k}}e^{-2}\right)\ff{O\cup Z}\right], \tag{submodularity}
\end{align*}
where Inequality $(a)$ and $(b)$ follow from Inequality~\ref{inq:grg3-m-O},
Lemma ~\ref{lemma:grg2-m} and~\ref{lemma:recurrence}.
\end{proof}

\begin{proof}[Proof of Theorem~\ref{thm:irg} under matroid constraint]
  Let $(f,\iset)$ be an instance of \sm, with optimal solution set $O$.
If $\ff{Z}\ge (0.305-\epsi)\ff{O}$ under matroid constraint,
the approximation ratio holds immediately.
Otherwise, by Corollary~\ref{cor:fls}, 
\fls returns a set $Z$ which is an $((1+\epsi)\alpha, \alpha)$-guidance set,
where $\alpha = 0.305-\epsi$.

By Lemma~\ref{lemma:grg3-matroid} and $Z$ is an $((1+\epsi)\alpha, \alpha)$-guidance set
with $\alpha = 0.305-\epsi$,
\begin{align*}
&\ex{\ff{A_k}}
\ge \frac{1}{2}\left[\frac{1}{2}+\left(\frac{3}{2}-t-\epsi\right)(1-2\epsi)e^{2(t-1)}-e^{-2}
-(0.305-0.695\epsi)\left(\left(1-2\epsi\right)^2 e^{2(t-1)}-e^{-2}\right)\right.\\
&\hspace{2em}-\left.(0.305-\epsi)\left(\left(\frac{1}{2}+\frac{1-t}{1-2\epsi}\right)e^{2(t-1)}-\left(\frac{3}{2}-\epsi\right)e^{-2}\right)\right]\ff{O}
\tag{$\forall k \ge \frac{1}{\epsi}$}\\
&\ge (0.305-\epsi)\ff{O}. \tag{$t=0.559$}
\end{align*}
\end{proof}

\section{Analysis of Deterministic Approximation Algorithm}\label{apx:determ-full}
In this section, we present the pseudocode of deterministic algorithm and its analysis.
The organization of this section is as follows:
firstly, in Appendix~\ref{apx:ig}, we introduce the deterministic subroutine, 
\gig and \gigm, along with their analysis;
then, in Appendix~\ref{apx:determ}, we provide a randomized version of the deterministic algorithm
for analytical purposes;
finally, in Appendix~\ref{apx:derand}, we provide the deterministic algorithm and its theoretical guarantee.
\subsection{Deterministic Subroutines - \gig and \gigm}\label{apx:ig}
Inspired by \ig algorithm,
a subroutine of \itpg, proposed by~\citet{chen2023approximation},
we introduce guided versions of it for both size and matroid constraints.
The algorithm for the size constraint closely resembles
\ig in~\citet{chen2023approximation}.
Hence, we provide the pseudocode (Alg.~\ref{alg:gig}), guarantees, and analysis in Appendix~\ref{apx:ig-size}.
In this section, we focus on presenting \gigm for matroid constraints as Alg.~\ref{alg:gigm}.
This algorithm addresses the feasibility issue by incorporating \ig into matroid constraints.
Moreover, while it compromises the approximation ratio over size constraint
to some extent,
it no longer has the drawback of low success probability,
which the size-constrained version has.
      \begin{algorithm}[t]
    \caption{A guided \ig subroutine for matroid constraints.} \label{alg:gigm}
    \Proc{\gigm($f, \iset(\mtr), Z, G, \ell$)}{
        \textbf{Input:} oracle $f$, matroid constraint $\mtr$, guidance set $Z$,
        starting set $G$ , set size $\ell$\;
          \textbf{Initialize:} $A, A_1, \ldots, A_\ell\gets \emptyset$\;
          \For{$i \gets 1$ to $k$}{
                      $X_i \gets \{x \in \uni\setminus ( G\cup A\cup Z): A+x \in \iset(\mtr)\}$\;
                      $j_i^*, a_i^* \gets \argmax_{j \in [\ell], x\in X}\marge{x}{G\cup A_j}$\;
                      $A\gets A+a_i^*$, $A_{j_i^*}\gets A_{j_i^*}+a_i^*$\;}
          $\sigma\gets$ a bijection from $G$ to $A$ \st $(G\cup A_j)\setminus\left(\sum_{x\in A_j}\sigma^{-1}(x)\right)$ is a basis\;
        \textbf{return} $\left\{(G\cup A_j)\setminus\left(\sum_{x\in A_j}\sigma^{-1}(x)\right) :1 \le j \le \ell\right\}$}
\end{algorithm}
      

\textbf{Algorithm overview.}
Under size constraints, \itpg~\citep{chen2023approximation} 
constructs the solution with $\ell$ iterations,
where each iteration involves calling the \ig subroutine 
and adding $k/\ell$ elements
into the solution.
However, this approach is not applicable to  matroid constaint
due to the feasibility problem.
Consequently, we propose \gigm for matroid constraints
designed as follows: 
1) consider adding a basis ($k$ elements) $A$ to
$\ell$ solution sets,
where each addition dominates the gain of a distinct element in $O$;
2) by exchange property, establish a bijection between the basis $A$
and the starting set $G$;
3) delete elements in each solution set that are mapped to by the basis $A$.
This procedure avoids the extensive guessing
of \gig for size constraints
and reduces the number of potential solutions from $\ell(\ell+1)$ to $\ell$.
We provide the theoretical guarantees and 
the detailed analysis below.
\begin{restatable}{lemma}{thmigm}
\label{thm:ig-m}
  Let $O \in \iset(\mtr)$, and suppose
  $\gigm$(Alg.~\ref{alg:gigm}) is called with $(f, \mtr, Z, G, \ell)$,
  where $Z\cap G = \emptyset$.
  Then $\gigm$ outputs $\ell$ candidate sets
with $\oh{\ell kn}$ queries. 
Moreover, a randomly selected set $G'$ from the output satisfies that:
\begin{align*}
    &\text{1) }\ex{f(G')} \ge \brk{1-\frac{2}{\ell}} f(G)+
    \frac{1}{\ell+1} \brk{1-\frac{1}{\ell}}f((O\setminus Z) \cup G);\\
    &\text{2) } \ex{f(O \cup G')} \ge 
    \brk{1-\frac{2}{\ell}}f(O\cup G)+\frac{1}{\ell}\left(f(O)+\ff{O\cup (Z\cap G)}-f(O\cup Z)\right);\\
    &\text{3) }\ex{\ff{(O\setminus Z)\cup G'}} \ge  
    \brk{1-\frac{2}{\ell}}f((O\setminus Z)\cup G)+\frac{1}{\ell}\left(f(O\setminus Z)+\ff{(O\setminus Z)\cup (Z\cap G)}-f(O\cup Z)\right) .
\end{align*}
\end{restatable}
\begin{proof}
    $A = \{a_1^*, a_2^*, \ldots, a_k^*\}$
    be the sequence with the order of elements being added.
    Since $A$ and $O\setminus Z$ are basis of $\mtr$ (by adding dummy elements into $O\setminus Z$),
    we can order $O\setminus Z = \{o_1, o_2, \ldots, o_k\}$ 
    \st for any $1\le i\le k$, 
    $\{a_1^*,\ldots, a_{i-1}^*, o_i\}$ is an independent set.
    Thus, $o_i\in X_i$.
    Let $A_j^{(i)}$ be $A_j$ after $i$-th iteration.
    Therefore, for any $1\le j \le \ell$, by submodularity,
    \[\marge{o_i}{G\cup A_j}\le \marge{o_i}{G\cup A_j^{(i-1)}}
    \le\marge{a_i^*}{G\cup A_{j_i^*}^{(i-1)}}\]
    \[\Rightarrow \ff{(O\setminus Z)\cup G\cup A_j} -\ff{G\cup A_j}
    \le \sum_{i=1}^k \marge{o_i}{G\cup A_j}
    \le \sum_{i=1}^k\marge{a_i^*}{G\cup A_{j_i^*}^{(i-1)}}
    = \sum_{l=1}^\ell \marge{A_l}{G}\]
    By summing up the above inequality with $1\le j\le \ell$,
    \begin{align*}
    (\ell+1)\sum_{j=1}^\ell \marge{A_j}{G}
    &\ge \sum_{j=1}^\ell \ff{(O\setminus Z)\cup G\cup A_j}-\ell f(G)\\
    &\ge (\ell-1)\ff{(O\setminus Z)\cup G}+\ff{(O\setminus Z)\cup G \cup A} - \ell f(G)\\
    &\ge (\ell-1)\ff{(O\setminus Z)\cup G} - \ell f(G) \tag{nonnegativity}
    \end{align*}
    Then, we can prove the first inequality as follows,
    \begin{align*}
        \ex{f(G')-f(G)}&=\frac{1}{\ell}\sum_{j=1}^\ell
        \brk{\ff{G\setminus \sigma^{-1}(A_j)\cup A_j} - f(G)}\\
        &\ge \frac{1}{\ell}\sum_{j=1}^\ell
        \brk{\marge{A_j}{G}+\ff{G\setminus \sigma^{-1}(A_j)}-f(G)}\\
        &\ge \frac{1}{\ell+1}\brk{1-\frac{1}{\ell}}
            \ff{(O\setminus Z)\cup G} -\frac{1}{\ell+1}f(G)- \frac{1}{\ell} f(G)\\
        \Rightarrow \hspace{4em} \ex{f(G')} &\ge \brk{1-\frac{2}{\ell}}f(G) + \frac{1}{\ell+1}\brk{1-\frac{1}{\ell}}
        \ff{(O\setminus Z)\cup G}
    \end{align*}
    By submodularity, nonnegativity, and $Z\cap A = \emptyset$,
    the last two inequalities can be proved as follows,
    \begin{align*}
        \ex{\ff{O\cup G'}} &= \frac{1}{\ell}\sum_{j-1}^\ell
        \ff{O\cup \brk{G\setminus \sigma^{-1}(A_j)\cup A_j}}\\
        &\ge \frac{1}{\ell}\sum_{j-1}^\ell
        \left[\marge{A_j}{O\cup G} + \ff{O\cup \brk{G\setminus \sigma^{-1}(A_j)}}\right]\\
        &\ge \frac{1}{\ell} \brk{f(O\cup G\cup A)- f(O\cup G)
        +(\ell-1)f(O\cup G) +f(O)}\\
        &\ge \brk{1-\frac{2}{\ell}}f(O\cup G)+\frac{1}{\ell}\left(f(O)+\ff{O\cup (Z\cap G)}-f(O\cup Z)\right)
    \end{align*}
    \begin{align*}
        \ex{\ff{(O\setminus Z)\cup G'}} &= \frac{1}{\ell}\sum_{j-1}^\ell
        \ff{(O\setminus Z)\cup \brk{G\setminus \sigma^{-1}(A_j)\cup A_j}}\\
        &\ge \frac{1}{\ell}\sum_{j-1}^\ell
        \left[\marge{A_j}{(O\setminus Z)\cup G} + \ff{(O\setminus Z)\cup \brk{G\setminus \sigma^{-1}(A_j)}}\right]\\
        &\ge \frac{1}{\ell} \brk{f((O\setminus Z)\cup G\cup A)- f((O\setminus Z)\cup G)
        +(\ell-1)f((O\setminus Z)\cup G) +f((O\setminus Z))}\\
        &\ge \brk{1-\frac{2}{\ell}}f((O\setminus Z)\cup G)+\frac{1}{\ell}\left(f(O\setminus Z)+\ff{(O\setminus Z)\cup (Z\cap G)}-f(O\cup Z)\right)
    \end{align*}
\end{proof}

\subsubsection{\gig and its Analysis}\label{apx:ig-size}
In this section, we provide the pseudocode, guarantees and analysis of \gig,
which highly resembles \ig in~\citet{chen2023approximation}.
\begin{algorithm}[h]
    \caption{A $({\approx}\ell)$-approximation that interlaces $\ell$ greedy procedures together and uses only $1/\ell$ fraction of the budget.} \label{alg:gig}
    \Proc{\gig($f, k, Z, G, \ell$)}{
    \textbf{Input:} oracle $f$, constraint $k$, guidance set $Z$,
    starting set $G$ , set size $\ell$\;
      $\{a_1, \ldots , a_\ell\} \gets$ top $\ell$ elements 
      in $\mathcal{U} \setminus (G\cup Z)$ with respect to marginal gains on $G$ \label{line:ig-maxl}\;
      \For{$u \gets 0$ to $\ell$ in parallel\label{line:ig-outerfor}}{
      \eIf{ $u = 0 $}{
        $A_{u,l} \gets G\cup \{a_l\}$, for all $1\le l \le \ell$\;}{
      $A_{u,l} \gets G \cup \{a_u\}$, for any $1\le l \le \ell$\;}
    \For{$j \gets 1$ to $k/\ell-1$\label{line:ig_j}}{
            \For{$i \gets 1$ to $\ell$}{
              $x_{j,i} \gets \argmax_{
              x \in \mathcal{U} \setminus Z\setminus \brk{\cup_{l=1}^\ell A_{u,l}^j }} 
              \marge{x}{A_{u,i} }$\label{line:max}\;
              $A_{u,i} \gets A_{u,i} \cup \{x_{j,i}\}$\;}}}
    \textbf{return} $\{A_{u,i}:1 \le i \le \ell, 0 \le u \le \ell\}$}
\end{algorithm}
\begin{restatable}{lemma}{thmig}
    \label{thm:ig}
    Let $O \subseteq \uni$ be any set of size at most $k$, and suppose
  $\gig$(Alg.~\ref{alg:gig}) is called with $(f, k, Z, G, \ell)$.
  Then $\gig$ outputs $\ell(\ell+1)$ candidate sets
with $\oh{\ell kn}$ queries. 
Moreover, with a probability of $(\ell+1)^{-1}$, a randomly selected set $A$ from the output satisfies that:
\begin{align*}
    &\text{1) }\ex{\ff{A}}\ge \frac{1}{\ell+1}\ex{\ff{(O\setminus Z)\cup A}}+\frac{\ell}{\ell+1} \func{f}{G};\\
    &\text{2) } \ex{f(O \cup A)} \ge
        \left(1-\frac{1}{\ell}\right) f(O \cup G)
    +\frac{1}{\ell}\brk{\ff{O\cup (Z \cap G)}-\ff{O\cup Z}}\\
    &\text{3) }\ex{\ff{(O\setminus Z)\cup A}} \ge 
        \brk{1-\frac{1}{\ell}}\func{f}{(O\setminus Z) \cup G} + 
\frac{1}{\ell}(\ff{(O\setminus Z) \cup (Z \cap G)} - f(O\cup Z)).
\end{align*}
\end{restatable}
\begin{proof}
  Let $o_{\text{max}} = \argmax_{o \in O\setminus (G\cup Z)}\marge{o}{G}$,
  and let $\{ a_1, \ldots, a_\ell \}$ be the largest $\ell$ elements
  of $\{ \marge{x}{G} : x \in \uni \setminus (G\cup Z) \}$, as chosen on
  Line \ref{line:ig-maxl}. We consider the following two cases.

\textbf{Case  $(O\setminus (G\cup Z)) \cap \{a_1, \ldots, a_\ell\}=\emptyset$.}
Then, $o_{\text{max}} \not \in \{a_1, \ldots, a_\ell\}$ 
which implies that
$\marge{a_u}{G} \ge \marge{o}{G}$,
for every $1 \le u \le \ell$ and $o \in O\setminus (G\cup Z)$;
and, after the first iteration of the \textbf{for} loop on Line~\ref{line:ig_j} of Alg.~\ref{alg:gig},
none of the elements in $O \setminus (G\cup Z)$ is added into any of
$\{A_{0,i}\}_{i=1}^\ell$. We will analyze the iteration
of the \textbf{for} loop on Line \ref{line:ig-outerfor} with $u = 0$.

Since none of the elements in $O\setminus (G\cup Z)$ is added into the collection
when $j=0$, we can order $O\setminus (G\cup Z)= \{o_1, o_2,\ldots\}$ such that the first $\ell$ elements
are not selected in any set before we get to $j = 1$,
the next $\ell$ elements are not selected in any set before we get
to $j = 2$,
and so on. Let $i \in \{1, \ldots, \ell \}$.
Let $A_{0,i}^{j}$ be the value of $A_{0,i}$ after $j$ elements are added into it,
and define $A_{0,i} = A_{0,i}^{k/\ell}$, the final value.
Finally, denote by $\delta_j$ the value $\marge{x_{j, i}}{A_{0,i}^{j}}$.
Then, 
\begin{align*}
&\func{f}{(O\setminus Z) \cup A_{0,i}} - \func{f}{A_{0,i}}
\le \sum_{o \in O\setminus(A_{0,i}\cup Z)}\marge{o}{A_{0,i}}\tag{submodularity}\\
&\le \sum_{o \in O\setminus (G\cup Z)}\marge{o}{A_{0,i}}\tag{$G \subseteq A_{0,i}$}\\
&\le \sum_{l = 1}^{\ell} \marge{o_l}{A_{0,i}^{0}}
+ \sum_{l = \ell+1}^{2\ell} \marge{o_l}{A_{0,i}^{1}}+\ldots \tag{submodularity}\\
&\le \ell \sum_{j=1}^{k/\ell}\delta_j
=\ell (\func{f}{A_{0,i}}-\func{f}{G}),
\end{align*}
where the last inequality follows from the ordering of $O$ and the
selection of elements into the sets.
By summing up the above inequality with all $1\le i \le \ell$,
it holds that,
\begin{align*}
&\ex{f(A)} = \frac{1}{\ell} \sum_{i=1}^\ell \ff{A_{0, i}}
 \ge \frac{1}{\ell(\ell+1)}\sum_{i=1}^\ell f((O\setminus Z)\cup A_{0,i}) + \frac{\ell}{\ell+1} f(G)\\
&= \frac{1}{\ell+1}\ex{\ff{(O\setminus Z)\cup A}}+\frac{\ell}{\ell+1} \func{f}{G},
\end{align*}

Since $A_{0,i_1} \cap A_{0,i_2} = G$
for any $1 \le i_1\neq i_2 \le \ell$,
and each $x_{j, i}$ is selected outside of $Z$,
by repeated application of submodularity, it can be shown that
\begin{align*}
    &\hspace{-1em}\ex{\ff{(O\setminus Z)\cup A}}= \frac{1}{\ell}
    \sum_{i=1}^\ell f((O\setminus Z)\cup A_{0,i})\\
    &\ge \brk{1-\frac{1}{\ell}}\func{f}{(O\setminus Z) \cup G} + 
\frac{1}{\ell}\func{f}{(O\setminus Z)\cup \brk{\cup_{i=1}^\ell A_{0,i}}}\\
&\ge \brk{1-\frac{1}{\ell}}\func{f}{(O\setminus Z) \cup G} + 
\frac{1}{\ell}(\ff{(O\setminus Z) \cup (Z \cap G)} - f(O\cup Z))
\end{align*}
\begin{align*}
&\hspace{-1em} \ex{\ff{O\cup A}} = \frac{1}{\ell}\sum_{i=1}^\ell f(O\cup A_{0,i}) 
\ge \brk{1-\frac{1}{\ell}}\func{f}{O \cup G} + 
\frac{1}{\ell}\func{f}{O\cup \brk{\cup_{i=1}^\ell A_{0,i}}}\\
&\ge \brk{1-\frac{1}{\ell}}\func{f}{O \cup G}+
\frac{1}{\ell}\brk{\ff{O\cup (Z \cap G)}-\ff{O\cup Z}}. 
\end{align*}
Therefore, if we select a random set from $\{A_{0,i}: 1\le i \le \ell\}$,
the three inequalities in the Lemma hold
and we have probability $1/(\ell + 1 )$ of this happening.

\textbf{Case  $(O\setminus (G\cup Z)) \cap \{a_1, \ldots, a_\ell\}\neq\emptyset$.}
Then $o_{\text{max}} \in \{a_1, \ldots, a_\ell\}$,
so $a_u = o_{\text{max}}$, for some $u \in 1, \ldots, \ell$.
We analyze the iteration $u$ of the \textbf{for}
loop on Line \ref{line:ig-outerfor}.
Similar to the previous case,
let $i \in \{1, \ldots, \ell \}$, 
define $A_{u,i}^{j}$ be the value of $A_{u,i}$ after we add $j$ elements into it,
and we will use $A_{u,i}$ for $A_{u,i}^{k/\ell}$,
Also, let $\delta_j = \marge{x_{j, i}}{A_{u,i}^{j-1}}$.
Finally, let $x_{1,i} = a_u$ and 
observe $A_{u,i}^{1} = G \cup \{a_u\}$,
for any $i \in \{1, \ldots, \ell \}$.

Then, we can order $O\setminus G= \{o_1, o_2,\ldots\}$
such that: 1) for the first $\ell$ elements $\{o_l\}_{l=1}^\ell$, 
$\marge{o_l}{G} \le \marge{o_{\text{max}}}{G}=\delta_1$;
2) the next $\ell$ elements
$\{o_l\}_{l=\ell+1}^{2\ell}$ are not selected by any set before we get 
to $j = 2$, which implies that $\marge{o_l}{A_{u,i}^{1}} \le \delta_2$,
and so on.
Therefore, analagous to the the previous case,
we have that 
\begin{equation}\label{inq:ig-rec-2}
\ex{\ff{A}}\ge \frac{1}{\ell+1}\ex{\ff{(O\setminus Z)\cup A}}+\frac{\ell}{\ell+1} \func{f}{G}
\end{equation}
Since, $A_{u,i_1} \cap A_{u,i_2} = G\cup \{a_u\}$
for any $1 \le i_1\neq i_2 \le \ell$, $a_u \in O \setminus Z$,
and each $x_{j, i}$ is selected outside of $Z$,
by submodularity and nonnegativity of $f$, it holds that
\begin{align*}
    &\hspace{-1em}\ex{\ff{(O\setminus Z)\cup A}}= \frac{1}{\ell}
    \sum_{i=1}^\ell f((O\setminus Z)\cup A_{u,i})\\
    &\ge \brk{1-\frac{1}{\ell}}\func{f}{(O\setminus Z) \cup G} + 
\frac{1}{\ell}\func{f}{(O\setminus Z)\cup \brk{\cup_{i=1}^\ell A_{u,i}}}\\
&\ge \brk{1-\frac{1}{\ell}}\func{f}{(O\setminus Z) \cup G} + 
\frac{1}{\ell}(\ff{(O\setminus Z) \cup (Z \cap G)} - f(O\cup Z)) 
\end{align*}
\begin{align*}
&\hspace{-1em} \ex{\ff{O\cup A}} = \frac{1}{\ell}\sum_{i=1}^\ell f(O\cup A_{u,i}) 
\ge \brk{1-\frac{1}{\ell}}\func{f}{O \cup G} + 
\frac{1}{\ell}\func{f}{O\cup \brk{\cup_{i=1}^\ell A_{u,i}}}\\
&\ge \brk{1-\frac{1}{\ell}}\func{f}{O \cup G}+
\frac{1}{\ell}\brk{\ff{O\cup (Z \cap G)}-\ff{O\cup Z}}. 
\end{align*}
Therefore, if we select a random set from $\{A_{u,i}: 1\le i \le \ell\}$,
the three inequalities in the lemma holds, and
this happens with probability $(\ell+1)^{-1}$.
\end{proof}

\subsection{Randomized Version of our Deterministic Algorithm} \label{apx:determ}
In this section, we provide the randomized version (Alg~\ref{alg:determ}) of our deterministic algorithm (Alg.~\ref{alg:determ2}, provided in Appendix~\ref{apx:determ-full}).
The deterministic version simply evaluates all possible paths
and returns the best solution.
In the following, we provide the theoretical guarantee and
its analysis under different constraints.
\begin{restatable}{theorem}{thmdeterm}
\label{thm:determ}
Let $(f,\iset)$ be an instance of \sm, with the optimal solution set $O$.
Algorithm~\ref{alg:determ} achieves 
an expected $(0.385-\epsi)$ approximation ratio 
with $(\ell+1)^{-\ell}$ success probability and input $t=0.372$ under size constraint,
where $\ell = \frac{10}{9\epsi}$.
Moreover, it achieves an expected $(0.305-\epsi)$ approximation ratio 
with $t=0.559$ under matroid constraint.
The query complexity of the algorithm is $\oh{kn/\epsi}$.
\end{restatable}
\begin{algorithm}[t]
   \caption{The randomized algorithm suitable for derandomization.}\label{alg:determ}
          \textbf{Input:} oracle $f$, constraint $\iset$, an approximation result $Z_0$, switch point $t$, error rate $\epsi$\;
          $Z\gets \fls(f, \iset, Z_0, \epsi)$\;
          \textbf{Initialize} $\ell \gets \frac{10}{9\epsi},A_{0} \gets \emptyset$\;
            \eIf{$\iset$ is a size constraint}{
            \For{$i \gets 1$ to $ \ell$}{
                  \textbf{if} $i \le t\ell$ \textbf{then}
                  $A_i \gets$ a random set in $\gig(f,\iset,Z,A_{i-1},\ell)$\;
                \textbf{else}$A_i \gets$ a random set in $\gig(f,\iset,\emptyset,A_{i-1},\ell)$\;
                  }
            }{
            \For{$i \gets 1$ to $ \ell$}{
                  \textbf{if} $i \le t\ell$ \textbf{then}
                  $A_i \gets$ a random set in $\gigm(f,\iset,Z,A_{i-1},\ell)$\;
                    \textbf{else} $A_i \gets$ a random set in $\gigm(f,\iset,\emptyset,A_{i-1},\ell)$\;
                  }
            }
            
          \textbf{return} $A^*\gets \arg\max\{f(Z), f(A_\ell)\}$\;
\end{algorithm}
\subsubsection{Size constraints}
By Lemma~\ref{thm:ig} in Appendix~\ref{apx:ig-size}
and the closed form for a recurrence provided in Lemma~\ref{lemma:recurrence}, 
the following corollary holds,
\begin{corollary}\label{cor:gig}
    After iteration $i$ of the for loop in Alg.~\ref{alg:determ},
    the following inequalities hold with a probability of $(\ell+1)^{-i}$
    \begin{align*}
    &\ex{\ff{A_i}}\ge \frac{\ell}{\ell+1} \ex{\ff{A_{i-1}}}+
    \frac{1}{\ell+1}\left(\ff{O\setminus Z}-\brk{1-\brk{1-\frac{1}{\ell}}^i}\ff{O\cup Z}\right), && 1\le i \le t\ell\\
    &\ex{\ff{A_i}}\ge \frac{\ell}{\ell+1} \ex{\ff{A_{i-1}}}+
    \frac{1}{\ell+1}\brk{1-\frac{1}{\ell}}^{i-\lfloor t\ell \rfloor} \brk{f(O)-\brk{1-\brk{1-\frac{1}{\ell}}^{\lfloor t\ell \rfloor}}\ff{O\cup Z}}, && t\ell < i \le \ell
    \end{align*}
\end{corollary}
\begin{proof}[Proof of approximation ratio]
    If $\ff{Z}\ge (0.385-\epsi)\ff{O}$, the approximation ratio holds immediately. 
    So, we analyze the case $\ff{Z}< (0.385-\epsi)\ff{O}$ in the following.
    
    Recall in Corollary~\ref{cor:fls} that $Z$ is a $\left(1+\epsi)\alpha, \alpha\right)$-guidance set,
    it holds that $\ff{O\cup Z} + \ff{O \cap Z}< (0.77-1.615\epsi)\ff{O}$
    and $\ff{O \cap Z} < (0.385-0.615\epsi)\ff{O}$.

    By repeatedly implementing Lemma~\ref{lemma:recurrence}
    with the recursion in Corollary~\ref{cor:gig}, it holds that
    \begin{align*}
    &\ex{\ff{A_{\lfloor t\ell \rfloor}}}\ge \left(1-\left(1-\frac{1}{\ell+1}\right)^{\lfloor t\ell \rfloor}\right)\left(\ff{O\setminus Z}-\ff{O\cup Z}\right)+\frac{\lfloor t\ell \rfloor}{\ell+1}\left(1-\frac{1}{\ell}\right)^{\lfloor t\ell \rfloor}\ff{O\cup Z}\\
    &\ge \left(1-\left(1-\frac{1}{\ell+1}\right)^{\lfloor t\ell \rfloor}\right)\left(\ff{O}-\ff{O\cap Z}-\ff{O\cup Z}\right)+\frac{\lfloor t\ell \rfloor}{\ell+1}\left(1-\frac{1}{\ell}\right)^{\lfloor t\ell \rfloor}\ff{O\cup Z}\tag{submodularity}\\
    &\ex{\ff{A_{\ell}}}\ge\left(1-\frac{1}{\ell+1}\right)^{\ell-\lfloor t\ell \rfloor}\ex{\ff{A_{\lfloor t\ell \rfloor}}}+ \frac{\ell-\lfloor t\ell \rfloor}{\ell+1}\left(1-\frac{1}{\ell}\right)^{\ell-\lfloor t\ell \rfloor}\left(\ff{O}-\left(1-\left(1-\frac{1}{\ell}\right)^{\lfloor t\ell \rfloor}\right)\ff{O\cup Z}\right)\\
    &\ge \brk{\brk{1-\frac{1}{\ell+1}}^{(1-t) \ell+1}-\brk{1-\frac{1}{\ell+1}}^{\ell}}\brk{\ff{O} -\ff{O\cap Z}-\ff{O\cup Z}}\\
        &+\frac{t\ell}{\ell+1}\brk{1-\frac{1}{\ell+1}}^{(1-t)\ell+1}
        \brk{1-\frac{1}{\ell}}^{t\ell}\ff{O\cup Z}\\
        &+(1-t)\brk{1-\frac{1}{\ell+1}}\brk{1-\frac{1}{\ell}}^{(1-t)\ell+1}
        \brk{f(O)-\brk{1-\brk{1-\frac{1}{\ell}}^{t\ell-1}}\ff{O\cup Z}}\tag{$t\ell-1< \lfloor t\ell\rfloor \le t\ell$}\\
        &=\brk{\brk{1-\frac{1}{\ell+1}}\brk{1-\frac{1}{\ell+1}}^{(1-t) \ell}-\brk{1+\frac{1}{\ell}}\brk{1-\frac{1}{\ell+1}}^{\ell+1}}\brk{\ff{O} -\ff{O\cap Z}-\ff{O\cup Z}}\\
        &+t\cdot \frac{\ell-1}{\ell+1}\brk{1-\frac{1}{\ell+1}}\brk{1-\frac{1}{\ell+1}}^{(1-t)\ell}
        \brk{1-\frac{1}{\ell}}^{t\ell-1}\ff{O\cup Z}\\
        &+(1-t)\brk{1-\frac{1}{\ell+1}}^2
        \left[\brk{1-\frac{1}{\ell}}\brk{1-\frac{1}{\ell}}^{(1-t)\ell-1}f(O)\right.\\
        &\left.\hspace{11em} -\brk{1-\frac{1}{\ell}}^{(1-t)\ell}\ff{O\cup Z}
        +\brk{1-\frac{1}{\ell}}\brk{1-\frac{1}{\ell}}^{\ell-1}\ff{O\cup Z}\right]\\
        &\ge \brk{\brk{1-\frac{1}{\ell}}e^{t-1}-e^{-1}}
        \brk{\ff{O} -\ff{O\cap Z}-\ff{O\cup Z}}
        +t\cdot \brk{1-\frac{1}{\ell}}^3e^{-1}\ff{O\cup Z}\\
        &+(1-t)\brk{1-\frac{1}{\ell}}^2
        \brk{\brk{1-\frac{1}{\ell}}e^{t-1}f(O)-e^{t-1}f(O\cup Z)
        +\brk{1-\frac{1}{\ell}}e^{-1}f(O\cup Z)}\tag{$\ff{O}-\ff{O\cap Z}-\ff{O\cup Z} > 0$; nonnegativity; Lemma~\ref{lemma:val-inq}}\\
        &\ge \left[c (c^2(1-t)+1)e^{t-1}-e^{-1}\right]\ff{O}-\left[c(c(1-t)+1)e^{t-1}-(c^3+1)e^{-1}\right]\left(\ff{O\cup Z}+\ff{O\cap Z}\right)\\
        &\hspace{2em}-\left[c^3e^{-1}-c^2(1-t)e^{t-1}\right]\ff{O\cap Z}\tag{Let $c =1-\frac{9\epsi}{10}= 1-\frac{1}{\ell}$}\\
        &\ge \left[c (c^2(1-t)+1)e^{t-1}-e^{-1}\right]\ff{O}-\left[c(c(1-t)+1)e^{t-1}-(c^3+1)e^{-1}\right](0.77-1.615\epsi)\ff{O}\\
        &\hspace{2em}-\left[c^3e^{-1}-c^2(1-t)e^{t-1}\right](0.385-0.615\epsi)\ff{O}\tag{$\ff{O\cup Z}+\ff{O\cap Z} < (0.77-1.615\epsi)\ff{O}$; $\ff{O\cap Z} < (0.385-0.615\epsi)\ff{O}$}\\
        &\ge (0.385-\epsi)\ff{O} \tag{$0 < \epsi < 0.385$; $t = 0.372$}
    \end{align*}
\end{proof}

\subsubsection{Matroid Constraints}
By Lemma~\ref{thm:ig-m} in Appendix~\ref{apx:ig}
and the closed form for a recurrence provided in Lemma~\ref{lemma:recurrence}, 
the following corollary holds,
\begin{corollary}\label{cor:gig-m}
    After iteration $i$ of the for loop in Alg.~\ref{alg:determ},
    the following inequalities hold,
    \begin{align*}
        & \ex{\ff{A_i}}\ge \left(1-\frac{2}{\ell}\right)\ex{\ff{A_{i-1}}}\\
        &\hspace{1em}+\frac{1}{\ell+1}\left(1-\frac{1}{\ell}\right)\left(\ff{O\setminus Z} - \frac{1}{2}\left(1-\left(1-\frac{2}{\ell}\right)^{i-1}\right)\ff{O\cup Z}\right), 1\le i\le t\ell\\
        & \ex{\ff{A_i}}\ge \left(1-\frac{2}{\ell}\right)\ex{\ff{A_{i-1}}}\\
        &\hspace{1em}+\frac{1}{\ell+1}\left(1-\frac{1}{\ell}\right)\left(\frac{1}{2} \brk{1+\brk{1-\frac{2}{\ell}}^{i-t\ell}}f(O)
        -\frac{1}{2}\brk{\brk{1-\frac{2}{\ell}}^{i-t\ell}-\brk{1-\frac{2}{\ell}}^{i}}f(O\cup Z)\right),t\ell< i\le \ell
    \end{align*}
\end{corollary}
\begin{proof}[Proof of approximation ratio]
If $\ff{Z}\ge (0.305-\epsi)\ff{O}$, the approximation ratio holds immediately. 
    So, we analyze the case $\ff{Z}< (0.305-\epsi)\ff{O}$ in the following.
    
    Recall that $Z$ is a $\left(1+\epsi)\alpha, \alpha\right)$-guidance set in Corollary~\ref{cor:fls},
    it holds that $\ff{O\cup Z} + \ff{O \cap Z}< (0.61-1.695\epsi)\ff{O}$
    and $\ff{O \cap Z} < (0.305-0.695\epsi)\ff{O}$.

    By repeatedly implementing Lemma~\ref{lemma:recurrence}
    with the recursion in Corollary~\ref{cor:gig-m}, it holds that
    \begin{align*}
        &\ex{\ff{A_{\lfloor t\ell \rfloor}}}\ge \frac{\ell-1}{2(\ell+1)}\left[\left(1-\left(1-\frac{2}{\ell}\right)^{\lfloor t\ell \rfloor}\right)\left(\ff{O\setminus Z}-\frac{1}{2}\ff{O\cup Z}\right)+t\left(1-\frac{2}{\ell}\right)^{\lfloor t\ell \rfloor-1} \ff{O\cup Z}\right]\\
        &\hspace{2em}\ge\frac{\ell-1}{2(\ell+1)}\left[\left(1-\left(1-\frac{2}{\ell}\right)^{\lfloor t\ell \rfloor}\right)\left(\ff{O}-\ff{O\cap Z}-\frac{1}{2}\ff{O\cup Z}\right)+t\left(1-\frac{2}{\ell}\right)^{\lfloor t\ell \rfloor-1} \ff{O\cup Z}\right]\tag{submodularity}\\
        &\ex{\ff{A_{\ell}}}\ge
        \brk{1-\frac{2}{\ell}}^{\ell-\lfloor t\ell \rfloor}\ex{\ff{A_{\lfloor t\ell \rfloor}}}
        +\frac{\ell-1}{2(\ell+1)}\left[
        \brk{\frac{1}{2}+\brk{\frac{1}{2}-t+\frac{1}{\ell}}\brk{1-\frac{2}{\ell}}^{\ell-\lfloor t\ell \rfloor-1}}f(O)\right.\\
        &\hspace{2em}\left. -(1-t)\brk{\brk{1-\frac{2}{\ell}}^{\ell-\lfloor t\ell \rfloor-1}-\brk{1-\frac{2}{\ell}}^{\ell-1}}f(O\cup Z)\right]\\
        &\ge\frac{\ell-1}{2(\ell+1)}\left[\left(\frac{1}{2}+\left(\frac{3}{2}-t-\frac{1}{\ell}\right)\left(1-\frac{2}{\ell}\right)^{\ell- t\ell}-\left(1-\frac{2}{\ell}\right)^\ell\right)\ff{O}\right.\\
        &\hspace{2em}-\left(\left(1-\frac{2}{\ell}\right)^{(1-t)\ell}-\left(1-\frac{2}{\ell}\right)^\ell\right)\ff{O\cap Z}\\
        &\hspace{2em}\left. -\left(\left(\frac{1}{2}+\frac{1-t}{1-\frac{2}{\ell}}\right)\left(1-\frac{2}{\ell}\right)^{\ell- t\ell}-\left(\frac{3}{2}-\frac{1}{\ell}\right)\left(1-\frac{2}{\ell}\right)^{\ell-1}\right)\ff{O\cup Z}\right]\tag{$t\ell-1< \lfloor t\ell\rfloor \le t\ell$}\\
        &\ge \frac{\ell-1}{2(\ell+1)}\left[\left(\frac{1}{2}+\left(\frac{3}{2}-t-\frac{1}{\ell}\right)\brk{1-\frac{2}{\ell}}e^{2(t-1)}-e^{-2}\right)\ff{O}-\left(e^{2(t-1)}-\left(1-\frac{2}{\ell}\right)e^{-2}\right)\ff{O\cap Z}\right.\\
        &\hspace{2em}\left. -\left(\left(\frac{1}{2}+\frac{1-t}{1-\frac{2}{\ell}}\right)e^{2(t-1)}-\left(\frac{3}{2}-\frac{1}{\ell}\right)e^{-2}\right)\ff{O\cup Z}\right]\tag{Lemma~\ref{lemma:val-inq}; nonnegativity}\\
        &\ge \frac{\ell-1}{2(\ell+1)}\left[\left(\frac{1}{2}+\left(\frac{3}{2}-t-\frac{1}{\ell}\right)\brk{1-\frac{2}{\ell}}e^{2(t-1)}-e^{-2}\right)\ff{O}-\left(\left(\frac{1}{2}+\frac{1}{\ell}\right)e^{-2}-\left(\frac{1}{2}-t\right)e^{2(t-1)}\right)\ff{O\cap Z} \right.\\
        &\hspace{2em} \left. -\left(\left(\frac{1}{2}+\frac{1-t}{1-\frac{2}{\ell}}\right)e^{2(t-1)}-\left(\frac{3}{2}-\frac{1}{\ell}\right)e^{-2}\right)\left(\ff{O\cup Z}+\ff{O\cap Z}\right)\right]\\
        &\ge \frac{\ell-1}{2(\ell+1)}\left[\frac{1}{2}+\left(\frac{3}{2}-t-\frac{1}{\ell}\right)\brk{1-\frac{2}{\ell}}e^{2(t-1)}-e^{-2}-(0.305-0.695\epsi)\left(\left(\frac{1}{2}+\frac{1}{\ell}\right)e^{-2}-\left(\frac{1}{2}-t\right)e^{2(t-1)}\right)\right.\\
        &\hspace{2em}\left.-(0.61-1.695\epsi)\left(\left(\frac{1}{2}+\frac{1-t}{1-\frac{2}{\ell}}\right)e^{2(t-1)}-\left(\frac{3}{2}-\frac{1}{\ell}\right)e^{-2}\right)\right]\ff{O}\\
        & \ge (0.305-\epsi)\ff{O} \tag{$\ell = \frac{10}{9\epsi}$; $0 < \epsi < 0.305$; $t=0.559$}
    \end{align*}
\end{proof}

\subsection{Derandomize Alg.~\ref{alg:determ} in Section~\ref{sec:determ}}\label{apx:derand}
In this section, we present the deranmized version of Alg.~\ref{alg:determ},
which simply evaluates all possible paths and returns the best solution.
We reiterate the guarantee as follows.
\thmdgig*
\begin{algorithm}[h]
   \caption{Deterministic combinatorial approximation algorithm with the same ratio as Alg.~\ref{alg:irg}}\label{alg:determ2}
     \textbf{Input:} oracle $f$, constraint $\iset$, an approximation result $Z_0$,
     switch point $t$, error rate $\epsi$\;
     $Z\gets \fls(f, \iset, Z_0, \epsi)$\;
     \textbf{Initialize} $\ell \gets \frac{10}{9\epsi},G_{0} \gets \{\emptyset\}$\;
     \eIf{$\iset$ is a size constraint}{
     \For{$i \gets 1$ to $ \ell$}{
     $G_i \gets \emptyset$\;
        \For{$A_{i-1}\in G_{i-1}$}{
        \eIf{$i \le t\ell$}{
          $G_i \gets G_i \cup \gig(f,k,Z,A_{i-1},\ell)$\;}{
          $G_i \gets G_i \cup \gig(f,k,\emptyset,A_{i-1},\ell)$\;}
          }
        }}{
        \For{$i \gets 1$ to $ \ell$}{
     $G_i \gets \emptyset$\;
        \For{$A_{i-1}\in G_{i-1}$}{
        \eIf{$i \le t\ell$}{
          $G_i \gets G_i \cup \gigm(f,k,Z,A_{i-1},\ell)$\;}{
          $G_i \gets G_i \cup \gigm(f,k,\emptyset,A_{i-1},\ell)$\;}
          }
        }
        }
     \textbf{return} $A^*\gets \arg\max\{\ff{Z}, \ff{A_\ell}:A_\ell\in G_\ell\}$\;
\end{algorithm}


\section{Proofs for Nearly Linear Time Deterministic Algorithm in Section~\ref{sec:0.377}}\label{apx:0.377}
In this section, we provide the pseudocode and additional analysis of our nearly linear time deterministic algorithm introduced in Section~\ref{sec:0.377}.
We organize this section as follows:
in Appendix~\ref{apx:tgig},  we provide a speedup version of \gig
which queries $\oh{n\log(k)}$ times;
in Appendix~\ref{apx:prune}, we analyze the subroutine, \prune;
at last, in Appendix~\ref{apx:dgig}, we provide the pseudocode of
nearly linear time deterministic algorithm and its analysis.
\subsection{\gig Speedup}\label{apx:tgig}
In this section, we provide the algorithm \tgig,
which combines the guiding and descending threshold greedy techniques
with \ig~\citep{chen2023approximation}.
\begin{algorithm}[t]
\caption{A guided \ig subroutine with descending threshold technique for size constraints.}\label{alg:tgig}
\Proc{\tgig($f, k, Z, G, \ell, \epsi$)}{
    \textbf{Input:} oracle $f$, constraint $k$, 
    guidance set $Z$, starting set $G$, set size $\ell$, error $\epsi>0$\;
    $\{a_1, \ldots , a_\ell\} \gets$ top $\ell$ elements
    in $\uni \setminus (G\cup Z)$ with respect to marginal gains on $G$\label{line:tig-maxl}\; 
    \For{$u \gets 0$ to $\ell$ in parallel\label{line:tig-outerfor}}{ 
        \eIf{$u = 0$}{
            $M \gets \marge{a_\ell}{G}$\;
            $A_{u,l} \gets G\cup \{a_l\}$, for any $l \in [\ell]$ \label{line:A_ini}\;}{
                $M \gets \marge{a_u}{G}$\;
                $A_{u,l} \gets G \cup \{a_u\}$, for any $l \in [\ell]$ \label{line:A_ini_2}\;}
        $\tau_l \gets M$, $I_l\gets \textbf{true}$, for any $l \in [\ell]$\;
        \While{$\lor_{l=1}^\ell I_l $\label{line:tig_j}}{
            \For{$i \gets 1$ to $\ell$}{
                \If{$I_i$}{
                    $V\gets \uni\setminus \left(\cup_{l=1}^\ell A_{u,l}\cup Z\right)$\;
                     $A_{u,i}, \tau_i \gets \add(f, V, A_{u,i},\epsi, \tau_i, \frac{\epsi M}{k})$\;
                     \textbf{if} $|A_{u,i}\setminus G|=k/\ell \lor \tau_i < \frac{\epsi M}{k}$ \textbf{then} $I_i \gets \textbf{false}$\;
                      }}}
        $G_m \gets G_m \cup \{A_{u, 1}, A_{u, 2}, \ldots, A_{u, \ell}\}$\;}
    \textbf{return} $S^* \gets \argmax \{f(S_\ell): S_\ell \in G_\ell\}$\;}
\Proc{\add($f, V, A, \epsi, \tau, \tau_{\min}$)}{
    \textbf{Input:} oracle $f$, candidate set $V$, solution set $A$,
    $\epsi$, threshold $\tau$, and its lower bound $\tau_{\min}$\;
    \While{$\tau \ge \tau_{\text{min}}$}{
        \For{$x \in V$}{
            \textbf{if} $\marge{x}{A} \ge \tau$ \textbf{then return}
            $A \gets A \cup \{x\}$, $\tau$\;
        $\tau \gets (1-\epsi)\tau$\;}}
    \textbf{return} $A$, $\tau$\;}
\end{algorithm}
\begin{lemma}
\label{lemma:tgig}
    Let $O \subseteq \uni$ be any set of size at most $k$, and suppose
  $\tgig$(Alg.~\ref{alg:tgig}) is called with $(f, k, Z, G, \ell)$.
  Then $\tgig$ outputs $\ell(\ell+1)$ candidate sets
with $\oh{\ell^2n\log(k)/\epsi}$ queries. 
Moreover, with a probability of $(\ell+1)^{-1}$, a randomly selected set $A$ from the output satisfies that:
\begin{align*}
    &\text{1) }\left(\frac{\ell}{1-\epsi}+1\right)\ex{\ff{A}}
\ge \ex{\ff{(O\setminus Z) \cup A}}+
\frac{\ell}{1-\epsi}\func{f}{G}
-\epsi f(O);\\
    &\text{2) } \ex{f(O \cup A)} \ge
        \left(1-\frac{1}{\ell}\right) f(O \cup G)
    +\frac{1}{\ell}\brk{\ff{O\cup (Z \cap G)}-\ff{O\cup Z}}\\
    &\text{3) }\ex{\ff{(O\setminus Z)\cup A}} \ge 
        \brk{1-\frac{1}{\ell}}\func{f}{(O\setminus Z) \cup G} + 
\frac{1}{\ell}(\ff{(O\setminus Z) \cup (Z \cap G)} - f(O\cup Z)).
\end{align*}
\end{lemma}
\begin{proof}
    Let $o_{\text{max}} = \argmax_{o \in O\setminus (G\cup Z)}\marge{o}{G}$,
    and let $\{ a_1, \ldots, a_\ell \}$ be the largest $\ell$ elements
    of $\{ \marge{x}{G} : x \in \uni \setminus (G\cup Z) \}$, as chosen on
    Line \ref{line:tig_j}. We consider the following two cases.

    \textbf{Case  $(O\setminus (G\cup Z)) \cap \{a_1, \ldots, a_\ell\}=\emptyset$.}
    Then, $o_{\text{max}} \not \in \{a_1, \ldots, a_\ell\}$ 
    which implies that
    $\marge{a_u}{G} \ge \marge{o}{G}$,
    for every $1 \le u \le \ell$ and $o \in O\setminus (G\cup Z)$;
    and, after the first iteration of the \textbf{while} loop on Line~\ref{line:tig_j},
    none of the elements in $O \setminus (G\cup Z)$ is added into any of
    $\{A_{0,i}\}_{i=1}^\ell$. We will analyze the iteration
    of the \textbf{for} loop on Line \ref{line:tig-outerfor} with $u = 0$.

    For any $l \in [\ell]$, 
    let $A_{0,l}^{(j)}$ be $A_{0,l}$ after we add $j$ elements into it,
    $\tau_l^{(j)}$ be $\tau_l$ when we adopt $j$-th elements into $A_{0,l}$,
    and $\tau_l^{(1)} = M$.
    By Line~\ref{line:A_ini}, it holds that $A_{0,l}^{(1)}=G\cup \{a_l\}$.
    Since $(O\setminus (G\cup Z)) \cup \{a_1, \ldots, a_\ell\}=\emptyset$,
    and we add elements to each set in turn,
    we can order $O\setminus (G\cup Z) = \{o_1, o_2, \ldots\}$ 
    such that the first $\ell$ elements are not selected by any set 
    before we get $A_{0,l}^{(1)}$,
    the next $\ell$ elements are not selected in any set before we get $A_{0,l}^{(2)}$,
    and so on.
    Therefore, 
    for any $j \le |A_{0,l} \setminus G|$ 
    and $\ell (j-1)+1 \le i \le \ell j$,
    $o_i$ are filtered out by $A_{0,l}$
    with threshold $\tau_l^{(j)}/(1-\epsi)$, 
    which follows that 
    $\marge{o_i}{A_{0,l}^{(j)}} < \tau_l^{(j)}/(1-\epsi) \le (f(A_{0,l}^{(j)})- f(A_{0,l}^{(j-1)}))/(1-\epsi)$;
    for any $\ell |A_{0,l} \setminus G| < i \le |O \setminus (G\cup Z)|$,
    $o_i$ are filtered out by $A_{0,l}$
    with threshold $\frac{\epsi M}{k}$, 
    which follows that 
    $\marge{o_i}{A_{0,l}} < \epsi M/k$.
    Thus,
    \begin{align*}
    &\func{f}{(O\setminus Z) \cup A_{0,l}} - \func{f}{A_{0,l}}
    \le \sum_{o \in O\setminus A_{0,l}}\marge{o}{A_{0,l}}\tag{submodularity}\\
    &\le \sum_{o \in O\setminus (G\cup Z)}\marge{o}{A_{0,l}}\tag{$G \subseteq A_{0,l}$}\\
    &\le \sum_{i = 1}^{\ell} \marge{o_i}{A_{0,l}^{(1)}}
    + \sum_{i = \ell+1}^{2\ell} \marge{o_i}{A_{0,l}^{(2)}}+\ldots 
    + \sum_{i > \ell |A_{0,l} \setminus G|} \marge{o_i}{A_{0,l}}
    \tag{submodularity}\\
    &\le \ell \cdot \frac{\func{f}{A_{0,l}}-\func{f}{G}}{1-\epsi}
    +\epsi M\\
    &\le \ell \cdot \frac{\func{f}{A_{0,l}}-\func{f}{G}}{1-\epsi}
    +\epsi\ff{O}.
    \end{align*}
    By summing up the above inequality with all $1\le l \le \ell$,
    it holds that 
    \begin{equation}\label{inq:tig-rec}
    \left(\frac{\ell}{1-\epsi}+1\right)\ex{\ff{A}} \ge \ex{\ff{(O\setminus Z) \cup A}}+\frac{\ell}{1-\epsi}\ff{G}-\epsi\ff{O}.
    \end{equation}

    Since $A_{u,l_1} \cap A_{u,l_2} = G$
    for any $1 \le l_1 \neq l_2 \le \ell$
    it holds that 
    $\brk{(O\setminus Z) \cup A_{u,l_1}} \cap \brk{(O\setminus Z) \cup A_{u,l_2}} = (O\setminus Z)\cup G$.
    By repeated application of submodularity, it holds that
    \begin{align*}
    &\hspace{-1em}\ex{\ff{(O\setminus Z)\cup A}}= \frac{1}{\ell}
    \sum_{i=1}^\ell f((O\setminus Z)\cup A_{0,i})\\
    &\ge \brk{1-\frac{1}{\ell}}\func{f}{(O\setminus Z) \cup G} + 
    \frac{1}{\ell}\func{f}{(O\setminus Z)\cup \brk{\cup_{i=1}^\ell A_{0,i}}}\\
    &\ge \brk{1-\frac{1}{\ell}}\func{f}{(O\setminus Z) \cup G} + 
    \frac{1}{\ell}(\ff{(O\setminus Z) \cup (Z \cap G)} - f(O\cup Z))
    \end{align*}
    \begin{align*}
    &\hspace{-1em} \ex{\ff{O\cup A}} = \frac{1}{\ell}\sum_{i=1}^\ell f(O\cup A_{0,i}) 
    \ge \brk{1-\frac{1}{\ell}}\func{f}{O \cup G} + 
    \frac{1}{\ell}\func{f}{O\cup \brk{\cup_{i=1}^\ell A_{0,i}}}\\
    &\ge \brk{1-\frac{1}{\ell}}\func{f}{O \cup G}+
    \frac{1}{\ell}\brk{\ff{O\cup (Z \cap G)}-\ff{O\cup Z}}. 
    \end{align*}
    Therefore, the last two inequalities in the theorem hold.

    \textbf{Case  $(O\setminus (G\cup Z)) \cap \{a_1, \ldots, a_\ell\}\neq\emptyset$.}
    Then $o_{\text{max}} \in \{a_1, \ldots, a_\ell\}$.
    Suppose that $a_u = o_{\text{max}}$.
    We analyze that lemma holds with sets $\{A_{u,l}\}_{l=1}^\ell$.

    Similar to the analysis of the previous case, 
    let $A_{0,l}^{(j)}$ be $A_{0,l}$ after we add $j$ elements into it,
    $\tau_l^{(j)}$ be $\tau_l$ when we adopt $j$-th elements into $A_{0,l}$,
    and $\tau_l^{(1)} = M$.
    By Line~\ref{line:A_ini_2}, it holds that $A_{u,l}^{(1)}=G\cup \{a_u\}$.
    Then, we can order $O\setminus \left(G \cup Z \cup \{a_u\}\right) = \{o_1, o_2, \ldots\}$ 
    such that the first $\ell$ elements are not selected by any set 
    before we get $A_{u,l}^{(1)}$,
    the next $\ell$ elements are not selected in any set before we get $A_{u,l}^{(2)}$,
    and so on.
    Therefore, Inequality~\ref{inq:tig-rec} also holds in this case,
    where $A$ is a random set from
    $\{A_{u,l}\}_{l=1}^\ell$.

    Since, $a_u \in O$, and
    $A_{u,l_1} \cap A_{u,l_2} = G\cup \{a_u\}$
    for any $1 \le l_1\neq l_2 \le \ell$,
    it holds that 
    $\brk{(O\setminus Z) \cup A_{u,i_1}} \cap \brk{(O\setminus Z) \cup A_{u,i_2}} = O \cup G$.
    Following the proof of case $(O\setminus (G\cup Z)) \cap \{a_1,\ldots,a_\ell\}=\emptyset$, the last two inequalities still hold in this case.

    Overall, since either one of the above cases happens,
    the theorem holds.
\end{proof}

\subsection{Pruning Subroutine}\label{apx:prune}
\prune is used in Alg.~\ref{alg:dgig-2} to help construct the $(\alpha, \beta)$-guidance set,
where $\alpha = 0.377$ and $\beta = 0.46$.
In this section, we provide the pseudocode, guarantee and its analysis.
\begin{algorithm}[h]
    \caption{A pruning algorithm which deletes element with negative marginal gain} \label{alg:prune}
    \Proc{\prune($f, \mathcal A$)}{
        \textbf{Input:} oracle $f$, a sequence of subsets $\mathcal A$\;
        Initialize: $\mathcal A' \gets \emptyset$\;
        \For{$A \in \mathcal A $}{
            \For{$x \in A$}{
                \textbf{if} $\marge{x}{A-x} < 0$
                \textbf{then} $A \gets A-x$ \label{line:prune-del}\;}
            $\mathcal A' \gets \mathcal A' \cup \{A\}$}
         \textbf{return} $\mathcal A'$\;}
\end{algorithm}
\begin{lemma}
    \label{lemma:prune}
    Suppose $\prune$(Alg.~\ref{alg:prune}) is called with $(f, \mathcal A)$
    and returns the set $\mathcal A'$.
    For every $A \in \mathcal A$ and its related output set $A' \in \mathcal A'$,
    it holds that,
\begin{align*}
    &\text{1) }\ff{A'}\ge \ff{A};\\
    &\text{2) } \ff{S\cup A'} \ge \ff{S\cup A}, \forall S\subseteq \uni;\\
    &\text{3) }\ff{T} \le \ff{A'}, \forall T \subseteq A'.
\end{align*}
\end{lemma}
\begin{proof}
    Let $A_i$ be $A$ after we delete $i$-th element $x_i$,
    $A_0$ be the input set $A$, $A_m$ be the output set $A'$.
    Since any element $x_i$ being deleted follows that $\marge{x_i}{A_i} < 0$,
    it holds that $\ff{A_i}>\ff{A_{i}+x_i} = \ff{A_{i-1}}$.
    Therefore, $\ff{A'} = \ff{A_m}> \ldots > \ff{A_0} = \ff{A}$.
    The first inequality holds.

    For any $x_i \in A\setminus A'$, it holds that $\marge{x_i}{A_i}<0$.
    By submodularity,
    \[\ff{S\cup A}-\ff{S\cup A'} = \sum_{i=1}^m \marge{x_i}{S\cup A_i}
    \le \sum_{i=1}^m \marge{x_i}{A_i} < 0.\] 
    The second inequality holds.

    For any $y \in A'\setminus T$, since it is not deleted, 
    there exists $0\le i_y \le m$ such that $\marge{y}{A_{i_y}}\ge 0$.
    By submodularity,
    \[\ff{A'}-\ff{T}\ge \sum_{y\in A'\setminus T}\marge{y}{A'} 
    \ge \sum_{y\in A'\setminus T}\marge{y}{A_{i_y}}\ge 0.\]
    The third inequality holds.
\end{proof}
\subsection{Proofs for Theorem~\ref{thm:dgig-2} of Alg.~\ref{alg:dgig-2}}\label{apx:dgig}
\begin{algorithm}[t]
\caption{Nearly linear-time deterministic algorithm for size constraint}\label{alg:dgig-2}
    \textbf{Input:} oracle $f$, size constraint $k$, error rate $\epsi$\;
    \textbf{Initialize} $\epsi' = \frac{2}{\epsi}, \ell_1 \gets \frac{5}{3\epsi'},\ell_2 \gets \frac{5}{2\epsi'},t\gets 0.3$\;
    $\triangleright$ \textit{Create guided sets}\;
    $Z_{0} \gets \emptyset$\label{line:0.377-find-start}\;
    \For{$i \gets 1$ to $ \ell_1$}{
        $Z_i \gets \emptyset$\;     
        \For{$A_{i-1}\in Z_{i-1}$}{
        $Z_i \gets Z_i \cup \prune(\tgig(f,k,\emptyset,A_{i-1},\ell_1, \epsi'))$\;}}
    $Z \gets \cup_{i=1}^{\ell_1} Z_i$ \label{line:0.377-find-end} \;
    $\triangleright$ \textit{Build solution based on guided sets}\;
    $G \gets \emptyset$\;
    \For{$A \in Z$}{
        $G_0\gets \emptyset $\label{line:0.377-guide-start}\;
        \For{$i \gets 1$ to $\ell_2$}{
            $G_i \gets \emptyset$\;
            \For{$B_{i-1} \in G_{i-1}$}{
            \eIf{$i \le t\ell$}{
                $G_i \gets G_i \cup \tgig(f,k,A,B_{i-1},\ell_2, \epsi')$\;}{
                $G_i \gets G_i \cup \tgig(f,k,\emptyset,B_{i-1},\ell_2,\epsi')$\;}
                }}
        $G\gets G\cup G_{\ell_2}$\label{line:0.377-guide-end}\;}
    \textbf{return} $C^*\gets \arg\max_{C\in Z\cup G}f(C)$\;
\end{algorithm}
Prior to delving into the proof of Theorem~\ref{thm:dgig-2},
we provide the following corollary first.
It demonstrates the progression of the intermediate solution of \tgig,
after the pruning process, relying on
Lemma~\ref{lemma:prune} and~\ref{lemma:tgig}.
\begin{corollary}\label{cor:tgig-prune}
    Let $O \subseteq \uni$ be any set of size at most $k$.
  Then $\prune(\tgig(f, k, \emptyset, G, \ell))$ outputs $\ell(\ell+1)$ candidate sets
with $\oh{\ell^2n\log(k)/\epsi}$ queries. 
Moreover, with a probability of $(\ell+1)^{-1}$, a randomly selected set $A$ from the output satisfies that:
\begin{align*}
    &\text{1) }\left(\frac{\ell}{1-\epsi}+1\right)\ex{\ff{A}}
\ge \left(1-\frac{1}{\ell}\right) f(O \cup G)+
\frac{\ell}{1-\epsi}\func{f}{G}
-\epsi f(O);\\
    &\text{2) } \ex{f(O \cup A)} \ge
        \left(1-\frac{1}{\ell}\right) f(O \cup G);\\
        &\text{3) }f(O\cap A)\le f(A).
\end{align*}
\end{corollary}
\thmdgigtwo*
\begin{proof}
Following the proof of Theorem~\ref{thm:dgig-2}, 
we consider two cases of the algorithm.

\textbf{Case 1.} For every $A\in Z$, it holds that $f(O\cup A) \ge 0.46f(O)$.
    Then, we prove that $\max_{C\in Z}\ff{C}\ge (0.377-\epsi)\ff{O}$.

    In the following, we prove the theorem by analyzing the random case of the algorithm, where we randomly select a set from the output of \prune(\tgig).
    Suppose that we successfullly select a set where the inequalities in Corollary~\ref{cor:tgig-prune} hold.
    Let $A_i$ and $A_{i-1}$ be random sets in $Z_i$ and $Z_{i-1}$, respectively.
    Let $\brk{1-\frac{1}{\ell_1}}^{i^*-1} > 0.46 \ge \brk{1-\frac{1}{\ell_1}}^{i^*}$.
    Then, by Inequality (2) in Corollary~\ref{cor:tgig-prune},
    when $i < i^*$, $\ex{f(O\cup A_i)} \ge \brk{1-\frac{1}{\ell_1}}^{i}f(O)$;
    when $i \ge i^*$, $f(O\cup A_i)\ge 0.46 f(O)$ by assumption.
    By applying Inequality (1) in Corollary~\ref{cor:tgig-prune},
    \begin{align*}
        &\ex{\ff{A_{i^*}}} \ge \left[\frac{i^*}{\frac{\ell_1}{1-\epsi'}+1}\left(1-\frac{1}{\ell_1}\right)^{i^*}-\epsi'\left(1-\left(1-\frac{1}{\frac{\ell_1}{1-\epsi'}+1}\right)^{i^*}\right)\right]\ff{O}\\
        &\ex{\ff{A_{\ell_1}}} 
        \ge \brk{1-\frac{1}{\frac{\ell_1}{1-\epsi'}+1}}^{\ell_1-i^*}\ex{\ff{A_{i^*}}}+ \brk{ 1-\brk{1-\frac{1}{\frac{\ell_1}{1-\epsi'}+1}}^{\ell_1-i^*}}(0.46-\epsi')f(O)\\
        &\ge \left[ \frac{i^*}{\frac{\ell_1}{1-\epsi'}+1}\brk{1-\frac{1}{\ell_1}}^{i^*}\brk{1-\frac{1}{\frac{\ell_1}{1-\epsi'}+1}}^{\ell_1-i^*} +  \brk{ 1-\brk{1-\frac{1}{\frac{\ell_1}{1-\epsi'}+1}}^{\ell_1-i^*}}0.46\right.\\
        &\hspace{2em}\left. -\epsi'\left(1-\left(1-\frac{1}{\frac{\ell_1}{1-\epsi'}+1}\right)^{\ell_1}\right)\right] f(O)\\
        &\ge \left[\frac{\log(0.46)}{\left(\frac{\ell_1}{1-\epsi'}+1\right)\log\left(1-\frac{1}{\ell_1}\right)}\brk{1-\frac{1}{\ell_1}}e^{-1}+ \brk{1-\frac{e^{\epsi'-1}}{0.46\brk{1-\frac{1}{\frac{\ell_1}{1-\epsi'}+1}}}} 0.46-\epsi'\left(1-e^{-1+\epsi'}\right)\right]f(O) 
        \tag{Lemma~\ref{lemma:val-inq}; $\brk{1-\frac{1}{\ell_1}}^{i^*-1} > 0.46\ge \brk{1-\frac{1}{\ell_1}}^{i^*}$}\\
        &\ge (0.377-\epsi)\ff{O}\tag{$\ell_1 = \frac{5}{3\epsi'};\epsi'=\frac{\epsi}{2};0<\epsi<0.377$}
    \end{align*}
    Since we return the best solution in $Z$ and $G$, it holds that $f(C^*) \ge \ex{\ff{A_{\ell_1}}} \ge (0.377-\epsi)f(O)$.

    \textbf{Case 2.} There exists $A\in Z$, such that $f(O \cup A) < 0.46 f(O)$.
    Then, we prove that $\max_{C\in G}\ff{C}\ge (0.377-\epsi)\ff{O}$.

    Suppose that $f(A) < 0.377 f(O)$. 
    Otherwise, $f(C^*) \ge 0.377 f(O)$ immediately.
    By Inequality (3) in Corollary~\ref{lemma:prune},
    it holds that $\ff{O\cap A}\le \ff{A}< 0.377 \ff{O}$.
    Let $B_{\ell_2}$ be a randomly selected set in $G_{\ell_2}$,
    where we calculate $G_{\ell_2}$ with the guidance set $A$ and
    inequalities in Theorem~\ref{thm:ig} hold successfully.
    In the following, we also consider that randomized version of the algorithm,
    where we randomly select a set $B_i$ from all the solution set returned by \tgig.
    Suppose that we successfully select a set where the inequalities in 
    Lemma~\ref{lemma:tgig} hold.
    Then, the recursion of $\ex{\ff{B_i}}$ can be calculated as follows,
    \begin{align*}
    &\ex{\ff{B_i}} \ge \frac{\frac{\ell_2}{1-\epsi'}}{\frac{\ell_2}{1-\epsi'}+1}\ex{\ff{B_{i-1}}}+\frac{1}{\frac{\ell_2}{1-\epsi'}+1}\left(\ff{O\setminus A} - \left(1-\left(1-\frac{1}{\ell_2}\right)^i\right)\ff{O\cup A}-\epsi'\ff{O}\right),\\
    &\hspace{36em}1\le i \le t\ell_2\\
    &\ex{\ff{B_i}} \ge\frac{\frac{\ell_2}{1-\epsi'}}{\frac{\ell_2}{1-\epsi'}+1}\ex{\ff{B_{i-1}}}+\frac{1}{\frac{\ell_2}{1-\epsi'}+1}\left[\left(\left(1-\frac{1}{\ell_2}\right)^{i-\lfloor t\ell_2 \rfloor} -\epsi'\right)\ff{O}\right.\\
    &\hspace{4em} \left. - \left(\left(1-\frac{1}{\ell_2}\right)^{i-\lfloor t\ell_2 \rfloor}-\left(1-\frac{1}{\ell_2}\right)^i\right)\ff{O\cup A}\right] ,\hspace{12em}t\ell_2 < i\le \ell_2
    \end{align*}
    Then, by solving the above recursion, it holds that
    \begin{align*}
        &\ex{\ff{B_{\lfloor t\ell_2 \rfloor}}}\ge \left(1-\left(1-\frac{1}{\ell_2}\right)^{\lfloor t\ell_2 \rfloor}\right)\left(\ff{O\setminus A}-\ff{O\cup A}-\epsi' \ff{O}\right)+\frac{\lfloor t\ell_2 \rfloor}{\frac{\ell_2}{1-\epsi'}+1}\left(1-\frac{1}{\ell_2}\right)^{\lfloor t\ell_2 \rfloor}\ff{O\cup A}\\
        &\ex{\ff{B_{\ell_2}}}\ge \left(1-\frac{1}{\ell_2}\right)^{\ell_2\lfloor t\ell_2 \rfloor}\ex{\ff{B_{\lfloor t\ell_2 \rfloor}}}+\frac{\ell_2\lfloor t\ell_2 \rfloor}{\frac{\ell_2}{1-\epsi'}+1}\left(1-\frac{1}{\ell_2}\right)^{\ell_2\lfloor t\ell_2 \rfloor}\left[\ff{O}-\left(1-\left(1-\frac{1}{\ell_2}\right)^{\lfloor t\ell_2 \rfloor}\right)\ff{O\cup A}\right]\\
        &\hspace{6em}-\frac{\epsi'\ell_2}{\frac{\ell_2}{1-\epsi'}+1}\left(1-\left(1-\frac{1}{\ell_2}\right)^{\ell_2\lfloor t\ell_2 \rfloor}\right)\ff{O}\\
        &\ge \left(\left(1-\frac{1}{\ell_2}\right)^{(1-t)\ell_2+1}-\left(1-\frac{1}{\ell_2}\right)^{\ell_2}\right)\left(\ff{O}-\ff{O\cap A}-\ff{O\cup A}-\epsi'\ff{O}\right)\\
        &\hspace{2em}+\frac{(1-t)\ell_2}{\frac{\ell_2}{1-\epsi'}+1}\left(1-\frac{1}{\ell_2}\right)^{(1-t)\ell_2+1}\ff{O}
        -\frac{\ell_2}{\frac{\ell_2}{1-\epsi'}+1}\left((1-t)\left(1-\frac{1}{\ell_2}\right)^{(1-t)\ell_2}-\left(1-\frac{1}{\ell_2}\right)^{\ell_2}\right)\ff{O\cup A}\\
        &\hspace{2em}-\frac{\epsi'\ell_2}{\frac{\ell_2}{1-\epsi'}+1}\left(1-\left(1-\frac{1}{\ell_2}\right)^{(1-t)\ell_2+1}\right)\ff{O}\tag{$t\ell_2-1< \lfloor t\ell_2\rfloor \le t\ell_2$}\\
        &\ge \left(\left(1-\frac{1}{\ell_2}\right)^2e^{t-1}-e^{-1}\right)\left(\ff{O}-\ff{O\cap A}-\ff{O\cup A}-\epsi'\ff{O}\right)\\
        &\hspace{2em}+\frac{(1-t)(\ell_2-1)}{\frac{\ell_2}{1-\epsi'}+1}\left(1-\frac{1}{\ell_2}\right)e^{t-1}\ff{O}-\frac{\ell_2}{\frac{\ell_2}{1-\epsi'}+1}\left((1-t)e^{t-1}-\left(1-\frac{1}{\ell_2}\right)e^{-1}\right)\ff{O\cup A}\\
        &\hspace{2em}-\frac{\epsi'\ell_2}{\frac{\ell_2}{1-\epsi'}+1}\left(1-\left(1-\frac{1}{\ell_2}\right)^2e^{t-1}\right)\ff{O}
        \tag{$\ff{O\cap A}+\ff{O\cup A} \ge 0.837\ff{O}$; Lemma~\ref{lemma:val-inq}}\\
        &\ge \left[\left(\frac{(1-t)(\ell_2-1)}{\frac{\ell_2}{1-\epsi'}+1}+1-\frac{1}{\ell_2}\right)\left(1-\frac{1}{\ell_2}\right)e^{t-1}-e^{-1}-\epsi'\left(1-e^{-1}\right)\right]\ff{O}\\
        &-\left(e^{t-1}-e^{-1}\right)\ff{O\cap A}-\left((2-t) e^{t-1}-\left(2-\frac{1}{\ell_2}\right)e^{-1}\right)\ff{O\cup A}\\
        &\ge \left[\left(\frac{(1-t)(\ell_2-1)}{\frac{\ell_2}{1-\epsi'}+1}+1-\frac{1}{\ell_2}\right)\left(1-\frac{1}{\ell_2}\right)e^{t-1}-e^{-1}-\epsi'\left(1-e^{-1}\right)\right.\\
        &\left.-0.377\left(e^{t-1}-e^{-1}\right) -0.46\left((2-t) e^{t-1}-\left(2-\frac{1}{\ell_2}\right)e^{-1}\right)\right] \ff{O} \tag{$\ff{O\cap A} < 0.377\ff{O}$; $\ff{O\cup A} < 0.46\ff{O}$}\\
        &\ge (0.377-\epsi)\ff{O}\tag{$t=0.3;\ell_2 = \frac{5}{2\epsi'};\epsi'=\frac{\epsi}{2}$}
    \end{align*}
\end{proof}

\section{Experiments} \label{exp-appendix}

\paragraph{Experimental setup} We run all experiments on an Intel Xeon(R) W5-2445 CPU at 3.10 GHz with $20$ cores, $64$ GB of memory, and one NVIDIA RTX A4000 with $16$ GB of memory. For Maximum Cut experiments, we use the standard multiprocessing provided in Python, which takes about $20$ minutes to complete, while the video summarization finishes in under a minute.

\subsection{Additional tables and plots}
In this section, you can find the tables and plots omitted in the main paper due to space constraints. In Figure \ref{fig:frames}, we compare the frames selected by \fls+\rg and \sg, and in Figure \ref{max_cut_additional_results}, we report the results for Barab{'a}si-Albert  and Watts-Strogatz models for Maximum Cut.

\subsection{Problem Formulation}
In this section, we formally introduce video summarization and Maximum Cut.

\subsubsection{Video summarization}
Formally, given $n$ frames from a video, we present each frame by a $p$-dimensional
vector. Let $X \in \mathcal{R}^{n \times n}$  be the Gramian matrix of the $ n$ resulting vectors so $X_{ij}$ quantifies the similarity between two vectors through their inner product. The Determinantal Point Process (DPP) objective function is defined by the determinant function $f : 2^{n} \rightarrow \mathcal{R}: f(S) = \text{log}(\text{det} ( X_S)+1)$, where $X_S$ is the principal submatrix of $X $ indexed by $ S$ following  \citet{banihashem2023dynamic} to make the objective function $f$ a non-monotone non-negative submodular function. 



\begin{figure}[h] 
  \subfigure[\fls+\grg] { 
    \includegraphics[width=0.50\textwidth]{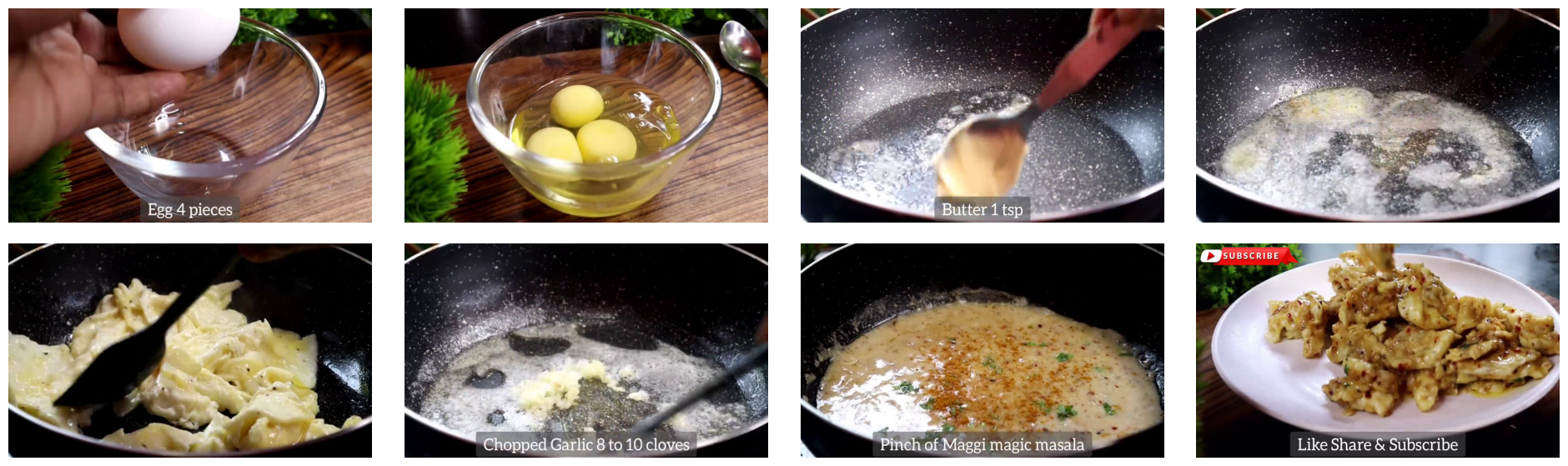}
  }
  \subfigure[\textsc{StandardGreedy}] { 
    \includegraphics[width=0.50\textwidth]{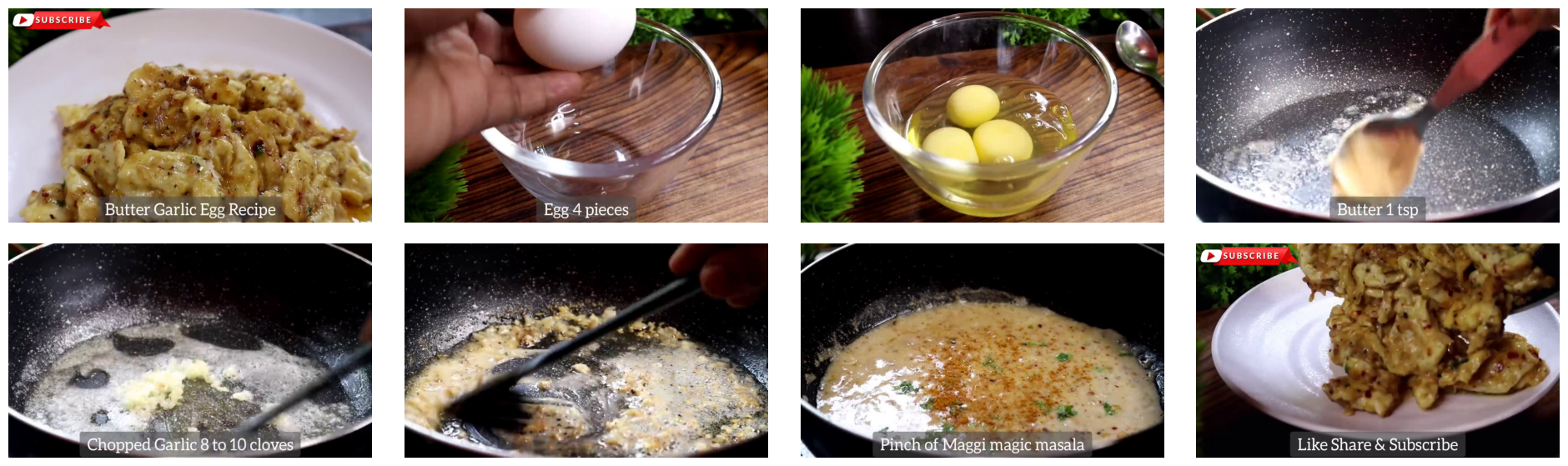}
  }
  \vspace*{-1em}
  \caption{Frames selected for Video Summarization}
  \vspace*{-1em}
  \label{fig:frames}
\end{figure}


\subsubsection{Maximum cut}
Given an undirected graph $G(V, E)$, where $V$ represents the set of vertices, $E$ denotes the set of edges and weights $w(u,v)$ on the edges $(u,v)\in E$, the goal of the Maximum Cut problem is to find a subset of nodes $S \subseteq V$ that maximizes the objective function, $f(S)= \sum_{u \in S, v \in V\setminus S} w(u, v)$. 
\begin{figure}[h] 
   \subfigure[BA, solution value] { 
    \includegraphics[width=0.50\textwidth]{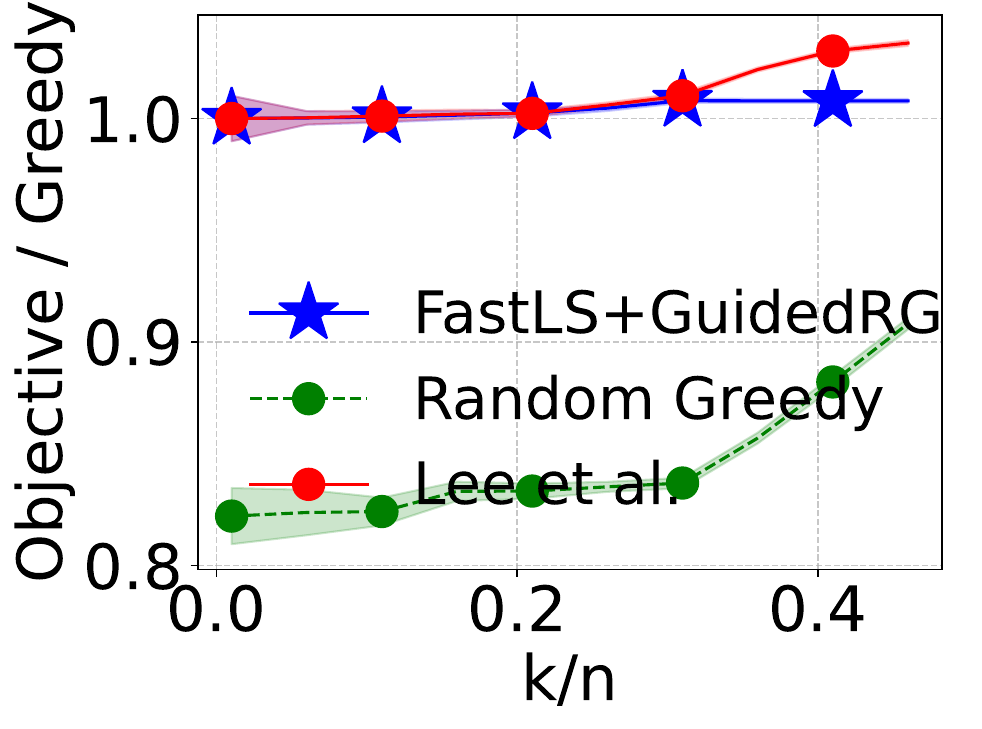}
  }
  \subfigure[BA, queries] { 
    \includegraphics[width=0.50\textwidth]{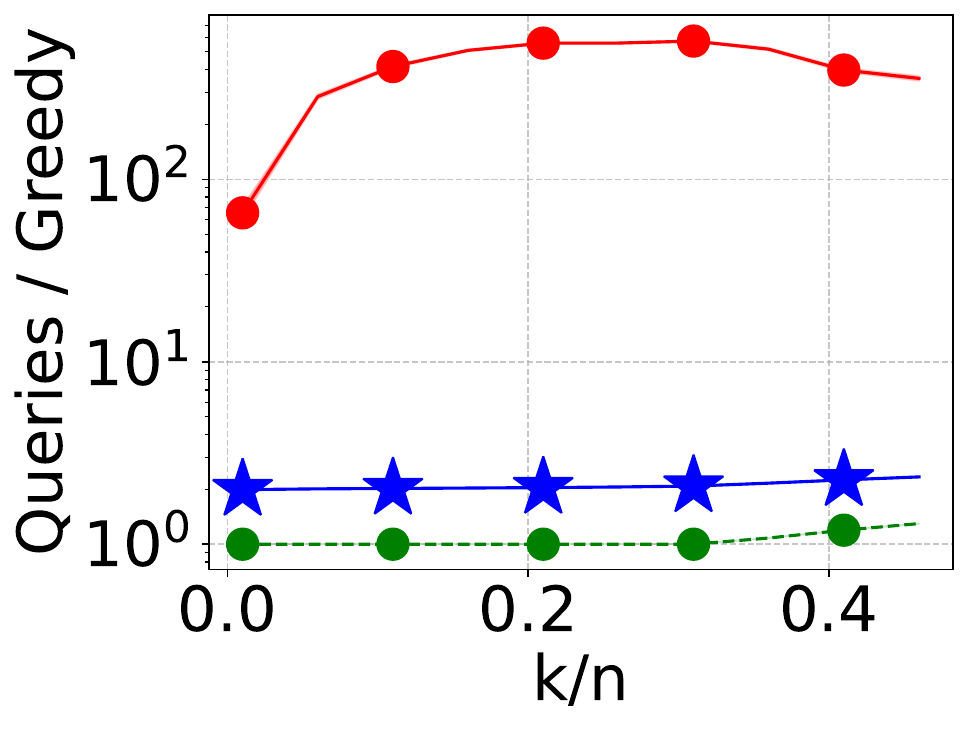}
  }
  \subfigure[Watts Strogatz, solution value] { 
    \includegraphics[width=0.50\textwidth]{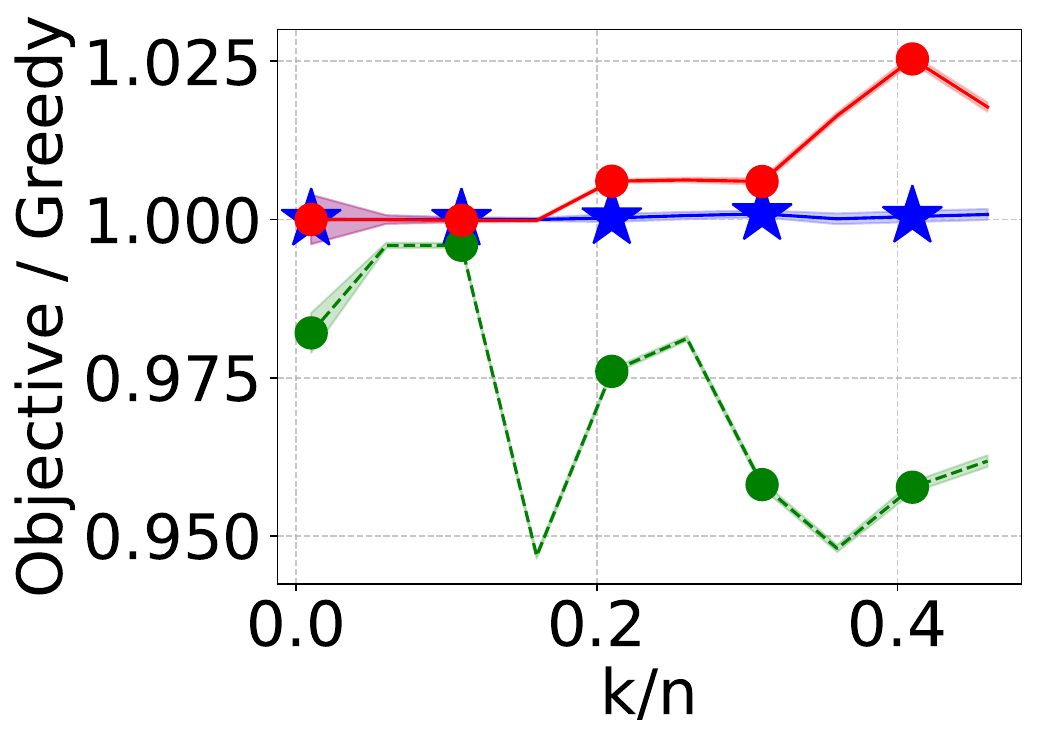}
  }
  \subfigure[Watts Strogatz, queries] { 
    \includegraphics[width=0.50\textwidth]{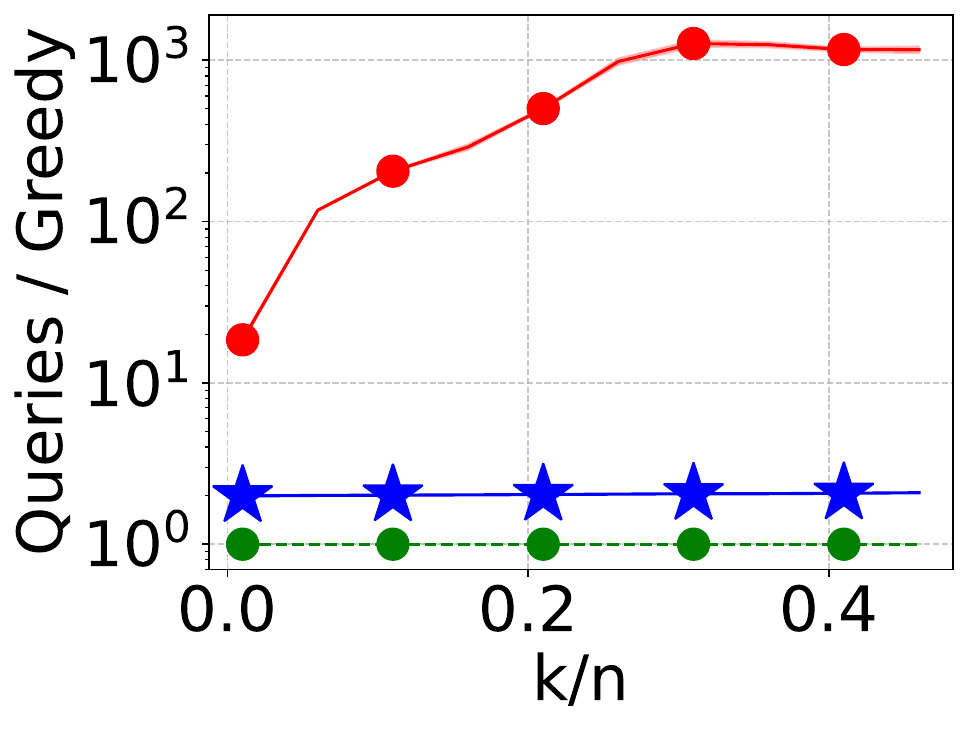}
  }
  \caption{The objective value (higher is better) and the number of queries (lower is better) are normalized by those of \textsc{StandardGreedy}. Our algorithm (blue star) outperforms every baseline on at least one of these two metrics. }
  \label{max_cut_additional_results}
\end{figure}




\subsection{Hyperparameters}

For all experiments, we set the error rate, $\epsilon$, to $0.01$ for \fls+\grg and to $0.1$ for \citet{DBLP:conf/stoc/LeeMNS09}. Additionally, for video summarization, we run \rg $20$ times and report the standard deviation of these runs. For all other experiments, we run the algorithms once per instance and report the standard deviation over instances.
\subsection{Datasets}
The video we select \citep{CookingShowWithRoy} (available under CC BY license) for video summarization lasts for roughly $4$ minutes, and we uniformly sample $100$ frames from the video to form the ground set. For maximum cut, we run experiments on synthetic random graphs, each distribution consisting of $20$ graphs of size $10,000$ vertices generated using the Erd{\H{o}}s-Renyi (ER), Barabasi-Albert (BA), and Watts-Strogatz (WS) models. The ER graphs are generated with $p = 0.001$, while the WS graphs are created with $p = 0.001$ and $10$ edges per node. For the BA model, graphs are generated by adding $m = 2$ edges in each iteration. Data and code are provided in the supplementary material to regenerate the empirical results provided in the paper.

\subsection{Implementation of $\fls$+$\grg$}
For our implementation of \fls,
we take the solution of \sg\ as our initial solution $Z_0$; the theoretical guarantee is thus $f(Z_0) > \opt / k$, since $Z_0$
has higher $f$-value than the maximum singleton. This
increases the theoretical query complexity of our algorithm as implemented to $\oh{\frac{kn}{\epsi}\log(\frac{k}{\epsi })}$. 
Then, for each swap, we find the best candidates to remove from and add to the current solution set (including the dummy element), rather
than any pair that satifies the criterion. For guided \grg, we implement exactly as in the pseudocode (Alg. \ref{alg:grg}) -- we remark that we could instead use (a guided version of)
the linear-time variant of \rg{} \citep{buchbinder2017comparing} to reduce the empirical number of queries further, but for simplicity we did not do this in our evaluation. 


\end{document}